\documentclass[10pt]{article}
\usepackage[a4paper,left=1in, top=1in, right=1in, bottom=1in]{geometry}

\usepackage{amsfonts}
\usepackage{amssymb}
\usepackage{url}
\usepackage{booktabs}
\usepackage{graphicx}
\graphicspath{{Figures/}}
\usepackage{amsthm,bm}
\usepackage{pbox}
\usepackage{mathrsfs}
\usepackage{enumitem}
\usepackage{subfigure}
\usepackage{color, colortbl}

\usepackage{arcs}
\usepackage{tipa}

\usepackage{framed}
\usepackage[framemethod=tikz]{mdframed}

\usepackage[many]{tcolorbox}
\tcolorboxenvironment{theorem}{
  breakable,
  colback=gray!20,
  boxrule=0pt,
  boxsep=1pt,
  left=1em,right=1em,top=1em,bottom=1em,
  arc=10pt,
  before skip=\topsep,
  after skip=\topsep,
}

\definecolor{Gray}{gray}{0.9}
\definecolor{LightCyan}{rgb}{0.88,1,1}
\newcolumntype{g}{>{\columncolor{Gray}}c}

\usepackage[hidelinks]{hyperref}
\hypersetup{
	colorlinks, 
	linkcolor=black, 
	anchorcolor=blue, 
	citecolor=blue
}

\usepackage{cite}

\newtheorem{theorem}{Theorem}
\newtheorem{corollary}{Corollary}
\newtheorem{lemma}{Lemma}

\newtheorem{definition}{Definition}
\newtheorem{assumption}{Assumption}
\newtheorem{proposition}{Proposition}

\usepackage{mathtools, cuted}

\usepackage[lined,ruled]{algorithm2e}
\usepackage{algorithmic}

\usepackage{bbm}
\usepackage{dsfont}

\usepackage{accents}

\begin{document}

\title{Online Combinatorial Auctions for Resource Allocation with Supply Costs and Capacity Limits\thanks{This paper was published by IEEE Journal on Selected Areas in Communications, vol. 38, no. 4, pp. 655-668, April 2020. A preliminary version of this paper was presented at ACM MAMA Workshop -- in conjunction with SIGMETRICS 2019. (\textit{Corresponding Author: Xiaoqi Tan})}
}

\author{Xiaoqi Tan\thanks{Department of Computing Science, University of Alberta. Email: {\tt xiaoqi.tan@ualberta.ca.} This work was done while the author was with the University of Toronto. } 
\and Alberto Leon-Garcia\thanks{University of Toronto. Email: {\tt alberto.leongarcia@utoronto.ca}} 
\and Yuan Wu\thanks{University of Macau. Email: {\tt yuanwu@um.edu.mo}} 
\and Danny H.K. Tsang\thanks{Hong Kong University of Science and Technology. Email: {\tt eetsang@ust.hk}} 
}

\date{}

\begin{titlepage}
\maketitle
\thispagestyle{empty}

\begin{abstract}
We study a general online combinatorial auction problem in algorithmic mechanism design. A provider allocates multiple types of capacity-limited resources to customers that arrive in a sequential and arbitrary manner. Each customer has a private valuation function on bundles of resources that she can purchase (e.g., a combination of different resources such as CPU and RAM in cloud computing). The provider charges payment from customers who purchase a bundle of resources and incurs an increasing supply cost with respect to the totality of resources allocated. The goal is to maximize the social welfare, namely, the total valuation of customers for their purchased bundles, minus the total supply cost of the provider for all the resources that have been allocated. We adopt the competitive analysis framework and provide posted-price mechanisms with optimal competitive ratios.  Our pricing mechanism is \textit{optimal} in the sense that no other online algorithms can achieve a better competitive ratio. We validate the theoretic results via empirical studies of online resource allocation in cloud computing. Our numerical results demonstrate that the proposed pricing mechanism is competitive and robust against system uncertainties and outperforms existing benchmarks.
\end{abstract}

\end{titlepage}

\section{Introduction}
Many auction problems involve allocation of distinct types of resources concurrently. For example, customers in auction-based cloud computing platforms can bid on virtual machines or containers with a package of resources such as CPU and RAM. In these problems, customers often have preferences for bundles or combinations of  different items, instead of a single one \cite{CA_survey}. For this reason, pricing and allocating resources to customers with combinatorial preferences or valuations, termed as combinatorial auctions (CAs) \cite{AMD2001} \cite{CA_PNAS}, play a critical role in enhancing economic efficiency. This is also considered a hard-core problem in algorithmic mechanism design \cite{AMD2001}. 

In this paper, we study an online version of CAs for resource allocation with supply costs and capacity limits. A single provider allocates multiple types of capacity-limited resources to customers that  arrive in a sequential and arbitrary manner. Each customer has a valuation function on possible bundles of resources that she wants to purchase.  
The provider charges payment from customers who purchase a bundle of resources and incurs an increasing marginal supply cost (i.e., the derivative of the supply cost function) per unit of consumed resource.  The goal is to  maximize the social welfare, namely, the total valuation of customers for their purchased bundles, minus the supply cost of the provider for all the resources allocated. 

When online CAs are subject to increasing supply costs and capacity limits, a fundamental challenge is how to properly price the resources in the absence of future information. Specifically, if the resources are sold too cheaply (i.e., too aggressive), then an excessive portion of them may be purchased by earlier customers with low valuations. This will increase the total cost for the provider and thus the price, which consequently prevents later customers from purchasing the resources even if their valuations are higher than the earlier ones. On the other hand, if the price is set too high (i.e., too conservative), then the provider may lose customers, leading to poor performance as well. This paper tackles this challenge by proposing pricing  mechanisms that achieve an optimal balance between aggressiveness and conservativeness without future information, leading to  the best-achievable competitive ratios under arbitrary increasing marginal cost functions.

Our results are applicable to a variety of resource allocation problems in the emerging paradigms of networking and computing systems.  For example, for auction-based resource allocation in infrastructure-as-a-service clouds,  providers can charge their users with a certain payment mechanism while also paying a considerable amount of energy costs to maintain the computing servers \cite{data_center_power}.  Another example is 5G network slicing, one of the key elements of 5G communications \cite{network_slicing}. The ultimate goal of network slicing is to dynamically package different types of network resources (e.g., the base stations and the spectrum channels) for different customers. Here,  the network operator needs to consider the cost for providing these resources. In this regard, the model studied in this paper offers a promising option to address such resource allocation problems in 5G network slicing.

\subsection{Related Work}
Online CAs without supply costs, which is essentially an online set-packing problem \cite{CA_survey},  has been widely studied, including online auctions \cite{Bartal2003}, \cite{Buchbinder2015}, online  matching \cite{b_matching} \cite{OM_concave}, AdWords problems \cite{Adword2009}, \cite{Mehta2013},  online covering  and packing problems \cite{covering_packing}, \cite{Buchbinder2009}, and online knapsack problems \cite{knapsack}. Among them, the authors of \cite{Bartal2003} studied an online CA problem and proposed an  $ O(\log(v_{\max}/v_{\min})) $-competitive online algorithm when there are $\Omega(\log(v_{\max}/v_{\min}))$ copies of each item and each customer's  valuation is assumed to be within the interval of $ [v_{\min},v_{\max} ] $.  Similar results have also been reported for online knapsack problems \cite{knapsack}. By assuming that the weight of each item is much smaller than the capacity of the knapsack, and that the value-to-weight ratio of each item is bounded within the interval of $ [L,U] $, the authors of \cite{knapsack} proposed  an algorithm which is $ (1+\ln(U/L)) $-competitive. 

One of the common assumptions made in the above papers is that the resources  can be allocated without incurring an increasing supply cost for the provider. Although this assumption is reasonable for the allocation of digital goods \cite{Mehta2013}, it may not hold for most paradigms of network resource managements, where the production cost or the operational cost is an increasing function of the allocated resources. Motivated by this, Blum et al. \cite{Blum2011} pioneered the study of online CAs with an increasing production cost.  In this setting, the provider can produce any number of copies of the items being sold (i.e., without capacity limit), but needs to pay an increasing marginal production cost per copy. Blum et al. proposed a pricing scheme called \textit{twice-the-index} for several reasonable marginal production cost functions such as linear, lower-degree polynomial and logarithmic functions. For each of these functions, a constant competitive ratio was derived. Huang et al. \cite{Huang2018} later studied a similar problem and achieved a tighter competitive ratio with a unified pricing framework. In particular, for power cost functions, they proved that the optimal competitive ratio can be achieved when there is no capacity limit. In contrast to \cite{Blum2011} and \cite{Huang2018}, in this work we prove that in the capacity-limited case, direct application of the pricing schemes designed in \cite{Blum2011} and \cite{Huang2018} is suboptimal, and a tighter (and optimal) competitive ratio can be achieved by our newly proposed pricing schemes.

\subsection{Major Contributions}
We develop an optimal posted-price mechanism (PPM), dubbed $\text{PPM}_\phi$, for online CAs with supply costs and capacity limits. $\text{PPM}_\phi$ is  optimal in the sense that no other online algorithms can achieve a tighter/better competitive ratio. One of the key elements in $\text{PPM}_\phi$ is a strategically-designed pricing function $ \phi $ that determines the selling price based on the current resource utilization levels only. In the general case where the supply cost function is convex and differentiable, we prove that the necessary and sufficient conditions for $ \text{PPM}_\phi $ to be competitive are related to the existence of an increasing pricing function $ \phi $ for a group of first-order two-point boundary value problems (BVPs) in the field of ordinary differential equations (ODEs) \cite{Perko2001, ODE1973}. We derive structural results based on these BVPs that lead to a fundamental characterization of the optimal competitive ratios and the optimal pricing functions. To validate our structural results,  we perform a case study when the supply cost function is a power function (e.g., $ f(y) = ay^s $),  which is an important case  widely exploited \cite{Blum2011, Huang2018, Adam_speed_scaling, speed_scaling}, and show that both the optimal competitive ratios and the corresponding optimal pricing functions can be characterized in analytical forms with some low-complexity numerical computations. Our optimal analytical results for the power cost function improve or generalize the  results in several previous studies, e.g., \cite{Blum2011}, \cite{Huang2018}, \cite{XZhang_ToN}, \cite{bo2018}. Moreover, our structural results can also be extended to general settings of online resource allocation with heterogeneous supply costs and multiple time slots. 

\section{The Basic Resource Allocation Model}
This section presents the basic model, the technical assumptions and the definition of competitive ratios for online resource allocation with supply costs and capacity limits. 

\subsection{The Basic Model}
\label{problem_setup}
We consider a single provider who allocates a set $ \mathcal{K}=\{1,\cdots,K\} $ of $ K $ types of resources to its customers. Each type of resource $ k\in\mathcal{K} $ is associated with a cost function $ f_k(y) $, where $ f_k(y) $ denotes the total supply cost of providing $ y $ units of resource $ k $. For example, if resource $ k $ represents the computing cycles in cloud/fog/edge computing \cite{fog_computing}, then the supply cost $ f_k(y) $ can represent the electricity cost of maintaining the computing servers. In the following we will also frequently use the derivative of $ f_k $, i.e., the marginal cost function $ f_k' $. \textit{For simplicity of exposition, we assume that the cost functions are identical for all types of resources, and thus we drop the index $ k $ and simply use $ f $ to denote the supply cost function of all resource types}. Our results are applicable to general cases with heterogeneous cost functions, and we will provide our general results in Section \ref{extension}. 

We consider an online setting where customers arrive one at a time in some arbitrary manner. In particular, for a set of customers $ \mathcal{N}=\{1,2,\cdots,N\} $, we denote the arrival time of customer $ n $ by $ t_n $. Meanwhile, we assume without loss of generality that $ t_1\leq t_2\leq \cdots\leq t_N $, where ties are broken arbitrarily if multiple customers arrive simultaneously. Each customer $ n $ wants to get a bundle of resources $ b\in\mathcal{B} $ based on their own preferences, where $ \mathcal{B} $ denotes the set of all the possible bundles (including the empty bundle $ \emptyset $). A bundle $ b $ of resources is denoted by the vector $ (r_1^b,\cdots,r_K^b) $, where $ r_k^b $ denotes the number of units for resource $ k\in\mathcal{K} $.  We consider the case of limited-supply, and normalize the capacity limit to be 1 for each resource type. Therefore,  $ r_k^b $ is also normalized to be the proportion of the capacity limit accordingly. Each customer $ n $ has a private valuation function $ v_n: \mathcal{B}\rightarrow \mathbb{R} $, where $ v_n(b) $  denotes the valuation of  customer $ n $  for getting bundle $ b\in\mathcal{B} $. For simplicity of notations, we denote the valuation by $ v_n^b = v_n(b) $  if customer $ n $ gets bundle $ b\in\mathcal{B} $. In the following we may use $ v_n^b $ and $ v_n(b) $ interchangeably.  We do not make any assumption regarding the valuation functions (except that $ v_n(\emptyset) = 0 $, i.e., valuation of the empty bundle is zero). 

In the standard setting of online CAs, the provider does not have any information about the customers. Upon the arrival of each customer $ n\in\mathcal{N} $, the customer reports a valuation function $ \hat{v}_n $ to the provider. The valuation function $ \hat{v}_n $ may or may not be the true valuation of customer $ n $ (i.e., customers may  strategically  manipulate their  bids). 
The provider collects the valuation function $ \hat{v}_n $ from customer $ n $ and decides an irrevocable decision about whether to accept this customer or not. The provider will wait for the next customer $ n+1 $  if customer $ n $ is rejected (or customer $ n $ gets an empty bundle $ \emptyset $). Otherwise, the provider needs to determine the payment $ \pi_n $ to be collected from customer $ n $ based on the  known information (including current valuation function $ \hat{v}_n $ and all the previous valuation functions before customer $ n $), 
and then allocates a bundle  $ b_n\in\mathcal{B} $ of resources to customer $ n $. The resulting \textit{payment rule}  (i.e., the determination of $ \{\pi_n\}_{\forall n} $)  and the \textit{allocation rule} (i.e., the determination of $ \{b_n\}_{\forall n} $) constitute an  \textit{online mechanism}. An important economic objective in mechanism design is \textit{incentive compatibility}. Specifically, a mechanism is incentive compatible or truthful if each customer maximizes its own quasilinear utility, i.e., $ v_n(b_n) - \pi_n $, by reporting the true valuation function, namely,  $ \hat{v}_n = v_n $. 

The objective is to design the payment rule  to incentivize customers to truthfully report their valuation functions, and the allocation rule
to maximize the social welfare $ \sum_{n\in\mathcal{N}} v_n(b_n) - \sum_{k\in\mathcal{K}}f(y_k) $. 

\subsection{Assumptions}\label{assumptions}
We make the following assumptions throughout the paper. 

\begin{assumption}\label{assumption_cost_function}
	The cost function 
	$ f(y) $ is differentiable and strictly-convex in $ y\in [0,1] $ and  $ f(0) = 0 $. 
\end{assumption}

If we denote the set of all differentiable and strictly-convex cost functions with $ f(0) = 0 $ by $ \mathcal{F} $, then \textit{Assumption \ref{assumption_cost_function} states that we only focus on the cases when $ f\in\mathcal{F} $}. In the following we will frequently use the minimum and maximum marginal costs defined as follows:
\begin{equation}
\underline{c} \triangleq f'(0), \overline{c} \triangleq f'(1).
\end{equation}
Intuitively, if $ f $ is known to the provider, then $ \underline{c} $ and $ \overline{c} $ are known to the provider as well. Note that a given cost function $ f\in  \mathcal{F} $ always has a strictly-increasing marginal cost $ f' $. 

\begin{assumption}\label{assumption_small}
	For each resource type $ k\in\mathcal{K} $, the number of units in each bundle $ b\in\mathcal{B} $ is much smaller than the total capacity limit, i.e.,  $ r_k^b\ll 1$.
\end{assumption}

Assumption \ref{assumption_small} states that allocating a bundle of resources to a single customer does not substantially influence the overall system and market (i.e., each customer's demand is very small), and thus allows us to focus on the online nature of the problem with mathematical convenience. In large-scale systems (e.g., when $ N $ is large), Assumption \ref{assumption_small} naturally holds. 

\begin{assumption}\label{assumption_bound}
	The \textit{per-unit-valuation} (PUV) of all customers, defined as $ v_n^b/r_k^b$,  is upper bounded by $ \overline{p} $, namely,
	\begin{equation}
	\max_{ n\in\mathcal{N},b\in\mathcal{B},k\in\mathcal{K}, r_k^b\neq 0}\ \left\{v_n^b/r_k^b\right\} \leq \overline{p}.
	\end{equation}
\end{assumption}

We will refer to $ \overline{p} $ as the \textit{upper bound} hereinafter.   Since $ r_k^b $ is finite, Assumption \ref{assumption_bound} states that the outputs of the value function $ v_n(\cdot) $ are upper bounded, and thus it helps to eliminate those irrational cases with extremely-high valuations. Alternatively, $ \overline{p} $ can be interpreted as the maximum price customers are willing to pay for purchasing a single unit of resource. Throughout the paper we also assume $ \overline{p}>\underline{c} $ in order to ensure that the problem setup is interesting. Otherwise, no resources will be allocated.

\subsection{Competitive Analysis}
We categorize all the parameters defined previously into the following two groups:
\begin{enumerate}
	\item The \textit{Setup} $\mathcal{S}$: all the parameters known at the beginning,  including the cost function $ f\in\mathcal{F} $, the upper bound $ \overline{p} $, the set of resource types $ \mathcal{K} $, and the set of bundles $ \mathcal{B} $.   
	
	\item The \textit{Arrival Instance} $ \mathcal{A} $: all the parameters revealed over time, including the set of customers $ \mathcal{N} $, their arrival times $ \{t_n\}_{\forall n\in\mathcal{N}} $, and the valuation functions $ \{v_n(\cdot)\}_{\forall n\in\mathcal{N}} $. 
\end{enumerate}

An arrival instance $ \mathcal{A} $ consists of all the information in the customer side that is not known to the provider a priori.  In the offline setting when we assume a complete knowledge of $ \mathcal{A} $,  the optimal social welfare $ W_{\textsf{opt}}(\mathcal{A}) $ can be obtained by solving the following mixed-integer  program:
\begin{subequations}\label{SWM}
	\begin{alignat}{2}
	W_{\textsf{opt}}(\mathcal{A}) =\   
	& \underset{\mathbf{x,y}}{\mathrm{maximize}} &  &   \sum_{n\in\mathcal{N}}\sum_{b\in\mathcal{B}} v_n^b x_n^b - \sum_{k\in\mathcal{K}} f\left(y_k\right), \\
	&\mathrm{subject\ to}\ & & \sum_{n\in\mathcal{N}}\sum_{b\in\mathcal{B}} r_k^b x_n^b = y_k, \forall k, \label{total_utilization_original}\\ 
	&    & & \sum_{b\in\mathcal{B}} x_n^b \leq 1,\forall n,\label{at_most_one_bundle}\\
	&    & & 0 \leq y_k \leq 1, \forall k,\label{capacity_limits}\\
	&    & & x_n^b\in \{0,1\},  \forall n,b,
	\end{alignat}
\end{subequations}
where $ x_n^b\in \{0,1\} $ is a binary variable that represents the status of bundle $ b $ for customer $ n $, and $ y_k $ denotes the total resource consumption of resource type $ k $ in the end. In particular, $ x_n^b = 1 $ means that bundle $ b $ is allocated to customer $ n $, and $ x_n^b = 0 $ otherwise. It is possible that $ x_n^b = 0 $ for all $ b\in\mathcal{B} $, meaning that customer $ n $ will leave without making any purchase. Constraint \eqref{at_most_one_bundle} indicates that at most one bundle will be allocated to each customer. Constraint \eqref{capacity_limits} denotes the normalized capacity limit for resource type $ k\in\mathcal{K} $. 

In the online setting when customers are revealed one-by-one in a sequential manner, the social welfare performance, denoted by $ W_{\textsf{online}}(\mathcal{A}) $, can be quantified via the standard competitive analysis framework \cite{Borodin1998}. Specifically, an online mechanism is $ \alpha $-competitive if
\begin{equation}\label{def_cr}
W_{\textsf{online}}(\mathcal{A}) \geq \frac{1}{\alpha}W_{\textsf{opt}}(\mathcal{A}) 
\end{equation}   
holds for all possible arrival instances $ \mathcal{A} $, where $ \alpha\geq 1 $. Our target is to design an online mechanism such that $ W_{\textsf{online}} $ is as close to $ W_{\textsf{opt}} $ as possible, i.e., $\alpha$ is as close to 1 as possible.  

\begin{algorithm}
	\caption{PPM with Pricing Function $ \phi $ ($ \text{PPM}_{\phi} $) }\label{online_mechanism}	
	\begin{algorithmic}[1]
		\STATE \textbf{Input:} Setup $ \mathcal{S} $ and $ \phi$, and initialize $ y_k^{(0)} = 0, \forall k$.   
		
		\WHILE{a new customer $ n $ arrives}
		\STATE Offer resource $ k\in\mathcal{K} $ at price $ p_k^{(n)}$ as follows: \label{publish_price}
		\begin{align}\label{pricing_function}
			p_k^{(n)} = 
			\phi(y_k^{(n-1)}). 
		\end{align}
		
		\STATE Customer chooses the utility-maximizing bundle\\ $ b_* $ 
		by solving the following problem:\label{decision_making}
		\begin{equation}\label{utility_maximizing}
			b_* = \arg\max_{b\in\mathcal{B}}\ v_n^b - \sum\nolimits_{k\in\mathcal{K}} p_k^{(n)} r_k^b,
		\end{equation}
		where $ r_k^b $ denotes the units of resource $ k $ in bundle \\ $ b $,  
		and then calculates the potential payment 
		\begin{equation}\label{payment}
			\pi_n =  \sum\nolimits_{k\in\mathcal{K}} p_k^{(n)}r_k^{b_*}.
		\end{equation}
		
		\IF{$v_n^{b_*} - \pi_n <  0$ or $ y_k^{(n-1)}  + r_k^{b_*} > 1 $ holds for any $ k\in\mathcal{K} $}
		\STATE Customer $ n $  leaves without purchasing anything (i.e., $ x_n^b = 0 $ for all $ b\in\mathcal{B} $).	
		\ELSE
		\STATE Customer $ n $ chooses bundle $ b_* $ and pays  $ \pi_n $
		to the provider (i.e., $ x_n^{b_*} = 1 $ and $ x_n^b = 0 $, $ \forall  b\in \mathcal{B}\backslash\{b_*\} $). \label{payment_transaction}
		\STATE Provider updates the resource consumption by  \label{update}
		\begin{equation}\label{update_resource}
			y_k^{(n)}  = y_k^{(n-1)}  + r_k^{b_*}, \forall k\in\mathcal{K}.
		\end{equation}	
		\ENDIF		
		\ENDWHILE
	\end{algorithmic}
	\label{PM}
\end{algorithm}

\section{PPM and Structural Results}\label{section_design_pricing_function}
In this section, we introduce our proposed online mechanism $ \text{PPM}_\phi  $, and present the necessary and sufficient conditions for  $ \text{PPM}_\phi  $ to be $ \alpha $-competitive. Based on these conditions, we derive structural results to characterize the minimum value of $ \alpha $.  

\subsection{$ \text{PPM}_\phi $: An Overview of How It Works}
We focus on the setting of posted-price \cite{PPM2010}  and propose  $ \text{PPM}_\phi $ in Algorithm \ref{PM}. In posted-price, the provider cannot ask the customers to submit their valuation functions, and thus cannot run Vickrey–Clarke–Groves auctions \cite{AMD2001}. Instead, the provider posts prices upon arrival of each customer $ n\in\mathcal{N} $, and lets customer $ n $ make her own decision on whether to purchase or not based on the posted prices. In this regard,  posted-price is privacy-preserving since it does not require customers to reveal their private valuation functions.  Meanwhile, by virtue of posted-price, our proposed $ \text{PPM}_\phi $ is incentive compatible since false reports naturally vanish \cite{PPM2010}.

In Algorithm \ref{PM}, at each round when there is a new arrival of customer $ n\in\mathcal{N} $, the provider offers her the prices $ \{p_k^{(n)}\}_{\forall  k} $ by Eq. \eqref{pricing_function}, where $ \phi $ is referred to as the \textbf{pricing function}  and $ y_k^{(n-1)} $ denotes the utilization of resource type $ k\in\mathcal{K} $ upon the arrival of customer $ n $, i.e., after processing customer $ n-1 $.  Note that when $ n = 1 $, the posted price  for the first customer is given by  $  p_k^{(1)} = \phi(y_k^{0}) $, where $ y_k^{(0)} $ is the resource utilization before processing the first customer, and thus is initialized to be zero.  Based on the offered prices $ \{p_k^{(n)}\}_{\forall  k} $, customer $ n $ selects the utility-maximizing bundle by solving the problem in Eq. \eqref{utility_maximizing} and calculates the potential payment in Eq. \eqref{payment}. If the maximum utility of customer $ n $, i.e., $ v_n^{b_*} - \pi_n $, is less than zero (i.e., negative utility), or the capacity limit constraint \eqref{capacity_limits} is violated, then customer $ n $ will leave without purchasing anything\footnote{ 
We assume  that customers are rational and will not purchase any bundle if they suffer from negative utilities. This is known as the individual rationality in economics and is 
a common design objective in mechanism design \cite{AMD2001}.} and the provider will wait for the next customer $ n+1 $. Otherwise, customer $ n $ will choose bundle $ b_* $. The provider will charge  this customer the payment $ \pi_n $ and update the total resource utilization level $ y_k $ in Eq. \eqref{update_resource}. The process repeats upon arrival of customer $ n+1 $. 

The above processes show that  the solutions found by $ \text{PPM}_\phi $, namely, $ \{x_n^b\}_{\forall n,b} $  and $ \{y_k^{(n)}\}_{\forall n,k} $, are always feasible to Problem \eqref{SWM}. Another observation is that the pricing function $ \phi $ plays a critical role in $ \text{PPM}_\phi $. Indeed, it is $ \phi $ that determines the posted prices in line \ref{publish_price}, and then influences each customer's decision-making  in line \ref{decision_making}-line \ref{payment_transaction}, which ultimately influences the social welfare achieved by $ \text{PPM}_\phi $, i.e., 
\begin{equation}\label{SW_online}
W_{\textsf{online}}(\mathcal{A}) = \sum_{n\in\mathcal{N}} v_n^{b_*} x_n^{b_*} - \sum_{k\in\mathcal{K}} f\big(y_k^{(N)}\big).
\end{equation}
In Eq. \eqref{SW_online}, $ y_k^{(N)} $ denotes the final resource utilization of resource type $ k\in\mathcal{K} $, and $ x_n^{b_*} $ denotes the status of the utility-maximized bundle $ b_* $ for customer $ n $, i.e., $ x_n^{b_*} = 1 $ denotes that customer $ n $ obtains bundle $ b_* $, and $ x_n^{b_*} = 0 $ otherwise. Note that both $ \{x_n^{b_*}\}_{\forall n} $ and $ \{y_k^{(N)}\}_{\forall k} $ depend on the pricing function $ \phi $, and thus the final competitive ratio of $ \text{PPM}_\phi $ depends on $ \phi $ as well.

\subsection{Conditions for $ \text{PPM}_\phi $ to Be $ \alpha $-Competitive} 
\label{section_theorem_1}
An important result in this paper is the development of the following Theorem \ref{a_unified_BVP}, which characterizes the sufficient and necessary conditions for the pricing function $ \phi $ such that $ \text{PPM}_\phi $ can be $ \alpha $-competitive. 

\begin{theorem} \label{a_unified_BVP}
	Given a setup $ \mathcal{S} $ with $ \overline{p}\in (\underline{c},+\infty) $, we have:
	\begin{itemize}
		\item \underline{\normalfont\textsf{Low-Uncertainty Case (LUC)}: $ \overline{p}\in (\underline{c},\overline{c}] $}. 
		\begin{itemize}
			\item {\normalfont\textbf{Sufficiency}}.  For any given $ \alpha\geq 1 $, if $ \phi(y) $ is a solution to the following first-order BVP:
			\begin{align*}
			\normalfont \textsf{L}(\alpha)
			\begin{cases}
			\phi'(y) =  \alpha\cdot \frac{\phi(y) - f'(y)}{f'^{-1}(\phi(y))}, y\in (0,v), \\
			\phi(0) = \underline{c}, 
			\phi\left(v\right) \geq  \overline{p},
			\end{cases}
			\end{align*}
			where $ v \triangleq f'^{-1}(\overline{p}) $ and $ f'^{-1} $ denotes the inverse of $ f' $, then $ \text{PPM}_{\phi}  $ is $ \alpha $-competitive. 
			
			\item {\normalfont\textbf{Necessity}}.  If there exists an $ \alpha$-competitive online algorithm, then there must exist a strictly-increasing function $ \phi(y) $ that satisfies $ \normalfont \textsf{L}(\alpha) $. 
		\end{itemize}
		
		\item \underline{\normalfont \textsf{High-Uncertainty Case (HUC)}: $ \overline{p}\in (\overline{c},+\infty) $}.
		\begin{itemize}
			\item {\normalfont\textbf{Sufficiency}}. For any given $ \alpha\geq 1 $, if $ \phi(y) $ is a solution to the following two first-order BVPs simultaneously:
			\begin{subequations}\label{BVP_high}
				\begin{align*}
				&\normalfont \textsf{H}_1(u,\alpha)
				\begin{cases}
				\phi'(y) =  \alpha\cdot \frac{\phi(y) - f'(y)}{f'^{-1}(\phi(y))}, y\in (0,u), \\
				\phi(0) = \underline{c}, 
				\phi(u) = \overline{c}.
				\end{cases}
				\\
				&\normalfont \textsf{H}_2(u,\alpha)
				\begin{cases}
				\phi'(y) =  \alpha\cdot \left(\phi(y) - f'(y)\right), y\in (u,1), \\
				\phi(u) = \overline{c}, 
				\phi(1) \geq \overline{p},
				\end{cases}
				\end{align*}
			\end{subequations}
			where $ u\in (0,1) $ is the resource utilization level such that $ \phi(u) = \overline{c} $, 
			then $ \text{PPM}_{\phi}  $ is $ \alpha $-competitive.
			
			\item {\normalfont\textbf{Necessity}}. If there exists an $ \alpha$-competitive online algorithm, then there must exist a resource utilization level $ u\in (0,1) $ and a strictly-increasing function $ \phi(y) $ such that $ \phi(y) $ satisfies $ \normalfont \{\textsf{H}_1(u,\alpha), \normalfont\textsf{H}_2(u,\alpha)\}$.
		\end{itemize} 
	\end{itemize}
\end{theorem} 
\begin{proof}
	The terms \textsf{LUC} and \textsf{HUC} arise from the fact that $ \overline{p} $ indicates the uncertainty level of the PUVs in the arrival instance $ \mathcal{A} $, namely, a larger $ \overline{p} $ implies a wider range of the  PUV distribution (note that the PUVs are randomly distributed within $ [0,\overline{p}]$ based on Assumption \ref{assumption_bound}), and vice versa. We emphasize that the division into cases \textsf{LUC} and \textsf{HUC} is not artificial, but arise from a principled online primal-dual analysis of Problem \eqref{SWM}. 
	The detailed proof is given in Appendix \ref{proof_a_unified_BVP}. 
\end{proof}

Theorem \ref{a_unified_BVP} consists of the conditions that are both sufficient and necessary. 
The sufficiency in  Theorem \ref{a_unified_BVP} argues that $ \text{PPM}_\phi $ is $ \alpha $-competitive as long as the pricing function $ \phi $ is a strictly-increasing solution to the corresponding BVPs in \textsf{LUC} and \textsf{HUC}. Hence, the discussion is within the domain of PPMs. 
The necessity of Theorem \ref{a_unified_BVP} argues that if there exists any $ \alpha $-competitive online algorithm, then there must exist a strictly-increasing solution to the corresponding BVPs. Therefore, the necessity of Theorem \ref{a_unified_BVP} is not restricted to PPMs only, and thus is more general.

(\textbf{Intuition of Theorem \ref{a_unified_BVP}}) In Fig. \ref{three_pricing_schemes}, we illustrate two pricing functions for both \textsf{LUC}  and \textsf{HUC}.  Fig. \ref{three_pricing_schemes}(a) illustrates a special case in \textsf{LUC} when $ \phi(v) = \overline{p} $, where $ v = f'^{-1}(\overline{p}) $  denotes the maximum-possible resource utilization level for $ \text{PPM}_\phi $. Here we use the pricing function illustrated in Fig. \ref{three_pricing_schemes}(a) to briefly explain the intuition behind the BVP of $ \textsf{L}(\alpha) $. The rationality of the two BVPs in \textsf{HUC} follows the same principle. Note that the ODE of $ \textsf{L}(\alpha) $ in Theorem \ref{a_unified_BVP} can be reorganized as
\begin{align}\label{ODE_LUC}
\phi(y) - f'(y) = \frac{1}{\alpha}\phi'(y)f'^{-1}\left(\phi(y)\right),  y\in (0, v).
\end{align}
The left-hand-side of Eq. \eqref{ODE_LUC} is illustrated by the grey area in Fig. \ref{three_pricing_schemes}(a). 
Since $ f(0) = 0 $, $ \phi(0) = \underline{c} $, and $ \phi(v) =  \overline{p} $, integrating both sides of Eq. \eqref{ODE_LUC} for $ y\in [0,v] $ leads to
\begin{align}
\int_{0}^{v}\phi(y)dy - f(v) =  \frac{1}{\alpha} \int_{0}^{v} \phi'(y)f'^{-1}\left(\phi(y)\right) dy = \frac{1}{\alpha} \int_{\underline{c}}^{\overline{p}} f'^{-1}\left(\phi\right) d\phi. \label{ODE_LUC_integration}
\end{align}
Notice that the last integration  in Eq. \eqref{ODE_LUC_integration} is over the inverse of the marginal cost function, which can be solved in analytical form so Eq. \eqref{ODE_LUC_integration} is equivalently written as
\begin{equation}\label{ODE_LUC_inequality}
\alpha = \frac{\overline{p} v - f(v)}{\int_{0}^{v}\phi(y)dy - f(v) }.
\end{equation}
Next we show that Eq. \eqref{ODE_LUC_inequality} essentially captures the worst-case ratio between the optimal offline social welfare and the social welfare achieved by  $ \text{PPM}_\phi $ under a special arrival instance.

\begin{figure}
	\centering
	\subfigure[\textsf{LUC}: $ \overline{p}\in {(\underline{c},\overline{c}]} $]{\includegraphics[width= 4.7 cm]{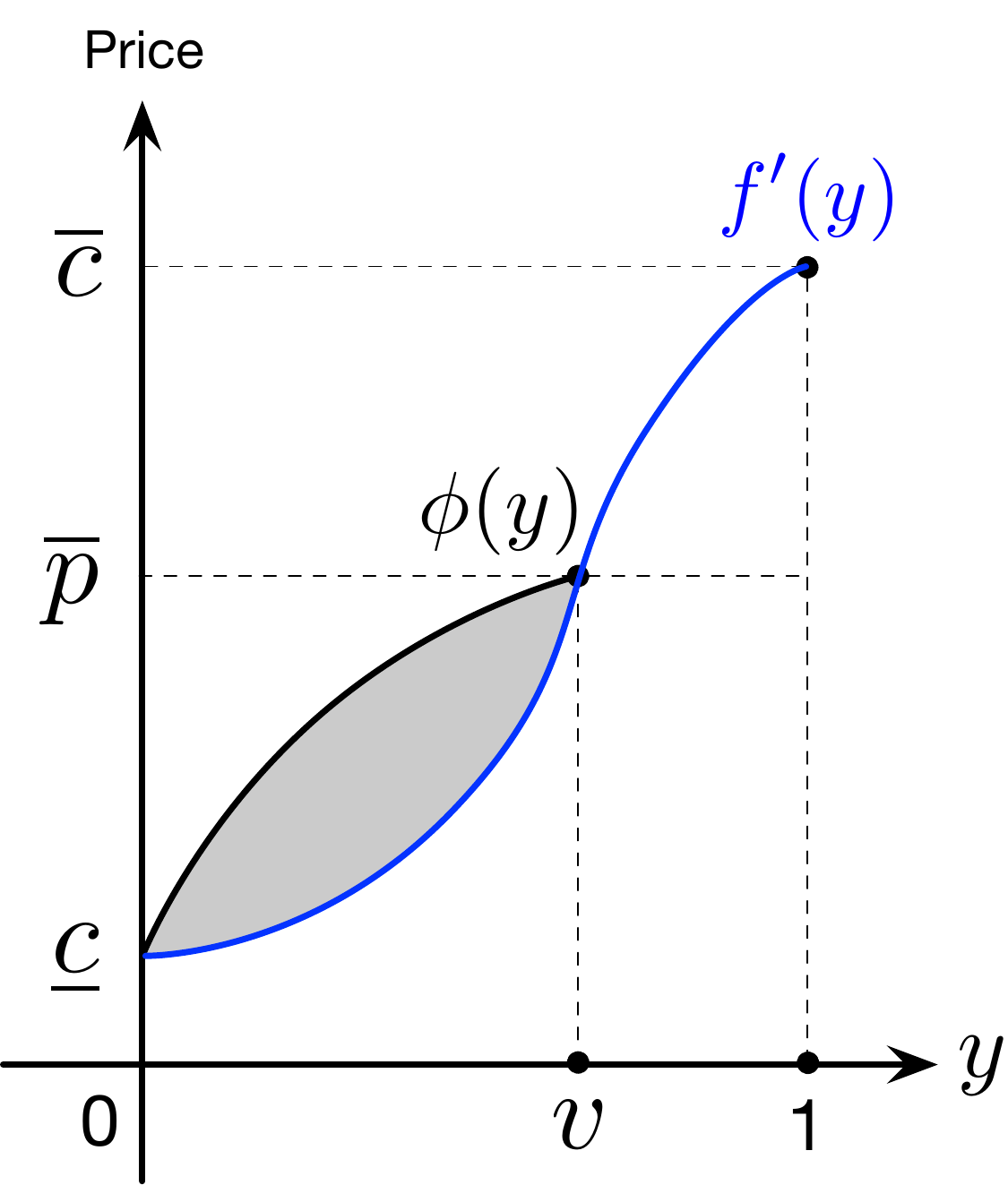}}
	\qquad 
	\subfigure[\textsf{HUC}: $ \overline{p}\in {(\overline{c},+\infty)} $]{\includegraphics[width= 4.7 cm]{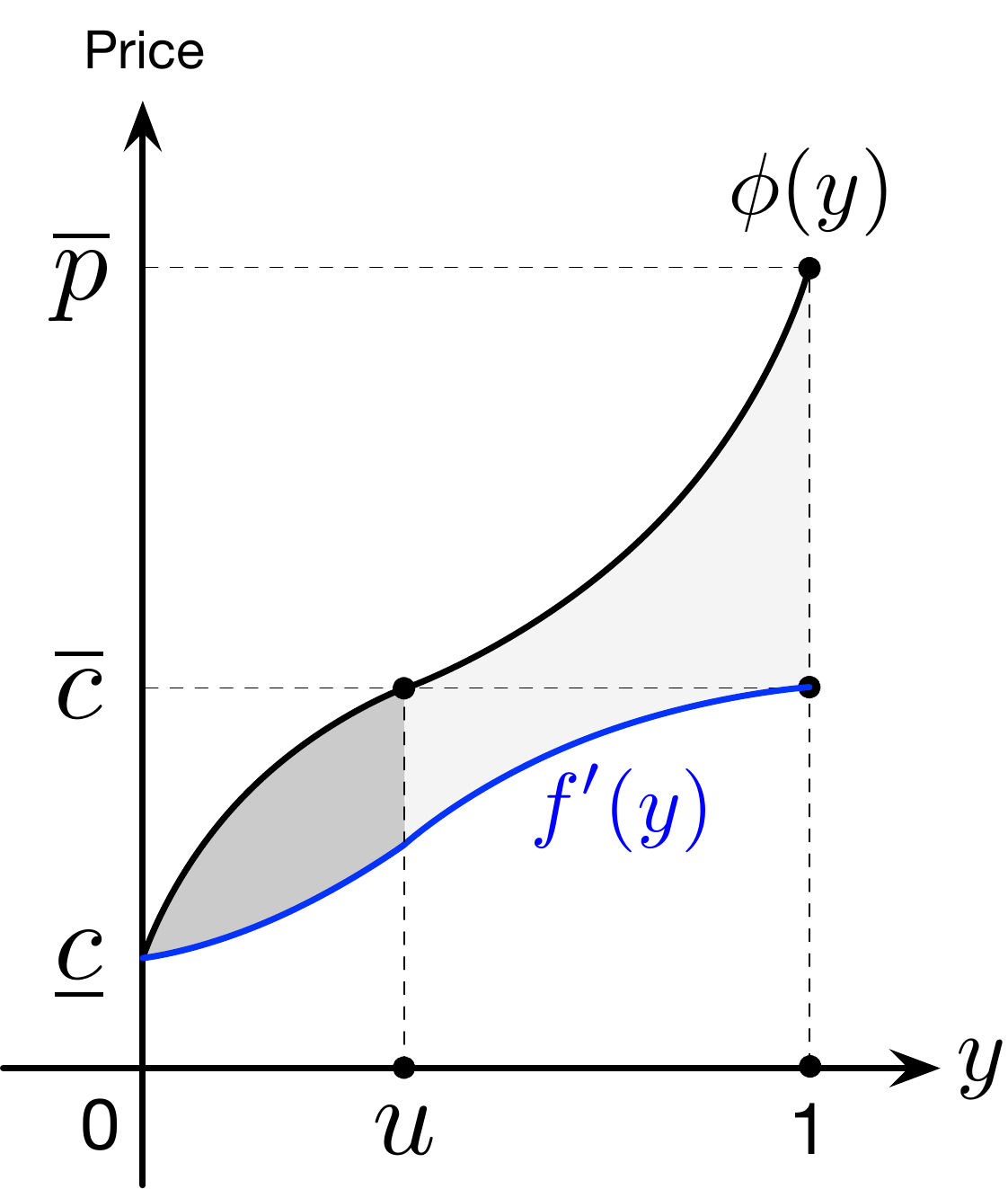}}
	\caption{Illustration of the pricing function $ \phi $ in \textsf{LUC} and \textsf{HUC}.}
	\label{three_pricing_schemes}
\end{figure} 

Suppose we have an arrival instance $ \mathcal{A}_v $ given as follows: \textit{for all $ y\in [0,v] $, there is a continuum of customers, indexed by $ y\in [0,v] $, whose valuations are given by $ v_y = \phi(y) \Delta y $, where $ \Delta y $ denotes the units of resources that are purchased by customer $ y $ and is infinitesimally small. For $ y\in (v,2v] $, there is another continuum of customers whose valuations are given by $ v_y = \overline{p}\Delta y $}. Given the arrival instance $ \mathcal{A}_v $, $ \text{PPM}_\phi $  will accept all the customers indexed by $ y\in [0,v] $. Thus, the social welfare achieved by $ \text{PPM}_\phi $ is the denominator of the right-hand-side of Eq. \eqref{ODE_LUC_inequality}, namely, the total valuation of all the accepted customers less the supply cost $ f(v) $.  The optimal offline social welfare, however, is to reject all the customers indexed by $ y\in [0,v] $ but only accept the second continuum of customers indexed by $ y\in (v,2v] $. Therefore, the optimal offline social welfare in hindsight is given by $  \overline{p}v - f(v) $, which is exactly the numerator of the right-hand-side of Eq. \eqref{ODE_LUC_inequality}. Therefore, a pricing function $ \phi(y) $ that satisfies $ \textsf{L}(\alpha) $ leads to the quotient in Eq. \eqref{ODE_LUC_inequality}, which captures the worst-case ratio  $ \alpha $ between the social welfare achieved by the optimal offline algorithm and $ \text{PPM}_\phi $. Based on the competitive ratio definition in Eq. \eqref{def_cr}, we can see that $ \text{PPM}_\phi $ is $ \alpha $-competitive.

(\textbf{Dividing Threshold})
Note that for the case of  \textsf{LUC}  in Fig. \ref{three_pricing_schemes}(a), the capacity limit 1 will never be reached. Otherwise, the system may suffer from negative social welfare (i.e., added valuations are smaller than the increased supply costs). In contrast, Fig. \ref{three_pricing_schemes}(b) illustrates a pricing function in \textsf{HUC} with $ \phi(0)=\underline{c} $ and $ \phi(1)=\overline{p} $. In this case, the capacity limit 1 can be reached as long as we have enough customers. In particular, there exists a threshold $ u\in (0,1) $ such that $ \phi(u) = f'(1) = \overline{c} $. In the following we refer to $ u\in (0,1) $ as the \textit{dividing threshold} of pricing function $ \phi $. The formal definition is given as follows.

\begin{definition}\label{def_dividing_threshold}
	Given a continuous pricing function $ \phi $ with $ \phi(0) < \overline{c} $ and $ \phi(1)>\overline{c} $, the \textit{dividing threshold} of $ \phi $ is the resource utilization level $ u\in (0,1) $  so that $ \phi(u) = f'(1) = \overline{c} $.
\end{definition}

In \textsf{HUC}, for any dividing threshold $ u\in (0,1) $, the whole interval of $ [0,1] $ is divided into segments $ [0,u] $ and $ [u,1] $.  When the lower and upper bounds of $ \phi $ are fixed, e.g., $ \phi(0)= \underline{c} $ and $ \phi(1) = \overline{p} $ in Fig. \ref{case_ABC}(b), the dividing threshold $ u $ has a strong impact on the curvature of $ \phi $. A smaller dividing threshold $ u $ indicates a steeper pricing curve in $ [0,u] $, and thus will perform better for arrival instances with high-PUVs. In contrast, a larger dividing threshold $ u $ indicates a less steep pricing curve within $ [0,u] $ and thus will perform better for arrival instances with low-PUVs. When there is no future information, we need to find a balance between these two so that the resulting online mechanism $ \text{PPM}_\phi $ has a stable performance regardless of arrival instances.  Theorem \ref{a_unified_BVP} captures this intuition by explicitly discriminating the pricing function design in $ [0,u] $ and $ [u,1] $ with two different BVPs in \textsf{HUC}. The next subsection shows that if the dividing threshold $ u$ is strategically chosen, the competitive ratio of $ \text{PPM}_\phi $ can be minimized.

\subsection{Structural Analysis for Optimal Design}
Recall that our objective is to design online mechanisms to achieve the value of $ \alpha $ which is as small as possible.  To quantify how small $ \alpha $ can possibly be, we define the \textbf{optimal competitive ratio}  in the following Definition \ref{def_opt_cr}.

\begin{definition}
	\label{def_opt_cr}
	Given a setup $ \mathcal{S} $, the competitive ratio $ \alpha $ is optimal if no other online algorithms can achieve a smaller competitive ratio under  Assumption \ref{assumption_cost_function}-Assumption \ref{assumption_bound}.
\end{definition}

Based on the necessity in Theorem \ref{a_unified_BVP}, to find the optimal competitive ratio for a given setup $ \mathcal{S} $, it suffices to find the minimum $ \alpha $ so that there exist strictly-increasing solutions to the BVPs in Theorem \ref{a_unified_BVP}. Hence, we give Proposition \ref{optimality} below. 
\begin{proposition}
	\label{optimality}
	Given a setup $ \mathcal{S} $, if $\alpha_*(\mathcal{S}) $ is defined as follows:
	\begin{align*}
	\normalfont
	\alpha_*(\mathcal{S}) \triangleq \inf\left\{ \alpha\ \Bigg|\ 
	\begin{matrix*}[l]
	\mbox{there exists a strictly-increasing}\\
	\mbox{function $ \phi $ so that i) if $ \overline{p}\in (\underline{c},\overline{c}] $,}\\
	\mbox{$ \phi $ is a solution to \normalfont $\textsf{L}(\alpha) $, or ii) if}\\
	\mbox{$ \overline{p}>\overline{c} $, $ \phi $ is a solution to \normalfont $\textsf{H}_1(u,\alpha) $}\\
	\mbox{and $\textsf{H}_2(u,\alpha)\} $ with a feasible}\\
	\mbox{dividing threshold $ u\in (0,1) $.}
	\end{matrix*}
	\right\},
	\end{align*}
	then $ \alpha_*(\mathcal{S}) $ is the optimal competitive ratio achievable by all online algorithms.
\end{proposition} 

Proposition \ref{optimality} directly follows the necessity of Theorem \ref{a_unified_BVP}. Based on Proposition \ref{optimality}, we have the following corollary.
\begin{corollary}\label{lower_bound_alpha_f_p}
	Given a setup $ \mathcal{S} $, there exists no $ (\alpha_*(\mathcal{S})-\epsilon) $-competitive online algorithm, $ \forall \epsilon>0 $. 
\end{corollary}

Based on Proposition \ref{optimality}, to obtain the optimal competitive ratio $ \alpha_*(\mathcal{S}) $, we just need to characterize the existence conditions of strictly-increasing solutions to the BVPs in Theorem \ref{a_unified_BVP}.  Note that in \textsf{LUC}, for a given setup $ \mathcal{S} $, $ \mathsf{L}(\alpha) $ is not indexed by any other parameters except the competitive ratio parameter $ \alpha $, and thus, $ \alpha_*(\mathcal{S}) $ is the minimum $ \alpha $ so that there exists a strictly-increasing solution to  $ \mathsf{L}(\alpha) $. However, in \textsf{HUC}, both the two BVPs are indexed by the dividing threshold $ u $, which is a design variable that can be flexibly chosen within $ (0,1) $. As a result, the minimum $ \alpha $ to guarantee the existence of  strictly-increasing solutions to  $ \{\textsf{H}_1(u,\alpha),  \textsf{H}_2(u,\alpha)\} $ will depend on $ u $. To characterize this dependency, we define the lower bound of $ \alpha $ for each given $ u\in (0,1) $ as follows.

\begin{definition}[\textbf{Lower Bound of $ \alpha $ in} \textsf{HUC}]\label{def_achievable_region}
	Given a setup $ \mathcal{S} $ with $ \overline{p}\in (\overline{c},+\infty) $, the lower bound of $ \alpha $ for any given $ u\in (0,1) $, denoted by $ \underline{\alpha}(u) $, is defined as follows:
	\begin{align*}\normalfont
	\underline{\alpha}(u) \triangleq 
	\inf \left\{ \alpha\ \Bigg|\ 
	\begin{matrix*}[l]
	\mbox{There exists a strictly-increasing}\\
	\mbox{pricing function $ \phi(y) $ that is a}\\
	\mbox{solution to \normalfont $\{\textsf{H}_1(u,\alpha),\textsf{H}_2(u,\alpha)\} $.}
	\end{matrix*}
	\right\}.
	\end{align*}
\end{definition}

Based on Definition \ref{def_achievable_region}, the optimal competitive ratio can be calculated as follows:
\begin{equation}\label{opt_cr_equal_R}
\alpha_*(\mathcal{S}) = \underline{\alpha}(u_*), \text{ where } u_* = \arg\min_{u\in (0,1)} \underline{\alpha}(u),
\end{equation} 
where $ u_* $ denotes the optimal dividing threshold. 

\begin{algorithm}
	\caption{Principles of Optimal Design}  
	\begin{algorithmic}[1]
		\STATE \textbf{Input}: the setup $ \mathcal{S} $ with $ \overline{p}\in (\underline{c},+\infty) $.
		\IF{$\overline{p}\in (\underline{c},\overline{c}]$}
		\STATE Get the minimum $ \alpha $, denoted by $\alpha_*(\mathcal{S}) $, so that \\
		there exists a strictly-increasing solution to $ \textsf{L}(\alpha) $.  \label{lower_bound_1_alg}
		\STATE Solve $ \textsf{L}(\alpha_*(\mathcal{S})) $ and get the optimal pricing \\
		function $ \phi $ so that $ \text{PPM}_\phi $ is $ \alpha_*(\mathcal{S}) $-competitive. \label{solve_L_opt}
		\ELSE 
		\STATE Get the lower bound $ \underline{\alpha}(u) $ based on Definition \ref{def_achievable_region}.  \label{lower_bound_2_alg}
		\STATE Obtain $ \alpha_*(\mathcal{S}) $ and $ u_*\in (0,1) $ based on Eq. \eqref{opt_cr_equal_R}. \label{solve_lower_bound_opt} 
		
		\STATE Solve $\{ \textsf{H}_1(u_*,\underline{\alpha}(u_*)), \textsf{H}_2(u_*,\underline{\alpha}(u_*))\} $ and get \\ 
		the optimal pricing function $ \phi $ so that $ \text{PPM}_\phi $ is $ \alpha_*(\mathcal{S}) $-competitive or $ \underline{\alpha}(u_*) $-competitive.	
		\ENDIF	
		
		\STATE \textbf{Output}: $ \alpha_*(\mathcal{S}) $ and optimal pricing functions. 
	\end{algorithmic}
	\label{principles_optimal_design}
\end{algorithm}

Algorithm \ref{principles_optimal_design} summarizes the above  structural results and provides a principled way to characterize the optimal competitive ratio and the corresponding optimal pricing function for any given setup $\mathcal{S}$. The key steps in Algorithm \ref{principles_optimal_design} are line \ref{lower_bound_1_alg} and line \ref{lower_bound_2_alg}, in which we need to characterize the conditions for the existence of strictly-increasing solutions to the BVPs in Theorem \ref{a_unified_BVP}. We emphasize that characterizing such existence conditions heavily depends on the cost function $ f $.  The next section will demonstrate how such conditions can be derived in analytical forms when $ f $ is a power function.

\section{Case Study: $ f(y) = ay^s $}
We now perform a case study for  $ f(y) = ay^s $ (i.e., power function), and show how to use Algorithm \ref{principles_optimal_design} to obtain the minimum value of $ \alpha $, the optimal dividing threshold $ u_* $, and the corresponding optimal pricing functions.  At the end of this section, we will discuss some important structural properties about the optimal pricing functions.

\subsection{Preliminaries: The BVPs in Both Cases}
We consider $ f(y) = ay^s $ with $ a>0 $ and $ s > 1$ so that the marginal cost $ f'(y) = asy^{s-1} $ is strictly increasing.   Such power cost functions are often used for modeling the  costs that are diseconomies-of-scale (i.e., no volume discounts). For example, when $ s \geq 2 $, $ f(y) $ is a classic power-rate curve, reflecting the power consumption of a general networking and computing device with the capability of speed-scaling \cite{Adam_speed_scaling, speed_scaling}, e.g., CPU, edge router, and communication link. It is also common to use $ s= 1\sim3 $ to model the power consumption of data centers in cloud computing \cite{data_center_power, XZhang_ToN}.

When $ f(y) = ay^s $,  the minimum marginal cost is $ \underline{c} = f'(0)=  0 $ and the maximum marginal cost is $ \overline{c} = f'(1) = as $. Based on Theorem \ref{a_unified_BVP},    $ \textsf{L}(\alpha),  $
$ \textsf{H}_1(u,\alpha) $, and $ \textsf{H}_2(u,\alpha) $ can be written as follows:
\begin{itemize}
	\item \underline{\textsf{LUC}: $ \overline{p}\in (\underline{c},\overline{c}] $}.  $ \textsf{L}(\alpha) $ is given by
	\begin{align}\label{BVP_power_LUC}
	\begin{cases}
	\phi'(y) =  \alpha\cdot \frac{\phi(y) - f'(y)}{(\phi(y)/\overline{c})^{\frac{1}{s-1}}}, y\in (0,v), \\
	\phi(0) = 0, 
	\phi(v) \geq  \overline{p},
	\end{cases}
	\end{align}
	where $ v = f'^{-1}(\overline{p}) = (\overline{p}/\overline{c})^{\frac{1}{s-1}} $. 
	
	\item \underline{\textsf{HUC}: $ \overline{p}\in (\overline{c},+\infty) $}. $\{ \textsf{H}_1(u,\alpha), \textsf{H}_2(u,\alpha)\}$ are given by
	\begin{subequations}\label{two_BVP_power}
		\begin{align}\label{two_BVP_power_1}
		&\begin{cases}
		\phi'(y) =  \alpha\cdot \frac{\phi(y) - f'(y)}{(\phi(y)/\overline{c})^{\frac{1}{s-1}}}, y\in (0,u), \\
		\phi(0) = 0, 
		\phi(u) = \overline{c},
		\end{cases}\\
		&\begin{cases}
		\phi'(y) =   \alpha\cdot \left(\phi(y) - \overline{c}y^{s-1}\right), y\in (u,1), \\
		\phi(u) = \overline{c},
		\phi(1) \geq   \overline{p},
		\end{cases} \label{two_BVP_power_2}
		\end{align} 
	\end{subequations}
    where Problem \eqref{two_BVP_power_1} corresponds to $ \textsf{H}_1(u,\alpha) $, and Problem \eqref{two_BVP_power_2} corresponds to $ \textsf{H}_2(u,\alpha) $. 
\end{itemize}

Following lines \ref{lower_bound_1_alg} and \ref{lower_bound_2_alg} in Algorithm \ref{principles_optimal_design}, the next subsection will characterize the conditions for the existence of strictly-increasing solutions to the BVPs in Eq. \eqref{BVP_power_LUC} and Eq. \eqref{two_BVP_power}.

\subsection{Lower Bound of $ \alpha $ in \textsf{LUC} and \textsf{HUC}}

\subsubsection{Lower Bound of $ \alpha $ in \textsf{LUC}}
We first focus on \textsf{LUC} and give the following Theorem \ref{lower_bound_LUC}.

\begin{theorem}
	\label{lower_bound_LUC}
	Given a setup $ \mathcal{S} $ with $ f(y) = ay^s$ and $ \overline{p}\in (\underline{c},\overline{c}]$, there exist strictly-increasing solutions to Problem \eqref{BVP_power_LUC} if and only if $ \alpha \geq \alpha_s^{\min} $, where $ \alpha_s^{\min} = s^{\frac{s}{s-1}} $.
\end{theorem} 

Theorem \ref{lower_bound_LUC} provides the lower bound of $ \alpha $ so that there exists a strictly-increasing solution to Problem \eqref{BVP_power_LUC} above.  Based on Proposition \ref{optimality}, we can conclude that the optimal competitive ratio $ \alpha_*(\mathcal{S}) = \alpha_s^{\min} $. According to line \ref{solve_L_opt} in Algorithm \ref{principles_optimal_design},  the design of optimal pricing functions in \textsf{LUC} is equivalent to solving  Problem \eqref{BVP_power_LUC} with $ \alpha = \alpha_*(\mathcal{S}) = \alpha_s^{\min} $. In Section \ref{section_optimal_pricing_functions}, we will discuss how to solve Problem \eqref{BVP_power_LUC} to get a set of infinitely-many optimal pricing functions. 

\subsubsection{Lower Bound of $ \alpha $ in \textsf{HUC}}  Theorem \ref{lower_bound_1} below summarizes a necessary and sufficient condition for $ \alpha $ such that we can guarantee the existence of a strictly-increasing solution to  Problem \eqref{two_BVP_power_1} and this solution is unique. 

\begin{theorem}
\label{lower_bound_1}
Given a setup $ \mathcal{S} $ with $ f(y) = ay^s$ and $ \overline{p} > \overline{c} $, for any $ u\in(0,1) $, there exists a unique strictly-increasing solution to Problem \eqref{two_BVP_power_1} if and only if $ \alpha \geq \underline{\alpha}_1(u) $, where $\underline{\alpha}_1(u) $ is given by
\begin{align}\label{def_lower_bound}
\underline{\alpha}_1(u) =  
\begin{cases}
\alpha_s(u)  & \text{ if } u \in\left(0,u_s\right),\\
\alpha_s^{\min} & \text{ if } u\in\left[u_s,1\right).
\end{cases}
\end{align}
In Eq. \eqref{def_lower_bound}, $ \alpha_s(u) $ and $ u_s$ are given as follows:
\begin{equation}\label{def_alpha_gamma}
\alpha_s(u) = \frac{s-1}{u-u^s},  
u_s =
\Big(\frac{1}{s}\Big)^{\frac{1}{s-1}}.
\end{equation}
\end{theorem}
\begin{proof}
	The proof of the above two theorems is non-trivial since the right-hand-side of the ODE in Problem \eqref{two_BVP_power_1} (also Problem  \eqref{BVP_power_LUC}) has a singular boundary condition at $ \phi(0)= 0 $ \cite{Perko2001}. The detailed proof is given in Appendix \ref{proof_lower_bound_1}.   
\end{proof}

Theorem \ref{lower_bound_1}  provides a lower bound of $ \alpha $ for each given dividing threshold $ u $. Note that $ \alpha_s(u_s) = \alpha_s^{\min} $. Thus, $ \underline{\alpha}_1(u) $ is continuous in $ u\in (0,1) $. Meanwhile, $ \underline{\alpha}_1(u) $ is non-increasing in $ u\in (0,1) $ and achieves its minimum $ \alpha_s^{\min}$ when $ u\in [u_s,1) $. However, we cannot directly conclude that the optimal competitive ratio in \textsf{HUC} is also $ \alpha_s^{\min} $.  This is because it is unclear whether there exists any strictly-increasing solution to Problem \eqref{two_BVP_power_2} when $ u\in [u_s,1) $ and $ \alpha = \alpha_s^{\min} $. To answer this question, below we give Theorem \ref{lower_bound_2}.

\begin{theorem}
	\label{lower_bound_2}
	Given a setup $ \mathcal{S} $  with $ f(y) = ay^s$ and $ \overline{p} > \overline{c}$, for any $ u\in(0,1) $, there exists a unique strictly-increasing solution to Problem \eqref{two_BVP_power_2} if and only if $ \alpha \geq \underline{\alpha}_2(u) $, where $\underline{\alpha}_2(u) $ is the unique root to the following equation
	\begin{align}\label{def_lower_bound_2}
	\int_{u\underline{\alpha}_2(u)}^{\underline{\alpha}_2(u)}\eta^{s-1} e^{-\eta}d\eta
	= \frac{\big(\underline{\alpha}_2(u)\big)^{s-1}}{\exp(u\underline{\alpha}_2(u))} - \frac{\overline{p} \big(\underline{\alpha}_2(u)\big)^{s-1}}{\overline{c} \exp(\underline{\alpha}_2(u))} .
	\end{align}
	Meanwhile, $ \underline{\alpha}_2(u) $ is strictly-increasing in $ u\in (0,1) $. 
\end{theorem}
\begin{proof}
	The proof of the lower bound $ \underline{\alpha}_2(u) $ is trivial since the ODE in Problem \eqref{two_BVP_power_2} can be solved in analytical forms. 
	The detailed proof is given in Appendix \ref{proof_lower_bound_2}.
\end{proof}

Based on Theorem \ref{lower_bound_1} and Theorem \ref{lower_bound_2}, to guarantee the existence of strictly-increasing solutions to Problem \eqref{two_BVP_power_1} and Problem \eqref{two_BVP_power_2} simultaneously,  $ \alpha $  must be jointly lower bounded by $ \underline{\alpha}_1(u) $ and $ \underline{\alpha}_2(u) $ for all $ u\in (0,1) $. Therefore, the lower bound of $ \alpha $ is given by
\begin{equation}
\underline{\alpha}(u) = \max\ \{\underline{\alpha}_1(u), \underline{\alpha}_2(u)\}, \forall u\in (0,1),
\end{equation} 
which follows our definition of $ \underline{\alpha}(u) $ in Definition \ref{def_achievable_region}. Note that if  $ \mathcal{R}(u,\alpha) \subset (0,1)\times [1,+\infty) $ is defined as follows:
\begin{equation}\label{achievable_region}
\mathcal{R}(u,\alpha) \triangleq  \{(u,\alpha)|\alpha \geq \underline{\alpha}(u),
u\in (0,1) \}.
\end{equation}
Then, for any given $ (u,\alpha)\in \mathcal{R}(u,\alpha) $, the resulting BVPs $\{\normalfont\textsf{H}_1(u,\alpha),\textsf{H}_2(u,\alpha)\} $ must have a strictly-increasing solution. For this reason, we will refer to $ \mathcal{R}(u,\alpha) $ as the \textit{achievable region} of $ (u,\alpha) $. 

Based on line \ref{solve_lower_bound_opt} in Algorithm \ref{principles_optimal_design}, to get the optimal competitive ratio  $ \alpha_*(\mathcal{S}) $ in \textsf{HUC}, we need to find the optimal dividing threshold $ u_* $ by solving the following problem
\begin{equation*}
u_* = \arg\min_{u\in (0,1)} \underline{\alpha}(u) = \arg\min_{u\in (0,1)}\max \left\{\underline{\alpha}_1(u), \underline{\alpha}_2(u)\right\},
\end{equation*}
where $ \underline{\alpha}_1(u) $ is analytically given in Eq. \eqref{def_lower_bound}, and $\underline{\alpha}_2(u)  $ is the unique root to Eq. \eqref{def_lower_bound_2}. 
The next section will show that the optimal dividing threshold $ u_* $ always exists. However, the uniqueness of $ u_* $ depends on the value of $\overline{p}$.

\subsection{Optimal Competitive Ratios}
To characterize the optimal dividing threshold $ u_* $, we give the following 
Proposition \ref{calculation_u_star} to show the unique existence of an intersection point between $ \underline{\alpha}_1(u) $ and $\underline{\alpha}_2(u) $, which we refer to as the \textbf{critical dividing threshold} (CDT), denoted by $ u_{\mathsf{cdt}} $.

\begin{proposition}
	\label{calculation_u_star}
	Given a setup $ \mathcal{S} $ with $ f(y) = ay^s$ and $ \overline{p}\in (\overline{c},+\infty)$, 
 there exists a unique CDT $ u_{\mathsf{cdt}}\in (0,1) $ such that $  \underline{\alpha}_1(u_{\mathsf{cdt}}) =  \underline{\alpha}_2(u_{\mathsf{cdt}}) $. Specifically, if we define $ C_s $ by
	\begin{align}\label{C_s}
	C_s \triangleq \overline{c}\cdot \Big(\frac{1}{e^s}- \frac{1}{s^s}\cdot\int_{s}^{\alpha_s^{\min}}\eta^{s-1} e^{-\eta}d\eta\Big)\cdot \exp(\alpha_s^{\min}),
	\end{align}
	 then the unique CDT can be calculated as follows:
	\begin{itemize}
		\item \underline{$ \normalfont\textsf{HUC}_1 ${\normalfont :} $ \overline{p}\in (\overline{c}, C_s] $}. In this case, the CDT is the unique root to the following equation in variable $ u_{\mathsf{cdt}}\in [u_s,1) $: 
		\begin{equation} \nonumber
	    \int_{u_{\mathsf{cdt}}\cdot\alpha_s^{\min}}^{\alpha_s^{\min}}\eta^{s-1} e^{-\eta}d\eta
		= \frac{s^s }{\exp(u_{\mathsf{cdt}}\cdot\alpha_s^{\min})} - \frac{\overline{p} s^s }{ \overline{c} \exp(\alpha_s^{\min})}. 
		\end{equation}
		\item \underline{$ \normalfont\textsf{HUC}_2 ${\normalfont :} $ \overline{p}\in (C_s, +\infty) $}. In this case, the CDT is the unique root to the following equation in variable $ u_{\mathsf{cdt}}\in(0,u_s) $: 
		\begin{align*}
		 \int_{u_{\mathsf{cdt}}\cdot\alpha_s(u_{\mathsf{cdt}})}^{\alpha_s(u_{\mathsf{cdt}})}\eta^{s-1} e^{-\eta}d\eta
		= \frac{\left(\alpha_s(u_{\mathsf{cdt}})\right)^{s-1} }{\exp(u_{\mathsf{cdt}}\cdot\alpha_s(u_{\mathsf{cdt}}))} - \frac{\overline{p}\cdot \left(\alpha_s(u_{\mathsf{cdt}})\right)^{s-1} }{ \overline{c}\cdot \exp(\alpha_s(u_{\mathsf{cdt}}))}.
		\end{align*}
	\end{itemize}
\end{proposition}
\begin{proof}
	This corollary follows the previous two theorems regarding the lower bound $ \underline{\alpha}_1(u) $ and $ \underline{\alpha}_2(u) $.  
	The detailed proof is given in Appendix \ref{proof_calculation_u_star}.
\end{proof}

\begin{figure}
	\centering
	\subfigure[$ \normalfont\textsf{HUC}_1 ${\normalfont :} $ \overline{p} \in {(\overline{c},C_s]} $]{\includegraphics[width = 4.7 cm]{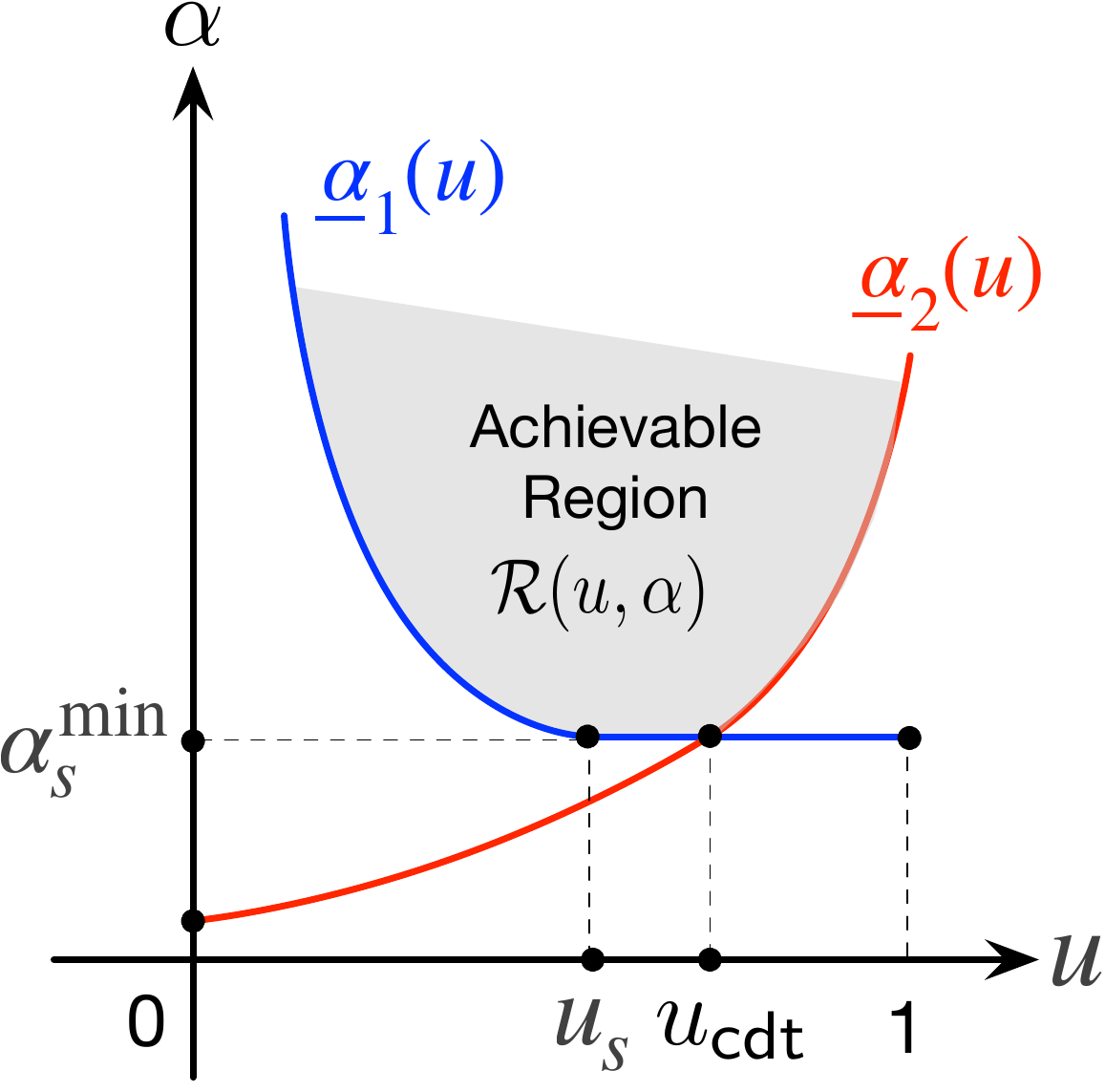}}
	\qquad 
	\subfigure[$ \normalfont\textsf{HUC}_2 ${\normalfont :}$ \overline{p} \in (C_s,+\infty) $]{\includegraphics[width = 4.7 cm]{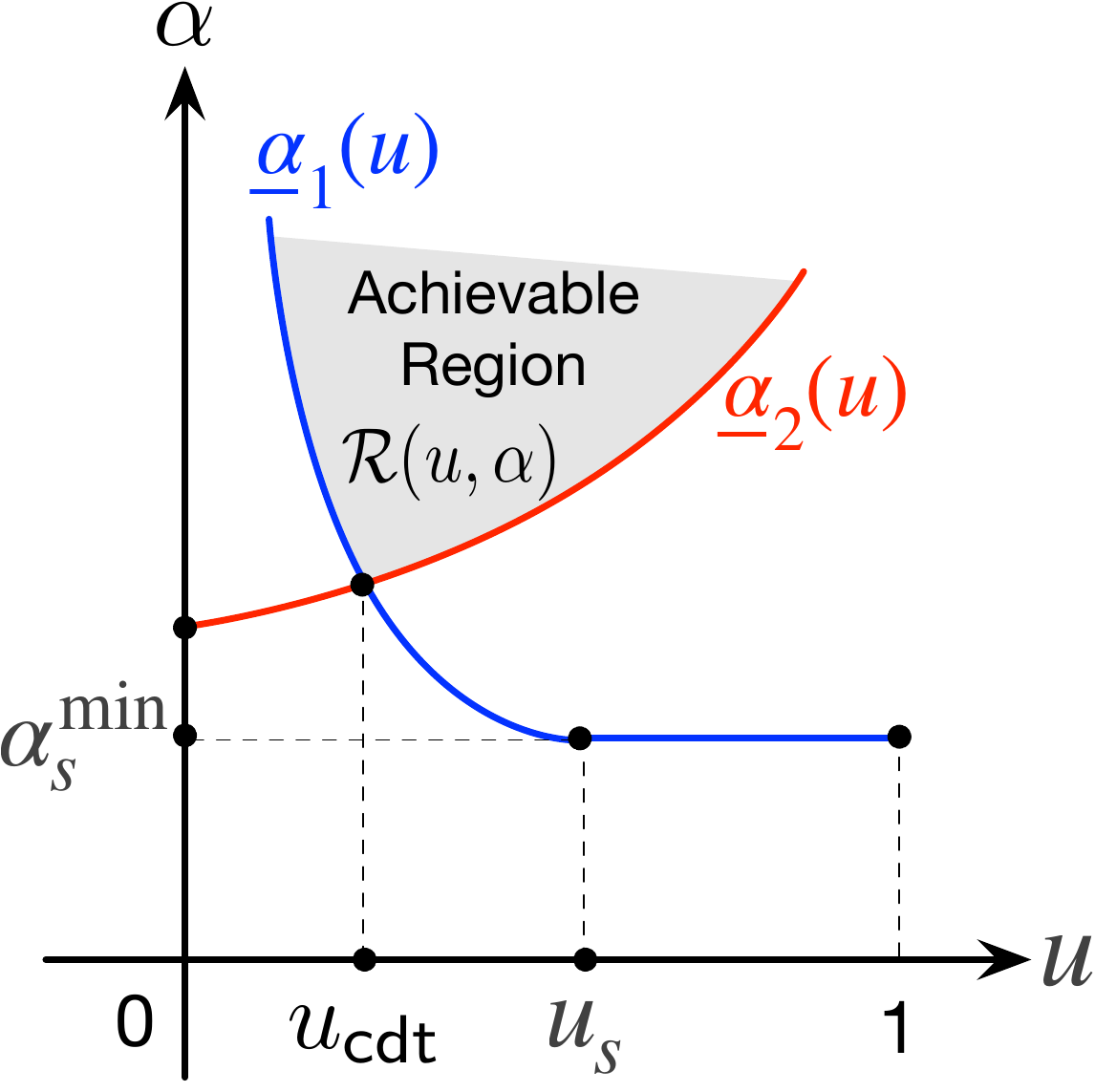}} 
	\caption{Illustration of the two lower bounds $ \underline{\alpha}_1(u) $, $ \underline{\alpha}_2(u) $, and  $ \mathcal{R}(u,\alpha) $.}
	\label{case_ABC}
\end{figure}

Fig. \ref{case_ABC} illustrates $ \underline{\alpha}_1(u) $ and $ \underline{\alpha}_2(u) $ in two cases. As can be seen from Fig. \ref{case_ABC}(a), in $ \textsf{HUC}_1 $ (i.e., $ \overline{p}\in (\overline{c}, C_s] $), the CDT $ u_{\mathsf{cdt}} \in  [u_s,1) $,  and the optimal competitive ratio $ \alpha_*(\mathcal{S}) = \underline{\alpha}(u_{\mathsf{cdt}}) = \alpha_s^{\min} $. In this case, any dividing threshold $ u\in [u_s,u_{\mathsf{cdt}}] $ and $ \alpha = \alpha_s^{\min} $  will determine an optimal pricing function that satisfies Problem \eqref{two_BVP_power_1} and Problem \eqref{two_BVP_power_2}. Therefore, the optimal dividing threshold $ u_* $ is not unique and can be any value within the interval  $  [u_s,u_{\mathsf{cdt}}] $.  In comparison, as shown in Fig. \ref{case_ABC}(b), in $ \textsf{HUC}_2 $ (i.e., $ \overline{p}\in (C_s,+\infty) $), the unique CDT $ u_{\mathsf{cdt}} $ is within the interval $ (0, u_s) $ and is the unique optimal dividing threshold (i.e., $ u_* = u_{\mathsf{cdt}}  $ and $  \alpha_*(\mathcal{S}) = \alpha_s(u_{\mathsf{cdt}}) $). In this case, the optimal pricing function is the unique solution to Problem \eqref{two_BVP_power_1} and Problem \eqref{two_BVP_power_2} with $ u = u_{\mathsf{cdt}} $ and $ \alpha = \alpha_s(u_{\mathsf{cdt}}) $.  

Corollary \ref{opt_cr} summarizes the optimal competitive ratios in \textsf{LUC} and the two sub-cases in \textsf{HUC}.

\begin{corollary}
	\label{opt_cr}
	Given a setup $ \mathcal{S} $ with $ f(y) = ay^s$,  the optimal competitive ratio $ \alpha_*(\mathcal{S}) $ is given by
	\begin{align}\label{cases_opt_cr}
	\alpha_*(\mathcal{S}) = 
	\begin{cases}
	s^{\frac{s}{s-1}} & \text{ if } \overline{p}\in (\underline{c},\overline{c}], \qquad  (\normalfont\textsf{LUC})\\
	s^{\frac{s}{s-1}} & \text{ if } \overline{p}\in (\overline{c},C_s ], \quad\   (\normalfont\textsf{HUC}_1)\\
	\frac{s-1}{u_{\mathsf{cdt}}-u_{\mathsf{cdt}}^s} & \text{ if } \overline{p}\in (C_s,+\infty), (\normalfont\textsf{HUC}_2)
	\end{cases}
	\end{align}
	where $ C_s $ and $ u_{\mathsf{cdt}} $ can be calculated based on  Proposition \ref{calculation_u_star}.
\end{corollary}

The optimal competitive ratio in \textsf{LUC} directly follows Theorem \ref{lower_bound_LUC}, and  the optimal competitive ratios in $ \textsf{HUC}_1 $ and $ \textsf{HUC}_2 $ follow Theorem \ref{lower_bound_1}, Theorem \ref{lower_bound_2}, and Proposition \ref{calculation_u_star}. Note that the first two cases of \textsf{LUC} and $ \textsf{HUC}_1 $ in Eq. \eqref{cases_opt_cr} can be combined together. However, we keep the current three-case form so that it clearly distinguishes \textsf{LUC} and \textsf{HUC}.

\subsection{Optimal Pricing Functions}\label{section_optimal_pricing_functions}
Based on Corollary \ref{opt_cr} and Algorithm \ref{principles_optimal_design}: i) to get the optimal pricing function for \textsf{LUC}, we need to solve  $ \textsf{L}(\alpha
) $ with $ \alpha = s^{\frac{s}{s-1}} $; ii) to get the optimal pricing function for $ \textsf{HUC}_1 $, we need to solve $\{ \textsf{H}_1(u,\alpha), \textsf{H}_2(u,\alpha)\} $  with any $ u\in [u_s, u_{\mathsf{cdt}}]$ and $ \alpha  =  s^{\frac{s}{s-1}} $; iii) to get the optimal pricing function for $ \textsf{HUC}_2 $, we need to solve $\{ \textsf{H}_1(u,\alpha), \textsf{H}_2(u,\alpha)\} $  with $ u=  u_{\mathsf{cdt}}$ and $ \alpha  = \frac{s-1}{u_{\mathsf{cdt}}-u_{\mathsf{cdt}}^s} $. 

To help characterize the optimal pricing functions for the above three cases,  we first focus on the following first-order initial value problem (IVP):
\begin{equation}\label{IVP}
\begin{cases}
\phi_{\mathsf{ivp}}'(y) =   \alpha\cdot \left(\phi_{\mathsf{ivp}}(y) - \overline{c}y^{s-1}\right), y\in (u,1),\\
\phi_{\mathsf{ivp}}(u) = \overline{c}.
\end{cases}
\end{equation}
Problem \eqref{IVP} is the same as Problem \eqref{two_BVP_power_2} if we exclude the second boundary condition $ \phi(1)\geq \overline{p} $. Based on the Picard-Lindel\"{o}f theorem \cite{Perko2001, ODE1973}, the IVP in Eq. \eqref{IVP} always has a unique strictly-increasing solution for all $ \alpha\in \mathbb{R} $. 
We solve Problem \eqref{IVP} with $ \alpha = \underline{\alpha}_1(u) $, and denote the unique solution by $ \phi_{\mathsf{ivp}}\big(y;u\big) $ as follows:
\begin{align}
\phi_{\mathsf{ivp}}\big(y;u\big) =  
 \overline{c}\cdot\frac{\exp\left(y\cdot\underline{\alpha}_1\big(u\big)\right)}{\left(\underline{\alpha}_1\big(u\big)\right)^{s-1}}\cdot
\int_{y\underline{\alpha}_1\left(u\right)}^{\underline{\alpha}_1\left(u\right) u}\eta^{s-1} e^{-\eta}d\eta  
+ \overline{c}\cdot\exp\big((y-u)\cdot\underline{\alpha}_1(u)\big), 
 y\in [u,1].\label{phi_2_u_opt}
\end{align}
Intuitively, if $ \phi_{\mathsf{ivp}}(1;u)\geq \overline{p} $, then $ \phi_{\mathsf{ivp}}\big(y;u\big) $ is also a solution to Problem \eqref{two_BVP_power_2}. Below in Lemma \ref{solution_to_phi_2} we show that $ \phi_{\mathsf{ivp}}(1;u)\geq \overline{p} $ holds as long as $ u\in [u_s, u_{\mathsf{cdt}}] $.

\begin{lemma}\label{solution_to_phi_2}
	Given $ \overline{p}\in (\overline{c},+\infty) $, for any $ u\in [u_s,u_{\mathsf{cdt}}]$, $ \phi_{\mathsf{ivp}}\big(y;u\big) $ is a solution to Problem \eqref{two_BVP_power_2} with $ \phi_{\mathsf{ivp}}(1;u) \geq   \overline{p} $. 
\end{lemma}

We also give the following lemma to show the existence of a unique resource utilization level $ \rho_s $ such that  $ \phi_{\mathsf{ivp}}(\rho_s;u_s) = \overline{p} $. 
\begin{lemma}\label{rho_s}
	If the value of $ \rho_s $ leads to $ \phi_{\mathsf{ivp}}(\rho_s;u_s) = \overline{p}$, then $ \rho_s $ is the unique root to the following equation:
	\begin{equation}\label{calculation_rho_s}
	\int_{s}^{\alpha_s^{\min} \rho_s} \eta^{s-1}e^{-\eta}d\eta = \frac{s^s}{\exp(s)} - \frac{\overline{p} s^s}{\overline{c}\cdot \exp(\alpha_s^{\min} \rho_s)}.
	\end{equation}
\end{lemma}

The proofs of the above two lemmas  are given in Appendix \ref{proof_of_solution_to_phi_2}. Based on Eq. \eqref{phi_2_u_opt}, Lemma \ref{solution_to_phi_2}, and Lemma \ref{rho_s} above, we next give Theorem \ref{optimal_pricing_functions_theorem} which summarizes the optimal pricing functions for all cases of \textsf{LUC}, $ \textsf{HUC}_1 $, and $ \textsf{HUC}_2 $. 

\begin{theorem}
	\label{optimal_pricing_functions_theorem} Given a setup $ \mathcal{S} $ with $ f(y) = ay^s$, the optimal pricing functions for $ \text{PPM}_\phi $ are determined as follows. 
	\begin{itemize} 
		\item \underline{\normalfont\textsf{LUC}: $ \overline{p} \in (\underline{c}, \overline{c}] $}. Let us define $ w \triangleq f'^{-1}(\overline{p}/s) $, then we have $ 0< w <v\leq 1 $, where $ v = f'^{-1}(\overline{p}) $. For any $ m\in [w,v] $, $ \text{PPM}_{\phi_m}  $ achieves the optimal competitive ratio of $ s^{\frac{s}{s-1}} $ if $  \phi_m $ is given by:
		\begin{equation}
		\begin{aligned}\label{opt_pricing_case_luc}
		\phi_m(y) = \begin{cases}
		0 &\text{if } y = 0,\\
		\overline{c} \big(\varphi_{\textsf{luc}}(y)\big)^{s-1} &\text{if } y\in (0,m],
		\end{cases}
		\end{aligned}
		\end{equation}
		where  for each given $ y\in (0,m] $, 
		$ \varphi_{\textsf{luc}}(y)$ is the unique root to the following equation in variable $ \varphi_{\textsf{luc}}\in(0,1] $:
		\begin{equation}\label{characteristic_poly_luc}
		\int_{1/m}^{\varphi_{\textsf{luc}}/y}\frac{\eta^{s-1}}{\eta^s-\frac{\alpha_s^{\min}}{s-1}\eta^{s-1}+\frac{\alpha_s^{\min}}{s-1}}\ d\eta = \ln\left(\frac{m}{y}\right).
		\end{equation}
		Meanwhile, when $ m = w = f'^{-1}(\overline{p}/s) $, the optimal pricing function $ \phi_w(y) $ is given by
		\begin{equation}\label{phi_w}
		\phi_w(y) = sf'(y), y\in [0,w]. 
		\end{equation}
		
		\item \underline{$ \normalfont\textsf{HUC}_1 ${\normalfont :} $\overline{p}\in (\overline{c}, C_s]$}. In this case, the CDT $ u_{\mathsf{cdt}}\in [u_s,1) $, and  for each $ u\in [u_s,u_{\mathsf{cdt}}] $,  $ \text{PPM}_{\phi_u}  $ achieves the optimal competitive ratio of $ s^{\frac{s}{s-1}} $ if $ \phi_{u} $ is given by:
		\begin{align}\label{opt_pricing_case_huc_1}
			\phi_u(y) =
			\begin{cases}
			0 &\text{if } y = 0,\\
			\overline{c} \big(\varphi_{\textsf{huc}}(y)\big)^{s-1} &\text{if } y\in (0,u),\\
			\phi_{\mathsf{ivp}}\big(y;u\big)  &\text{if } y\in [u,\rho],
			\end{cases}
		\end{align}
		where  for any given $ y\in (0,u) $, $ \varphi_{\textsf{huc}}(y) $ is the unique root to the following equation in variable $ \varphi_{\textsf{huc}}\in (0,1) $: 
		\begin{equation}\label{characteristic_poly_huc_1}
		\int_{1/u}^{\varphi_{\textsf{huc}}/y}\frac{\eta^{s-1}}{\eta^s-\frac{\alpha_s^{\min}}{s-1}\eta^{s-1}+\frac{\alpha_s^{\min}}{s-1}}\ d\eta  = \ln\left(\frac{u}{y}\right).
		\end{equation}
		In Eq. \eqref{opt_pricing_case_huc_1},  $ \rho\in [\rho_s, 1] $ is the maximum resource utilization level that satisfies $ \phi_{\mathsf{ivp}}(\rho;u) = \overline{p} $, where $ \rho_s $ is given by Lemma \ref{rho_s}. In particular, if $ u = u_s $, then $ \rho = \rho_s $; if $ u = u_{\mathsf{cdt}} $, then   $ \rho = 1 $. Meanwhile, if $ u = u_s $, the optimal pricing function $ \phi_{u_s}(y) $ can be given analytically by
		\begin{align}\label{phi_u_s}
			\phi_{u_s}(y) =
			\begin{cases}
			sf'(y) &\text{if } y\in [0,u_s),\\
			\phi_{\mathsf{ivp}}\big(y;u_s\big)  &\text{if } y\in [u_s,\rho_s].
			\end{cases}
		\end{align}
			
		\item \underline{$ \normalfont\textsf{HUC}_2 ${\normalfont :} $ \overline{p}\in  (C_s,+\infty) $}. In this case, the CDT $ u_{\mathsf{cdt}} \in (0, u_s) $, and $ \text{PPM}_{\phi_{u_{\mathsf{cdt}}}}  $ achieves the optimal competitive ratio of $ \frac{s-1}{u_{\mathsf{cdt}}-u_{\mathsf{cdt}}^s} $ if and only if  $ \phi_{u_{\mathsf{cdt}}} $ is given by:
		\begin{align}\label{pricing_function_high}
		\phi_{u_{\mathsf{cdt}}}(y) =
		\begin{cases}
			f'\left(\frac{y}{u_{\mathsf{cdt}}}\right), & \text{if } y\in [0,u_{\mathsf{cdt}}],\\
			\phi_{\mathsf{ivp}}(y;u_{\mathsf{cdt}}),  &\text{if } y\in [u_{\mathsf{cdt}},1).
			\end{cases} 
		\end{align} 
\end{itemize}
\end{theorem}
\begin{proof}
	The optimal pricing functions in the above three cases are derived by solving the corresponding BVPs in Eq. \eqref{BVP_power_LUC} and Eq. \eqref{two_BVP_power}. 
	The details are given in  Appendix \ref{proof_of_optimal_pricing_functions_theorem}.
\end{proof}

For Theorem \ref{optimal_pricing_functions_theorem} we make the following two points. First, the optimal pricing functions in Eq. \eqref{opt_pricing_case_luc} and Eq. \eqref{opt_pricing_case_huc_1} have a separated case when $ y = 0 $.  This is because Eq. \eqref{characteristic_poly_luc} and Eq. \eqref{characteristic_poly_huc_1} are not defined at $ y= 0 $.  However, we can prove that both $ \varphi_{\mathsf{luc}} $ and $\varphi_{\mathsf{huc}} $ approach 0 from the right when $ y\rightarrow 0^+ $, and thus both $ \phi_m(y) $ and $ \phi_u(y) $ are right-differentiable at $ y = 0 $, which is consistent with the ODEs in Eq. \eqref{BVP_power_LUC} and Eq. \eqref{two_BVP_power}. Second, we emphasize that although many  parameters in Theorem \ref{optimal_pricing_functions_theorem} are in analytical forms (e.g., $ u_s, \alpha_s^{\min}$, and $ \phi_{\mathsf{ivp}}\big(y;u\big) $, etc.), numerical computations of $ u_{\mathsf{cdt}}, \varphi_{\textsf{luc}}$,  and $ \varphi_{\textsf{huc}} $ are still needed. In particular, the CDT $ u_{\mathsf{cdt}} $ can be calculated offline, while the computations of  $\varphi_{\textsf{luc}}$ and $ \varphi_{\textsf{huc}} $ must be performed in real-time (i.e., ``on-the-fly"). This should not be a concern for the online implementation of $ \text{PPM}_\phi $ since these computations are light-weight (e.g., all the root-finding can be performed efficiently by bisection searching).

\begin{figure}
	\centering
	\subfigure[\textsf{LUC}: $ \overline{p}\in {(\underline{c},\overline{c}]} $]{\includegraphics[width= 4.7 cm]{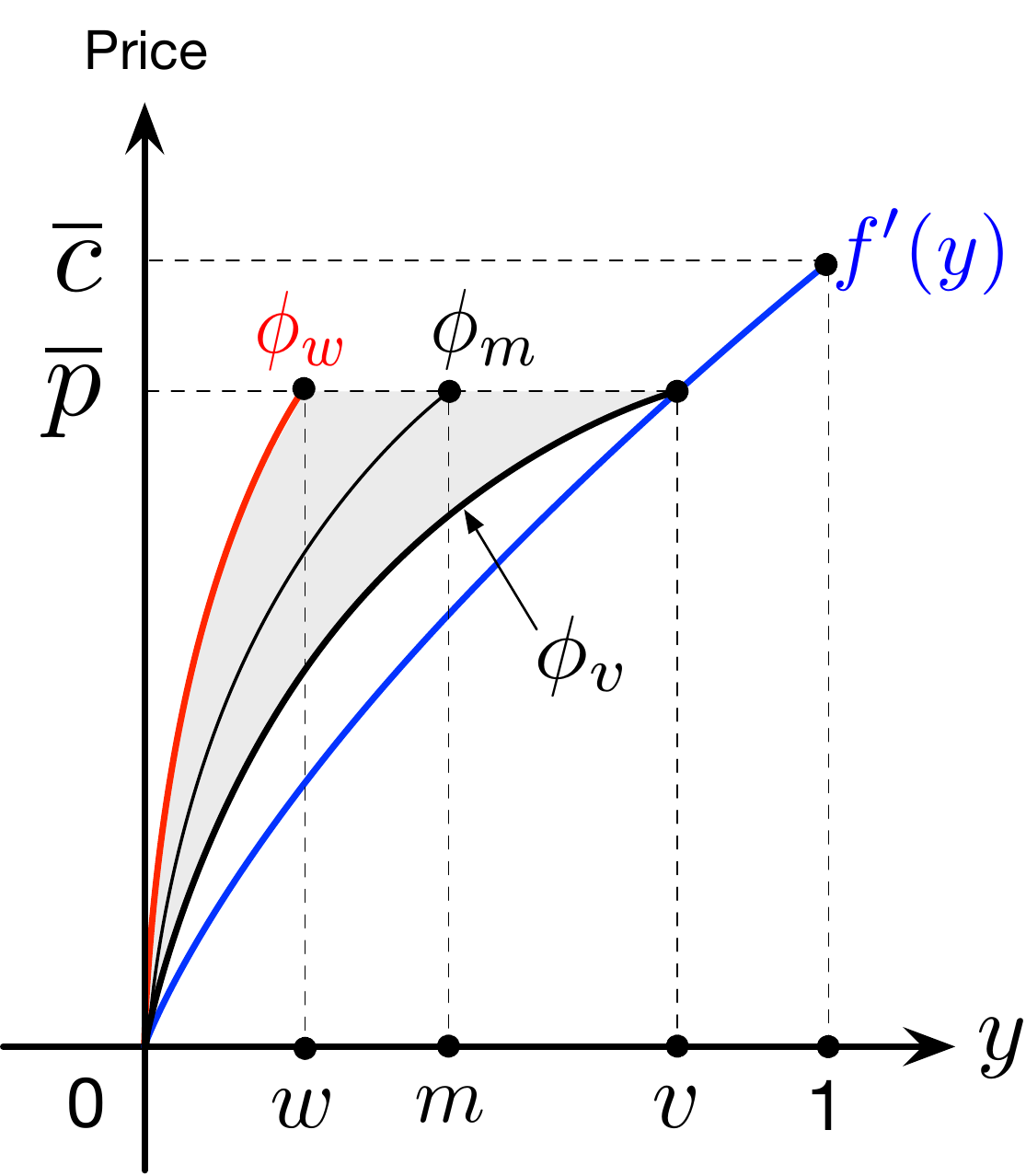}}
	\qquad
	\subfigure[$ \textsf{HUC}_1 $: $ \overline{p}\in {(\overline{c},C_s]} $]{\includegraphics[width= 4.7 cm]{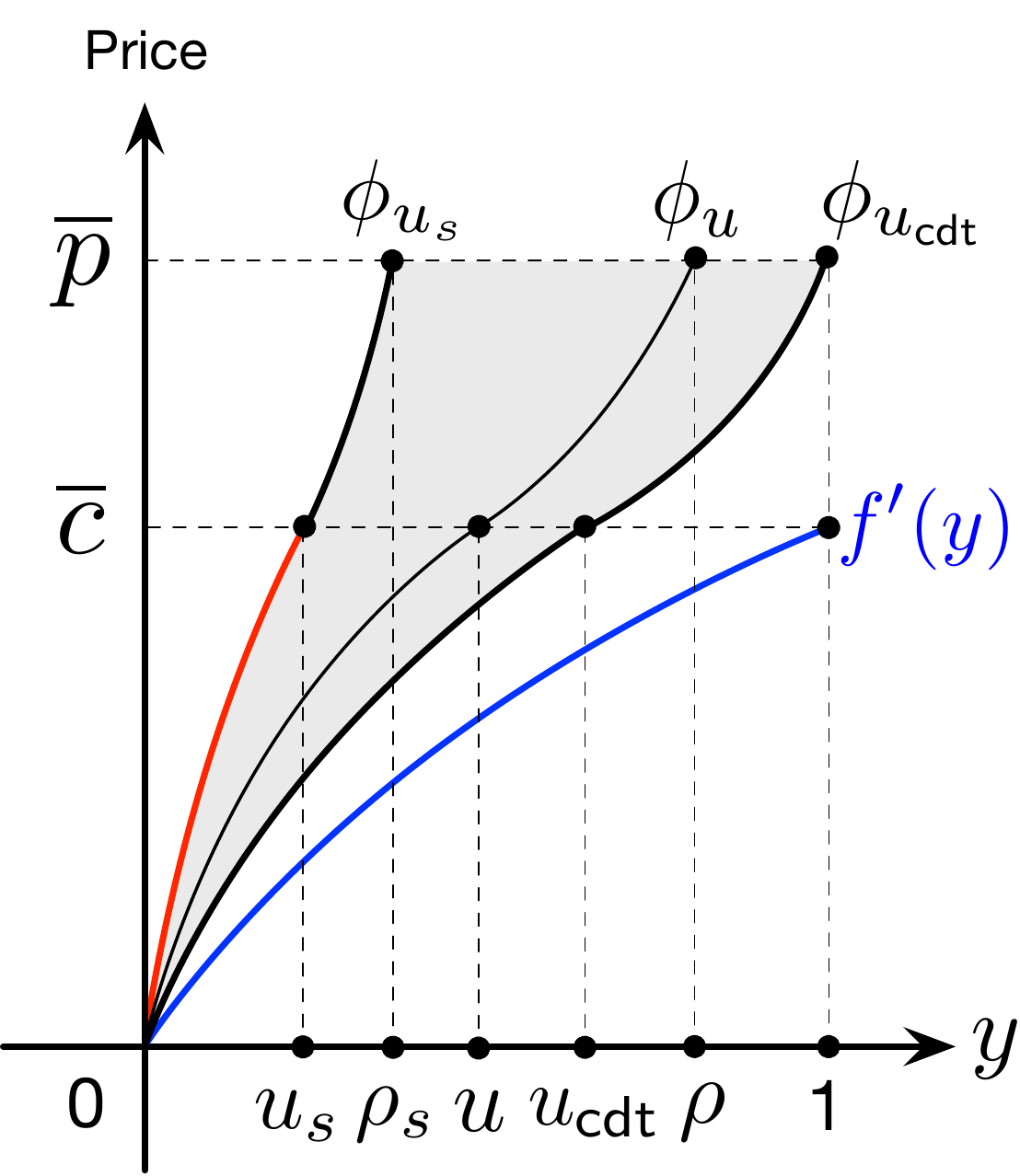}}
	\caption{Illustration of the optimal pricing functions in \textsf{LUC} and $ \textsf{HUC}_1 $. The two red curves represent the same function $ sf'(y)$ but with different domains.}
	\label{optimal_pricing_functions_figure}
\end{figure}

\subsection{Discussion of Structural Properties}
Fig. \ref{optimal_pricing_functions_figure} illustrates the optimal pricing functions for \textsf{LUC} and $ \textsf{HUC}_1 $. We do not illustrate the unique optimal pricing function for $ \textsf{HUC}_2 $ since it is similar to Fig. \ref{three_pricing_schemes}(b).  We discuss several interesting structural properties revealed by Theorem \ref{optimal_pricing_functions_theorem}.

(\textbf{Aggressiveness of Pricing Functions})  In both \textsf{LUC} and $ \textsf{HUC}_1 $, the optimal pricing functions are non-unique, while the optimal pricing function is unique in $ \textsf{HUC}_2 $. In particular, the optimal pricing functions for \textsf{LUC} and $ \textsf{HUC}_1 $ can be represented by two infinite sets  of functions  as follows:
\begin{equation}
\Omega_{\textsf{luc}} = \{\phi_m\}_{\forall m\in [w,v]}, \Omega_{\textsf{huc}_1} =\{\phi_u\}_{\forall u\in [u_s,u_{\mathsf{cdt}}]},
\end{equation}
where $ \phi_m $ and $ \phi_u  $ are given by Eq. \eqref{opt_pricing_case_luc} and Eq. \eqref{opt_pricing_case_huc_1}, respectively. Graphically, these two sets cover the grey area in Fig. \ref{optimal_pricing_functions_figure}. Specifically, as shown in Fig. \ref{optimal_pricing_functions_figure}(a), all the optimal pricing functions in $ \Omega_{\textsf{luc}} $ are lower bounded by $ \phi_v $ and upper bounded by $ \phi_w $. Similarly, in $\normalfont \textsf{HUC}_1 $, all the optimal pricing functions in $ \Omega_{\textsf{huc}_1} $ are lower bounded by $ \phi_{u_{\mathsf{cdt}}} $ and upper bounded by $ \phi_{u_s} $. In  economics, if a pricing scheme `A' sets the price cheaper than pricing scheme `B', then we say pricing scheme `A' is more aggressive than pricing scheme `B' \cite{economics1995}. In this regard, $ \phi_{u_{\mathsf{cdt}}} $ ($ \phi_v $) is the most aggressive optimal pricing function in $ \textsf{HUC}_1 $ (\textsf{LUC}), that is, $ \phi_{u_s} $ ($ \phi_w $) is the most conservative optimal pricing function in $ \textsf{HUC}_1 $ (\textsf{LUC}). Interestingly, the pricing scheme proposed by \cite{Huang2018} for the same setup of power cost functions is $ \phi_w(y) = sf'(y) $ (i.e., the red curves in Fig. \ref{optimal_pricing_functions_figure}), which is only a special case of all the optimal pricing functions characterized in $ \Omega_{\textsf{luc}} $ and $ \Omega_{\textsf{huc}_1} $. Moreover, in $ \textsf{HUC}_2 $, Theorem \ref{optimal_pricing_functions_theorem} shows that the pricing scheme $ \phi_w $ is suboptimal when $ \overline{p} $ is larger than $ C_s $. Therefore, our optimal pricing functions in Theorem \ref{optimal_pricing_functions_theorem} generalize and improve the results in \cite{Huang2018}. 
	
(\textbf{Pricing at Multiple-the-Index}) Note that the pricing function $ \phi_w $ in \textsf{LUC} and the first segment of $ \phi_{u_s} $ in $ \textsf{HUC}_1 $ can be written as 
$ sf'(y)= f'\big(s^{\frac{1}{s-1}}y\big) $,
which uses the marginal cost function $ f' $ to price the resource at $ s^{\frac{1}{s-1}} $-multiple-the-index, and the multiplicative factor $ s^{\frac{1}{s-1}}\in (e, 1) $ when $ s>1 $. In $ \textsf{HUC}_2 $, the optimal pricing function $ \phi_{u_{\mathsf{cdt}}} $ also prices the resources at $ \frac{1}{u_{\mathsf{cdt}}} $-multiple-the-index of $ f'(y) $ when $ y\in [0,u_{\mathsf{cdt}}] $.  The development of such pricing schemes is not entirely new in algorithmic mechanism design. For example, for similar setups of online CAs with supply or production costs (but without capacity limits), the authors of \cite{Blum2011} proposed a pricing scheme called ``twice-the-index" (i.e., $ \phi(y) = f'(2y) $), and the authors of \cite{Huang2018} proposed a more general pricing scheme of $ \phi(y) = f'( \beta y) $ with $ \beta>1 $.  However, to the best of our knowledge, our work here is the first to prove that such pricing schemes are optimal even if capacity limits are present, provided that the multiplicative factors are properly chosen.

\section{Extensions: The General Model}\label{extension}
In this section, we extend our previous results to more general settings of online resource allocation with heterogeneous cost functions and multiple time slots.

\subsection{The General Model} 
We consider the same problem setup as in Section \ref{problem_setup}, but make the following generalizations. First, the cost function for each resource type $ k\in\mathcal{K} $ is denoted by $ f_k $, which can be different among different resource types.  Second, if customer $ n\in\mathcal{N} $ chooses bundle $ b\in\mathcal{B} $, let $ r_k^b(t) $ denote the units of resource type $ k $ owned by customer $ n $ at time slot $ t $, where $ t\in\mathcal{T}_n $ and $ \mathcal{T}_n $ is the duration that customer $ n $ wants to own the resources in bundle $ b $. Suppose bundle $ b $  is denoted by the same vector $ (r_1^b,\cdots,r_K^b) $ as before, then $ r_k^b(t) $ is given by
\begin{align}
r_k^b(t)=
\begin{cases}
r_k^b &\text{ if } t\in\mathcal{T}_n,\\
0     &\text{ if } t\in\mathcal{T}\backslash\mathcal{T}_n,
\end{cases}
\end{align}
where $ \mathcal{T} $ denotes the total time horizon of interest. Based on the above generalizations, our extended model can account for \textit{multi-period online resource allocation with heterogeneous cost functions}. In particular,  the new offline social welfare maximization problem is given by:
\begin{subequations}\label{SWM_general}
	\begin{alignat}{3}
	&  \underset{\mathbf{x,y}}{\mathrm{maximize}}      \qquad\qquad    & & \sum_{n\in\mathcal{N}}\sum_{b\in\mathcal{B}}v_n^b x_n^b - \sum_{k\in\mathcal{K}}\sum_{t\in\mathcal{T}}f_k\big(y_k(t)\big)\\
	&\mathrm{subject\ to}\ & &   \sum_{n\in\mathcal{N}}\sum_{b\in\mathcal{B}}r_k^{b}(t) x_n^b =  y_k(t), \forall k,t, & &  \label{totald_load}\\
	&   & & \sum_{b\in\mathcal{B}}x_n^b \leq  1,\forall n, & &  \label{bidd_selection}\\
	&    & & 0 \leq y_k(t) \leq 1, \forall k,t,\\
	&   & & x_n^b \in \{0,1\},\forall n, b,\label{power_rating}
	\end{alignat}
\end{subequations}
where $ y_k(t) $ is the utilization of resource type $ k $ at time $ t $.

\subsection{Generalization of Theorem \ref{a_unified_BVP}}\label{sec_Redefinition} 
To generalize Theorem \ref{a_unified_BVP} to account for the above resource allocation model, we first need to redefine some key parameters as follows. We assume that
$
\max_{ n\in\mathcal{N},b\in\mathcal{B},r_k^b\neq 0}\ \{\frac{v_n^b}{|\mathcal{T}_n|\cdot r_k^b}\} \leq \overline{p}_k,
\underline{c}_k \triangleq f_k'(0), \text{ and } \overline{c}_k \triangleq f_k'(1), \forall k\in\mathcal{K}
$, 
where $ \overline{p}_k $, $ \underline{c}_k  $ and $ \overline{c}_k $  correspond to $ \overline{p} $, $ \underline{c}$ and $ \overline{c} $ in Section \ref{assumptions}, respectively. Here, we have an upper bound $ \overline{p}_k  $, a minimum marginal cost $ \underline{c}_k  $,  and a maximum marginal cost $ \overline{c}_k  $ for each  $ k\in\mathcal{K} $.  In particular, $ \overline{p}_k $ can be interpreted as the maximum price customers are willing to pay for purchasing a single unit of resource type $ k $ for each time slot. 

Below we give a general version of Theorem \ref{a_unified_BVP}. Specifically, we focus on the case of \textsf{HUC} only (i.e., $ \overline{p}_k> \overline{c}_k  $). The case of \textsf{LUC} (i.e., $ \overline{p}_k\leq  \overline{c}_k  $) is similar and is omitted for brevity.

\begin{theorem} \label{a_unified_BVP_general}
	For any $ k\in\mathcal{K} $, if $ f_k\in\mathcal{F} $ and the upper bound $ \overline{p}_k\in (\overline{c}_k,+\infty) $, then we have:
	\begin{itemize}
		\item {\normalfont\textsf{Sufficiency}}. For any given $ \alpha_k\geq 1 $, if $ \phi_k(y) $ is a solution to the following two first-order BVPs simultaneously:
		\begin{subequations}
			\begin{align}
			&\begin{cases}
			\phi_k'(y)  =  \alpha_k \cdot \frac{\phi_k (y) - f_k'(y)}{f_k'^{-1}(\phi_k(y))}, y\in (0,u_k), \\
			\phi_k(0) = \underline{c}_k, 
			\phi_k(u_k) = \overline{c}_k.
			\end{cases}\label{general_BVP_1}
			\\
			&\begin{cases}
			\phi_k'(y) =  \alpha_k\cdot \left(\phi_k(y) - f_k'(y)\right), y\in (u_k,1), \\
			\phi_k(u_k) = \overline{c}_k, 
			\phi_k(1) \geq \overline{p}_k,
			\end{cases}\label{general_BVP_2}
			\end{align}
		\end{subequations}
		where $ u_k\in (0,1) $ is the dividing threshold of $ \phi_k $, 
		then $ \text{PPM}_{\bm{\phi}}  $ is $ \max_{k\in\mathcal{K}} \{\alpha_k\} $-competitive.
		
		\item {\normalfont\textsf{Necessity}}. If there is an $ \alpha $-competitive online algorithm, then for all $ k\in\mathcal{K} $, there must exist a dividing threshold $ u_k\in (0,1) $ and a strictly-increasing pricing function $ \phi_k(y) $ such that $ \phi_k(y) $ satisfies Problem \eqref{general_BVP_1} and Problem \eqref{general_BVP_2} with a feasible competitive ratio parameter $ \alpha_k\in [1,\alpha] $. 
	\end{itemize} 
\end{theorem}

The proof of Theorem \ref{a_unified_BVP_general} is similar to that of Theorem \ref{a_unified_BVP}, and the details are given in  Appendix \ref{proof_a_unified_BVP_general}. Based on the two BVPs in Theorem \ref{a_unified_BVP_general}, for each resource type $ k\in\mathcal{K} $, we can define the minimum competitive ratio parameter $ 
\alpha_k^* $ in a similar way as Proposition \ref{optimality}. The final competitive ratio is then given by 
$ \alpha_*(\mathcal{S}) = \max\limits_{k\in\mathcal{K}}\  \{\alpha_k^*\} $. 
We can also define the lower bound of $ \alpha_k $ according to Definition \ref{def_achievable_region}. The principles in Algorithm \ref{principles_optimal_design} can thus be applied for  characterizing the competitive ratios and the corresponding pricing functions in the general case. Meanwhile, our analytical results for the setup with power cost functions also hold with some slight modifications. The details are omitted for brevity.

\section{Empirical Evaluation}
In this section we evaluate the performance of our designed online mechanism via extensive empirical experiments of online job scheduling in cloud computing. 

\subsection{Simulation Setup}
(\textbf{Supply Costs})
We consider two types of resources ($ K=2$), namely, CPU and RAM. We use the traces of one-month computing tasks in a Google cluster \cite{google_trace_analysis}. We assume each bundle $ b\in\mathcal{B} $ is given by $ (r_{\mathsf{cpu}}^b, r_{\mathsf{ram}}^b) $, where $ r_{\mathsf{cpu}}^b $ and $ r_{\mathsf{ram}}^b $ can be any value in $ \{0.001, 0.003, 0.005\} $ units of the total normalized capacity 1. Therefore,  in total we have $|\mathcal{B}| =  9 $ bundles. We assume $ T = 3600 $ time slots and each time slot is 10 seconds. The cost functions for CPU and RAM are given by $ f_{\mathsf{cpu}}(y) = a_{\mathsf{cpu}}y^{s_{\mathsf{cpu}}} $ and $ f_{\mathsf{ram}}(y) = a_{\mathsf{ram}}y^{s_{\mathsf{ram}}} $, respectively. Following \cite{Adam_speed_scaling, speed_scaling, XZhang_ToN}, we assume $ s_{\mathsf{cpu}} = 3 $  and $ s_{\mathsf{ram}} = 1.2 $. We set up the coefficients $ (a_{\mathsf{cpu}},a_{\mathsf{ram}}) = (0.223,8.38\times 10^{-6})  $ by keeping the ratio of $ a_{\mathsf{cpu}}/a_{\mathsf{ram}} $  based on  \cite{power_models}, where the dominate power consumption is from CPU. This setup of cost functions follows the typical power consumption models of data centers \cite{data_center_power}. 
The minimum marginal costs are zero and the maximum marginal costs are given by $ \overline{c}_{\mathsf{cpu}} \approx 0.67 $ and $ \overline{c}_{\mathsf{ram}} \approx  1.01\times10^{-5}  $. Since $ \overline{c}_{\mathsf{ram}} $ is much smaller than $ \overline{c}_{\mathsf{cpu}} $, our simulation mainly focuses  on the power costs of CPU consumptions. For simplicity, we write $ \overline{c}_{\mathsf{cpu}} = 0.67 $ hereinafter without the approximation sign.

(\textbf{Job Arrivals}) 
We consider the total number of jobs is $ N = 4000 $. The arrival time and duration of each job follow the job arrival and departure times in Google cluster trace \cite{google_trace_analysis}. For job $ n $, the valuation $ v_n^b $ is given by $ v_n^b =  p |\mathcal{T}_n|r_{\mathsf{cpu}}^b
$, where $ |\mathcal{T}_n| $  denotes the duration of job $ n $ and $  p $ is a random variable constructed as follows: 
\begin{enumerate}
	\item \textsf{Uniform-Exact Case (Case-UE)}. The sequences of $ p $ are uniformly distributed within $ [0, \overline{p}] $ and the pricing functions are designed based on the exact value of $ \overline{p} $.
	
	\item \textsf{Extreme-Exact Case (Case-EE)}. This extreme case evaluates the performance robustness of online mechanisms. For the first-half of the total jobs, the sequences of $ p $ are uniformly distributed within $ [0, \frac{\overline{p}}{2}] $. While for the second-half, the sequences of $ p $ are uniformly distributed  within $ [\frac{\overline{p}}{2}, \overline{p}] $.  Meanwhile, the pricing functions are designed based on the exact value of $ \overline{p} $. 
	
	\item \textsf{Uniform-Inexact Case (Case-UI)}. The sequences of $ p $ are  uniformly distributed within $ [0, \overline{p}] $. However, the pricing function is designed based on the estimated upper bound $ \overline{p}_{\mathsf{estimate}} = \overline{p}(1+\delta) $, where $ \delta\in [-0.8,2.4] $, meaning that $ \overline{p} $ can be underestimated (overestimated) for as much as 80\% (240\%). We use this case to evaluate the impact of underestimations/overestimations of $ \overline{p} $ on the performances of different online mechanisms.
	
	\item \textsf{Extreme-Inexact Case (Case-EI)}. This is a mixture of the second and third case. Specifically, the sequences of $ p $  are  generated in the same way as those in \textsf{Case-EE}, and $ \overline{p}_{\mathsf{estimate}} $ follows the same setup as \textsf{Case-UI}.  
\end{enumerate}

(\textbf{Performance Metrics})
Given any arrival instance $ \mathcal{A} $, we define the empirical ratio (ER)  by
\begin{equation*}
\text{ER}(\mathcal{A}) \triangleq \frac{W_{\mathsf{opt}}(\mathcal{A})}{W_{\mathsf{online}}(\mathcal{A})},
\end{equation*}
where $ W_{\mathsf{opt}}(\mathcal{A}) $ is the optimal objective of Problem \eqref{SWM}. For each sample of $ \mathcal{A} $, we  solve Problem \eqref{SWM} by Gurobi 8.1 via its Python API\footnote{http://www.gurobi.com}, and then evaluate ERs over 1000 samples of $ \mathcal{A} $'s to get the average ER of each online mechanism. 

(\textbf{Benchmarks}) We refer to  our proposed PPM with optimal pricing as PPM-OP, and compare it with the offline benchmark and two existing PPMs as follows:
\begin{itemize}
	\item PPM with Twice-the-index Pricing (PPM-TP). This PPM is first proposed in \cite{Blum2011} and later extended for cloud resource allocation problems in \cite{XZhang_ToN}. By PPM-TP, when $ y\in [0,0.5] $, the pricing function is $ \phi(y) = f'(2y) $; when $ y\in (0.5,1] $, the pricing function is exponential and the detailed expression is referred to \cite{XZhang_ToN}.
	\item PPM with Myopic Pricing (PPM-MP). This PPM  prices the resources based on the current marginal costs, i.e., $\phi(y) =  f'(y) $, and thus is myopic in the sense that the resources will be allocated aggressively without reservation for potential high-PUV customers in the future.
\end{itemize}

For any given resource utilization level $ y\in (0,1)$, PPM-TP always has the highest posted prices and PPM-MP always has the cheapest ones. Therefore, among the three online mechanisms, PPM-TP (PPM-MP) is the most conservative (aggressive) one\footnote{Based on \eqref{phi_w}, the most conservative optimal pricing function is $ \phi_w(y) = sf'(y)$, which is still more aggressive than $ f'(2y) = 2^s f'(y) $ when $ s >1  $.}.

\subsection{Numerical Results}
Fig. \ref{normal_case} compares the different online mechanisms in \textsf{Case-UE}. As shown in Fig. \ref{normal_case}(a), $ \overline{p} $ varies within $ [\overline{c}_{\mathsf{cpu}}, 9\overline{c}_{\mathsf{cpu}}] $, where $  \overline{c}_{\mathsf{cpu}} = 0.67 $ and $ 9\overline{c}_{\mathsf{cpu}} = 6.03 $. Note that  based on Eq. \eqref{C_s}, we have $ C_s \approx 4.21 \approx 6.28\overline{c}_{\mathsf{cpu}}  $, and thus the setup of $ \overline{p} \in [\overline{c}_{\mathsf{cpu}}, 9\overline{c}_{\mathsf{cpu}}]$  in Fig. \ref{normal_case}(a) covers all the cases of $\textsf{LUC}$, $ \textsf{HUC}_1 $,  and $ \textsf{HUC}_2 $.  We can see that the ERs of our proposed PPM-OP are roughly around $1.12 \sim 1.22$, which strictly outperforms both PPM-TP and PPM-MP. An interesting result revealed by Fig. \ref{normal_case}(a) is that the ER performance of PPM-OP (PPM-TP) first improves (deteriorates) and then deteriorates (improves) when $ \overline{p} $ increases within $ [\overline{c}_{\mathsf{cpu}}, 9\overline{c}_{\mathsf{cpu}}] $. We argue that the  ER behaviours of PPM-OP for $ \overline{p}\in [\overline{c}_{\mathsf{cpu}}, 6\overline{c}_{\mathsf{cpu}}] $ are reasonable although the optimal competitive ratios are  the same when $ \overline{p}\in [\overline{c}_{\mathsf{cpu}}, 6\overline{c}_{\mathsf{cpu}}]\subset [\overline{c}_{\mathsf{cpu}}, C_s] $. The insight is that  when $ \overline{p} $ slightly increases from $ \overline{c}_{\mathsf{cpu}} $ to $ 3\overline{c}_{\mathsf{cpu}} $, the uncertainty level of the arrival instances also slightly increases, and this is beneficial for the online posted-price control since whatever  decisions made now may have remedies in the future. However,  when $ \overline{p}>3\overline{c}_{\mathsf{cpu}} $, the ER performance of PPM-OP becomes worse whenever $ \overline{p} $ increases. This is because the uncertainty level of the arrival instances is too high so that it becomes challenging to perform online posted-price control without future information. The differences of the three online mechanisms can also be seen by their total CPU resource utilizations in Fig. \ref{normal_case}(b). PPM-MP is the most aggressive and thus the total capacity is quickly depleted (i.e., 100\% utilization). PPM-TP is the most conservative and reserves over 40\% capacity for future jobs. The total CPU resource utilization of PPM-OP (around 85\% maximum utilization) stays between those of PPM-MP and PPM-TP, and achieves a better balance between aggressiveness and conservativeness. 

\begin{figure}
	\centering
	\subfigure[Empirical Ratios]{\includegraphics[height=4.7cm]{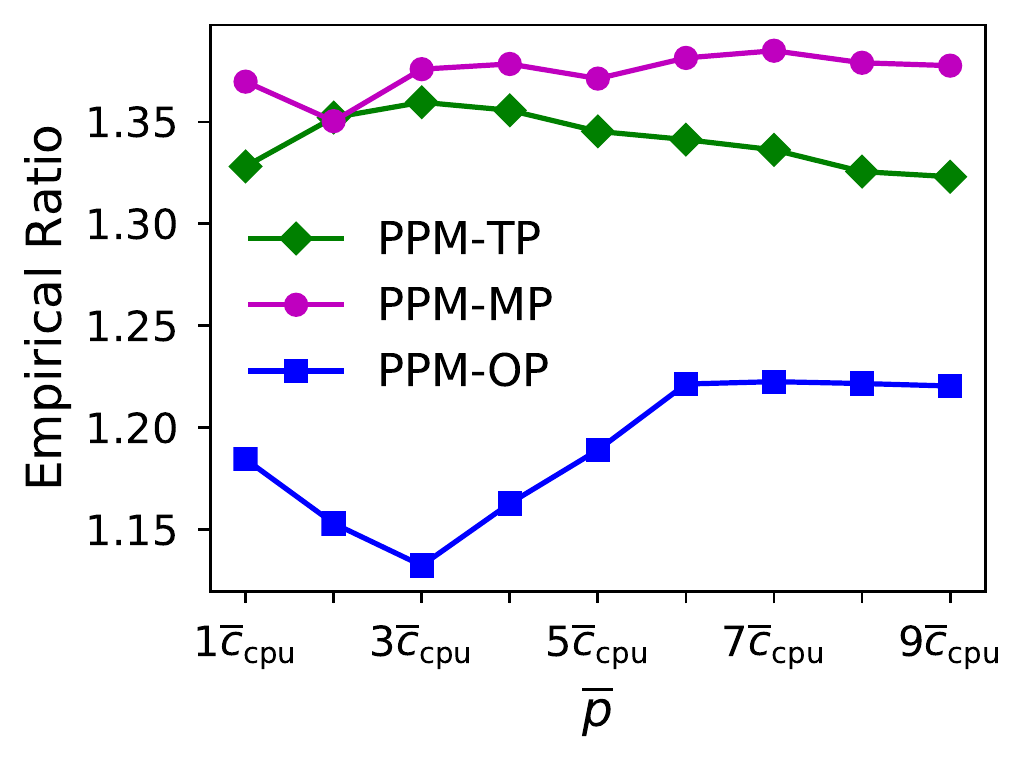} }
	\subfigure[Total Resource Utilization]{\includegraphics[height=4.7cm]{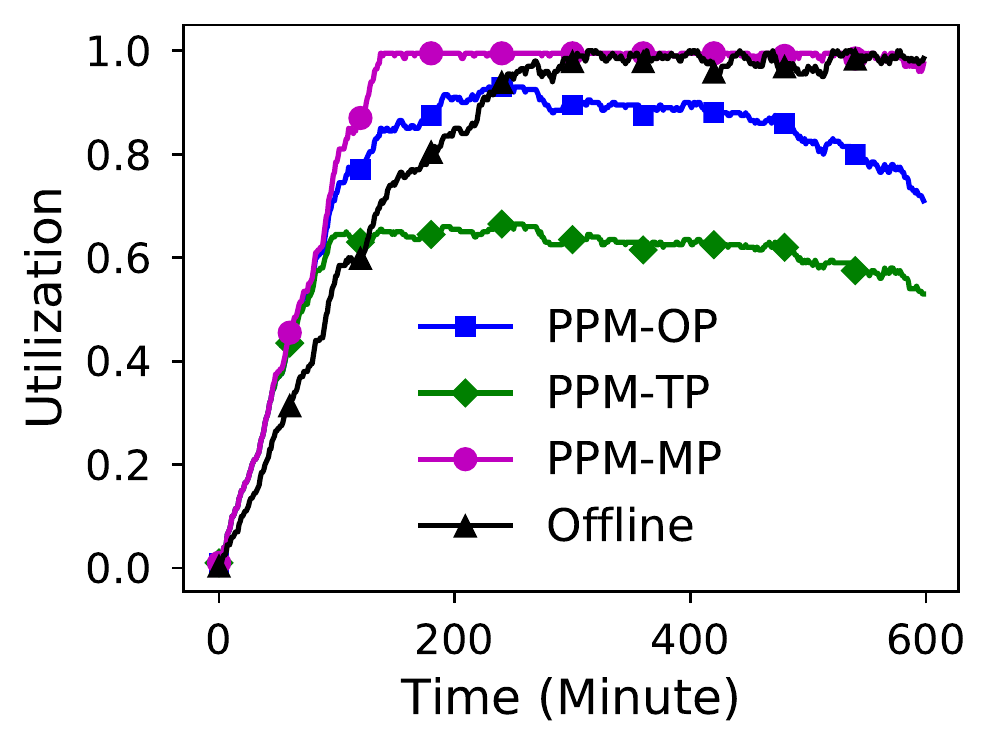}}
	\caption{ERs and total resource utilizations of different online mechanisms in \textsf{Case-UE}. Each point in the left figure is an average of 1000 instances. The right figure is for one instance of $ \overline{p} = 2 \overline{c}_{\mathsf{cpu}} = 1.34 $.}
	\label{normal_case}
\end{figure}

\begin{figure}
	\centering
	\subfigure[Empirical Ratios]{\includegraphics[height=4.7cm]{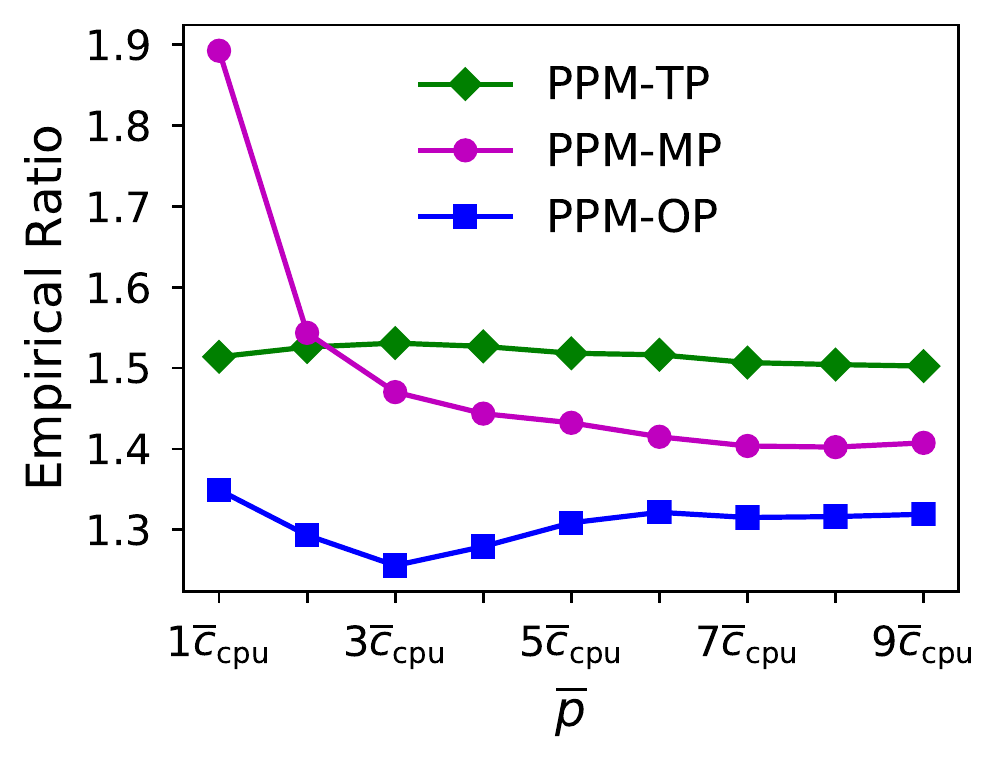} }
	\subfigure[Total Resource Utilization]{\includegraphics[height=4.7cm]{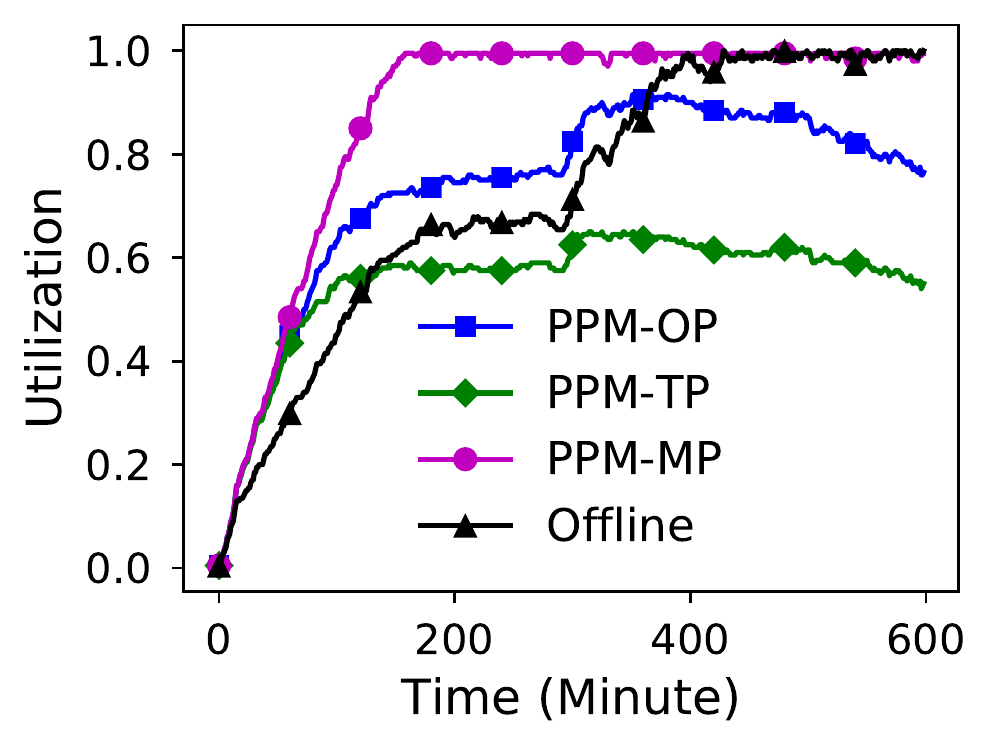}}
	\caption{ERs and total resource utilizations of different online mechanisms in \textsf{Case-EE}. Each point in the left figure is an average of 1000 instances. The right figure is for one instance of $ \overline{p} = 2 \overline{c}_{\mathsf{cpu}} = 1.34 $.}
	\label{low_high_case}
\end{figure}

Fig. \ref{low_high_case} shows the ERs of online mechanisms in \textsf{Case-EE}. The first result revealed by Fig. \ref{low_high_case}(a) is intuitive, namely, the ERs of all the three online mechanisms are worse than the ERs in \textsf{Case-UE}.  Second, our proposed PPM-OP achieves a very competitive performance even in this extreme case: the ERs of PPM-OP are always below 1.4, which outperforms PPM-TP by more than 15\% in average. Third, Fig. \ref{low_high_case}(a) also shows that the greedy mechanism PPM-MP is significantly worse than both PPM-TP and PPM-OP when $ \overline{p} $ is small, but outperforms PPM-TP when $ \overline{p} $ is large.  However, due to the greedy nature of PPM-MP,  the ERs of PPM-MP are considerably less robust than those of PPM-OP and PPM-TP, as illustrated in Fig. \ref{low_high_case}(a).  Fig. \ref{low_high_case}(b) shows the total CPU resource utilizations of different mechanisms when $ \overline{p} = 2 \overline{c}_{\mathsf{cpu}}  $. Since in \textsf{Case-EE} the first-half (second-half) of the total jobs have low (high) PUVs, the total CPU resource utilization profile of the offline benchmark depicts two distinct levels within the duration of $ t\in [0,300] $ and $ t\in [300,600] $. We can see that PPM-MP completely fails to achieve such a two-level utilization profile by quickly reaching the capacity limit before $ t = 200$ min; PPM-TP performs better than PPM-MP, but reserves too much available capacity for future jobs (too conservative). In comparison, PPM-OP shows the capability of distinguishing the two different intervals, and has a similar utilization profile to that of the offline benchmark.

\begin{figure}
	\centering
	\subfigure[\textsf{Case-UI} ($ \textsf{LUC} $: $\overline{p} = \overline{c}_{\mathsf{cpu}}$)]{\includegraphics[width = 5 cm]{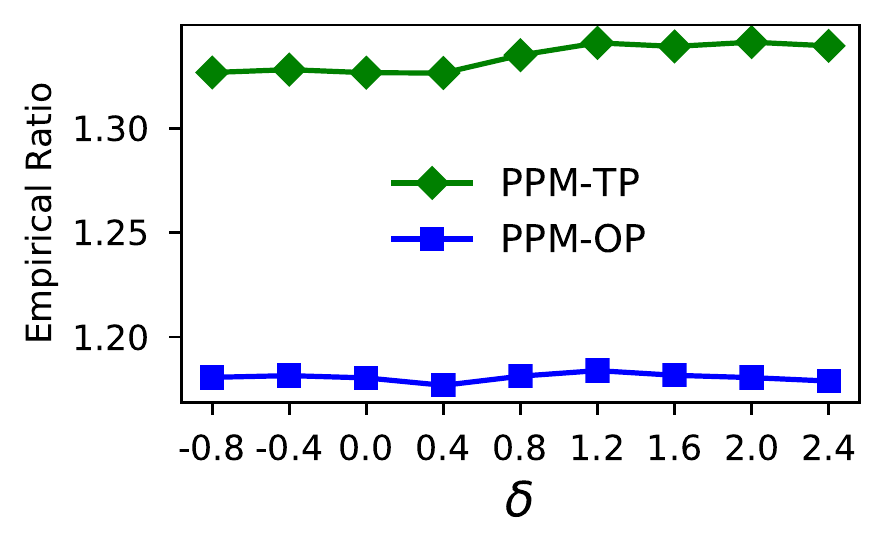}}
	\subfigure[\textsf{Case-EI} ($ \textsf{LUC} $: $\overline{p} = \overline{c}_{\mathsf{cpu}}$)]{\includegraphics[width = 5 cm]{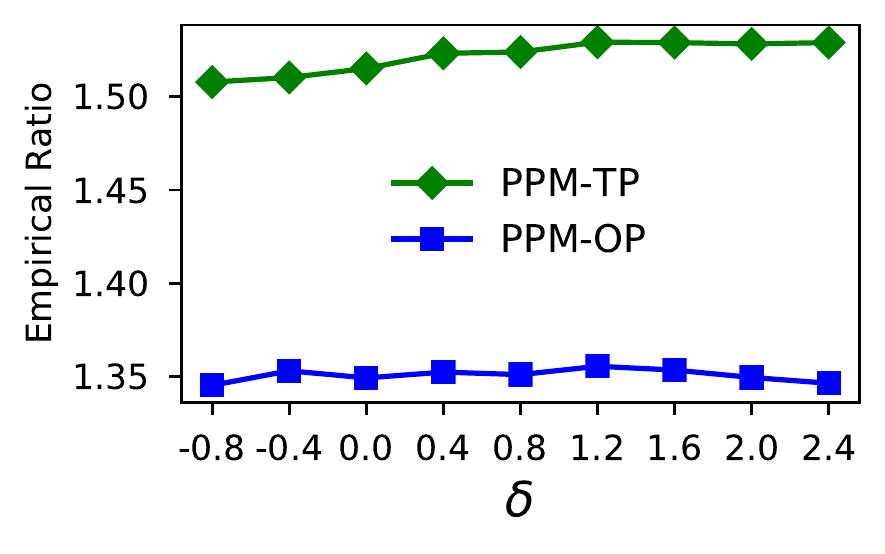}}
	\subfigure[\textsf{Case-UI} ($ \textsf{HUC}_1 $: $\overline{p} = 3\overline{c}_{\mathsf{cpu}}$)]{\includegraphics[width = 5 cm]{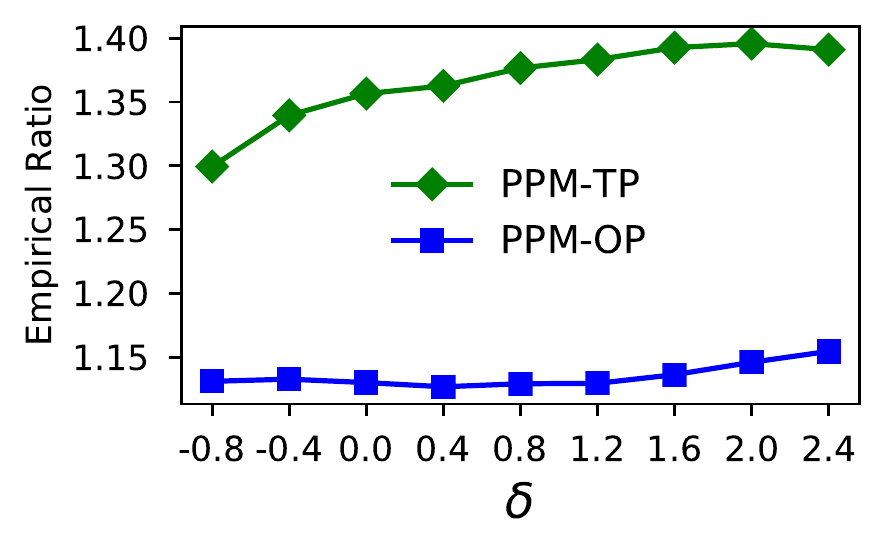}}
	\subfigure[\textsf{Case-EI} ($ \textsf{HUC}_1 $: $\overline{p} = 3\overline{c}_{\mathsf{cpu}}$)]{\includegraphics[width = 5 cm]{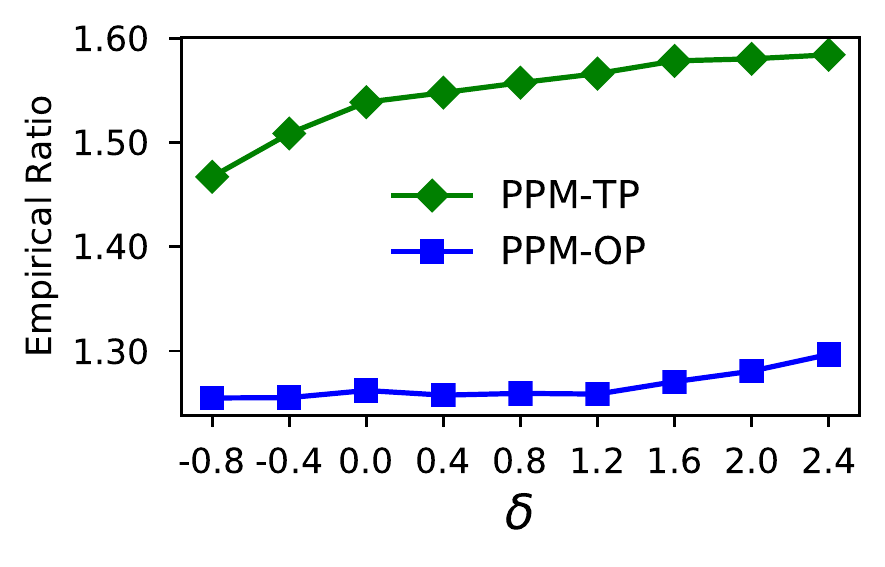}}
	\subfigure[\textsf{Case-UI} ($ \textsf{HUC}_2 $: $\overline{p} = 9\overline{c}_{\mathsf{cpu}}$)]{\includegraphics[width = 5 cm]{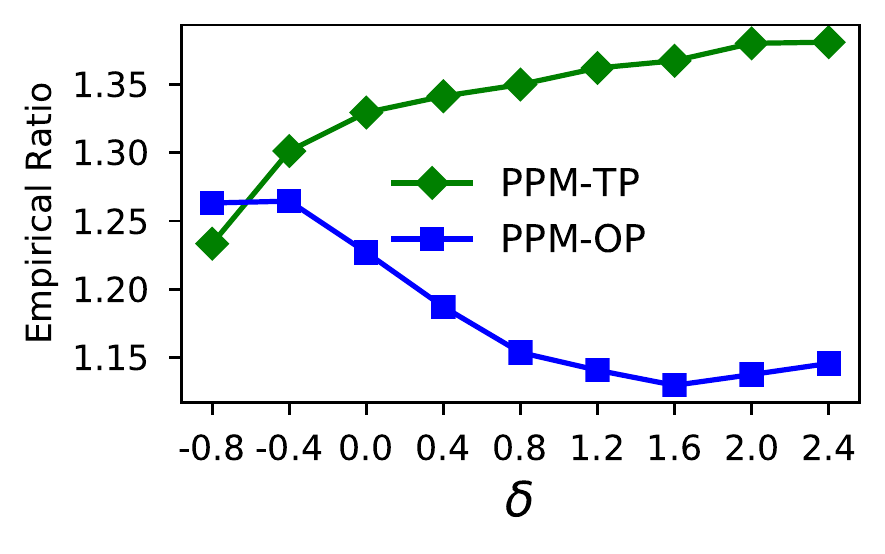}}
	\subfigure[\textsf{Case-EI} ($ \textsf{HUC}_2 $: $\overline{p} = 9\overline{c}_{\mathsf{cpu}}$)]{\includegraphics[width = 5 cm]{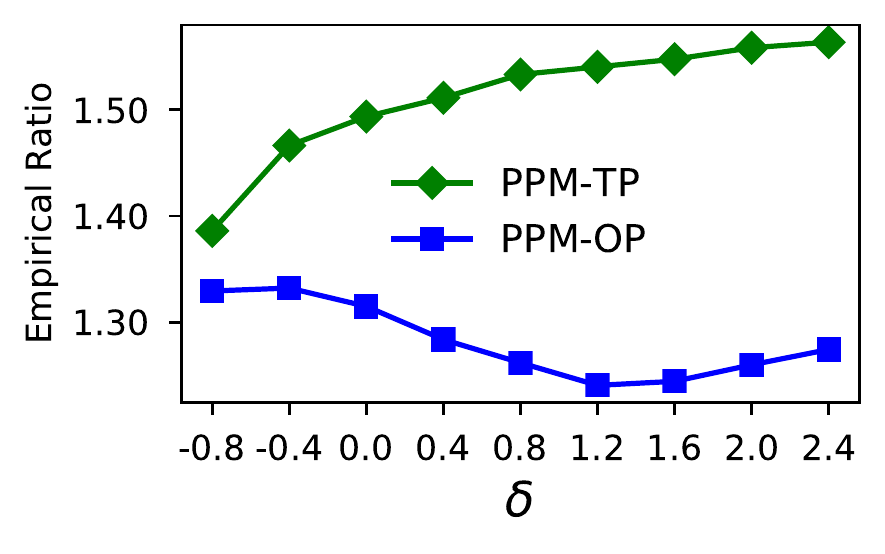}}
	\caption{Comparison between PPM-OP and PPM-TP when the estimated upper bound $ \overline{p}_{\mathsf{estimate}} $ is inexact, where  $ \overline{p}_{\mathsf{estimate}} = \overline{p}(1+\delta) $ and $ \overline{p}$ denotes the real upper bound. Each point in the figure is an average of 1000 instances.}
	\label{ER_inexact_case}
\end{figure}

We next demonstrate the impact of inexact estimations of $ \overline{p} $ on the ER performances of PPM-OP and PPM-TP (note that the performance of PPM-MP is independent of $ \overline{p} $). We perform an indepth comparison between PPM-OP and PPM-TP in both \textsf{Case-UI} and \textsf{Case-EI} with $ \overline{p} = \overline{c}_{\mathsf{cpu}} \in (0,\overline{c}_{\mathsf{cpu}}] $ (i.e., \textsf{LUC}), $ \overline{p} =  3\overline{c}_{\mathsf{cpu}}\in   (\overline{c}_{\mathsf{cpu}}, C_s] $ (i.e., $ \textsf{HUC}_1 $), and $ \overline{p} = 9\overline{c}_{\mathsf{cpu}} \in (C_s,+\infty)$ (i.e., $ \textsf{HUC}_2 $), where  $ C_s \approx 6.28\overline{c}_{\mathsf{cpu}} $. Hence, we have six cases in total, which correspond to the six sub-figures in Fig. \ref{ER_inexact_case}.  
We note that the choices of $ \overline{p} = 3 \overline{c}_{\mathsf{cpu}}  $  and $ \overline{p} = 9\overline{c}_{\mathsf{cpu}} $ have no specific reasons other than making them in $ \textsf{HUC}_1 $ and $ \textsf{HUC}_2 $, respectively.
\begin{itemize}
	\item Fig. \ref{ER_inexact_case}(a) and Fig. \ref{ER_inexact_case}(b) show that the ER performances of both PPM-OP and PPM-TP are insensitive to $ \delta $ in \textsf{LUC}. The insensitivity of PPM-TP is reasonable since the first segment of the pricing function of PPM-TP, i.e., $ \phi(y) = f'(2y) $, is independent of $ \overline{p} $. Therefore, when $ \overline{p} =  \overline{c}_{\mathsf{cpu}} $, the highest resource utilization level will not significantly exceed 50\% of the total capacity (since $ \overline{p} \leq \overline{c}_{\mathsf{cpu}} = f'(2*0.5) $). As a result, the first segment of the pricing function of PPM-TP is the major active part for most of the time slots. 
	Meanwhile, it is also not surprising that PPM-OP is insensitive to $ \delta $ in \textsf{LUC} since $ \overline{p}_{\mathsf{estimation}} $ does not influence PPM-OP when $ \overline{p}_{\mathsf{estimation}} \leq C_s \approx 6.28\overline{c}_{\mathsf{cpu}} $. 
	
	\item Fig. \ref{ER_inexact_case}(c) and Fig. \ref{ER_inexact_case}(d) show that the ER performance of PPM-TP always deteriorates with the increase of $ \delta $ in $ \textsf{HUC}_1 $ (underestimation is always better than overestimation). The ER behaviors of PPM-TP are interesting but quite reasonable since an overestimation of $ \overline{p} $ will make the second segment of the pricing function of PPM-TP over conservative, leading to a worse ER performance. Similar results have also been reported by  \cite{XZhang_ToN}.  Unlike PPM-TP, PPM-OP is insensitive to the estimation error $ \delta $ when $ \delta < C_s/\overline{p} - 1 \approx 1.1  $, meaning that as long as  the overestimation of $ \overline{p} $ does not change the design of optimal pricing functions from $ \textsf{HUC}_1 $ to  $ \textsf{HUC}_2 $, the ERs of PPM-OP will be the same. However, a larger estimation error $ \delta > 1.1 $ will slightly worsen the ER performance of PPM-OP as the optimal pricing function in $ \textsf{HUC}_2 $ is too conservative in $ \textsf{HUC}_1 $.
	
	\item Fig. \ref{ER_inexact_case}(e) and Fig. \ref{ER_inexact_case}(f) show that the ER performances of PPM-TP and PPM-OP have opposite behaviors w.r.t. the estimation error $ \delta $ in $ \textsf{HUC}_2 $. Specifically, overestimations of $ \overline{p} $ still increase the ERs of PPM-TP, similar to the results in $ \textsf{HUC}_1 $. In contrast, PPM-OP will benefit from overestimating $ \overline{p} $ when $ \delta  $ is within a certain range (e.g., when $ \delta\in (0,1.6) $ in Fig. \ref{ER_inexact_case}(e)), and then deteriorate when the estimation error $ \delta $ is too large (e.g., when $ \delta>1.6 $ in Fig. \ref{ER_inexact_case}(e)). Note that the ER behaviors of PPM-OP are very counter-intuitive since an overestimation of $ \overline{p} $ in $ \textsf{HUC}_2 $ will inevitably make the optimal pricing functions in PPM-OP more conservative, which intuitively should lead to a worse ER performance. However, Fig. \ref{ER_inexact_case}(e) and Fig. \ref{ER_inexact_case}(f) show that, the ER performance of PPM-OP will deteriorate only if the overestimation of $ \overline{p} $ exceeds some threshold (e.g., $ 1.6 $ in Fig. \ref{ER_inexact_case}(e) and $ 1.2 $ in Fig. \ref{ER_inexact_case}(f)).   
\end{itemize}

The above illustrations indicate that underestimations of $ \overline{p} $ should always be avoided when using our proposed PPM-OP. This is because a negative $ \delta $ either has no impact on the ER performance of PPM-OP in \textsf{LUC} and $ \textsf{HUC}_1 $ (the first four sub-figures in Fig. \ref{ER_inexact_case}), or makes it even worse in $ \textsf{HUC}_2 $ (the final two sub-figures in Fig. \ref{ER_inexact_case}). Meanwhile, it is generally beneficial to slightly overestimate $ \overline{p} $ when $ \overline{p} $ is larger than $ C_s $. 

\begin{figure}
	\centering
	\subfigure[$ \mathbb{N}(\overline{c}_{\mathsf{cpu}}, 1, 0, \overline{p}) $]{\includegraphics[width=4cm]{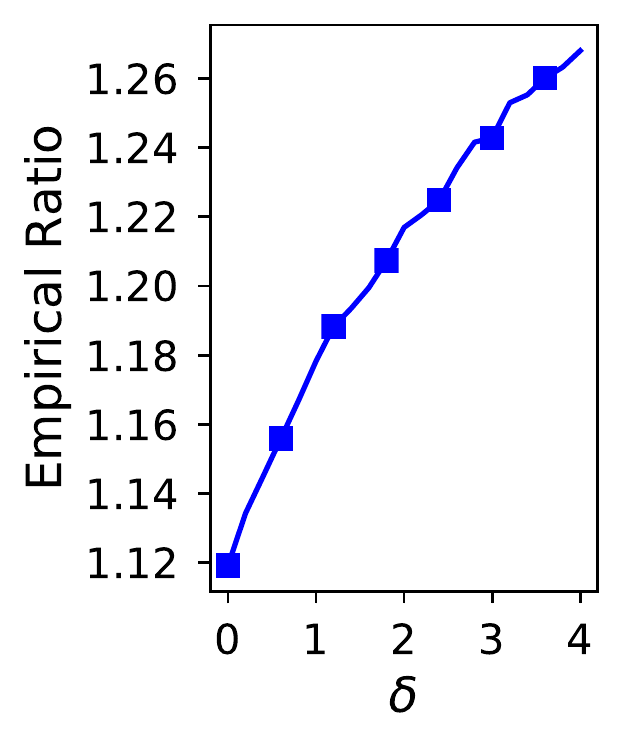}}
	\subfigure[$ \mathbb{N}(\overline{c}_{\mathsf{cpu}}, 4, 0, \overline{p}) $]{\includegraphics[width=4cm]{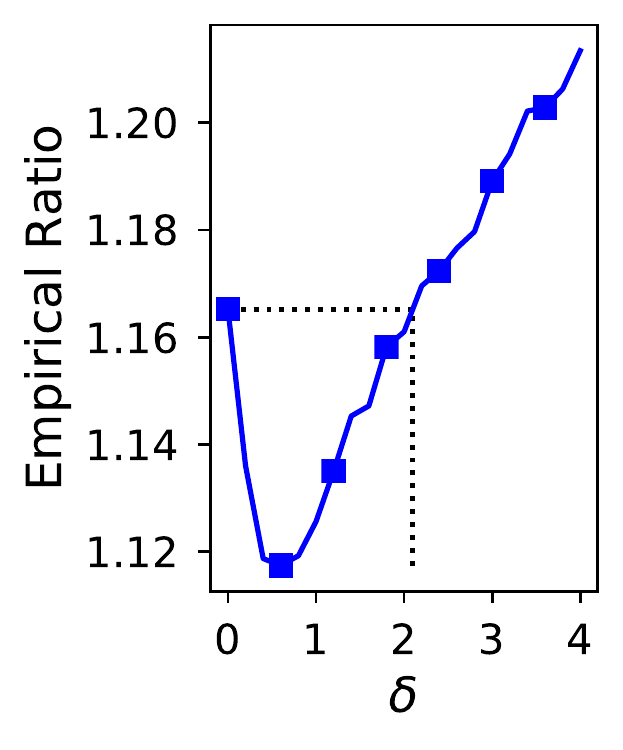}}
	\subfigure[$ \mathbb{N}(\overline{c}_{\mathsf{cpu}}, 100, 0, \overline{p}) $]{\includegraphics[width=4cm]{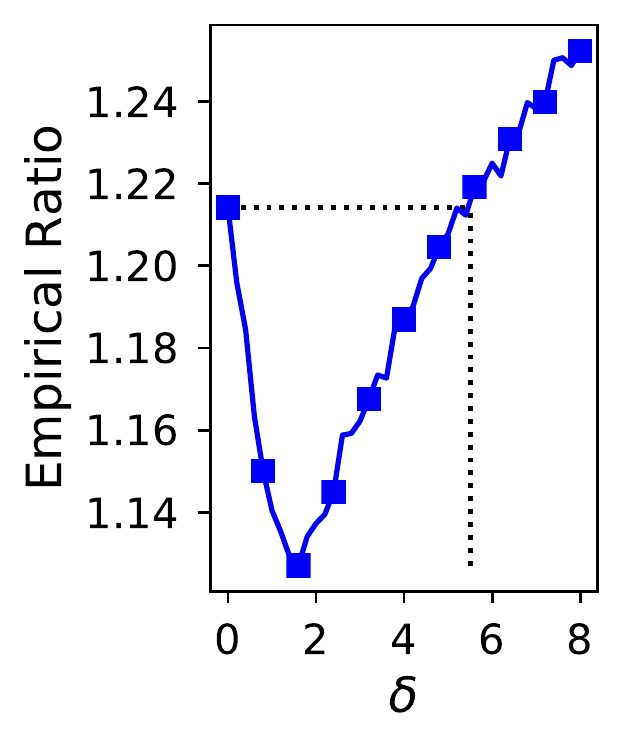}}
	\caption{Impact of overestimations of $ \overline{p} $ on the ER performance of PPM-OP. Each point in the figure is an average of 1000 instances.}
	\label{overestimation}
\end{figure}

To further evaluate the impact of overestimations of $ \overline{p} $ on the ER performance of PPM-OP, in particular, to quantify how much overestimation will lead to a worse ER performance than using the exact value of $ \overline{p} $, we change the uniform distribution of $ p $ in \textsf{Case-UI} to a truncated normal distribution as follows:
\begin{equation*}
p\sim \mathbb{N}(\mu,\sigma^2,0,\overline{p}),
\end{equation*} 
where $ \mu, \sigma, 0$, and $ \overline{p}$ denote the mean, the standard deviation, the lower bound, and the upper bound of random variable $ p $, respectively. We set $ \mu = \overline{c}_{\mathsf{cpu}} $ and $ \overline{p} = 9\overline{c}_{\mathsf{cpu}} $, and assume similarly as \textsf{Case-UI} that the optimal pricing function is designed based on the estimated upper bound $ \overline{p}_{\mathsf{estimate}} = \overline{p}(1+\delta) $, where $ \sigma>0 $ since here we only consider overestimation.  We
plot the ER performances of PPM-OP with different variances in Fig. \ref{overestimation}. It can be seen that when the variance is small, e.g., $ \sigma = 1 $ in Fig. \ref{overestimation}(a), the ER performance of PPM-OP becomes worse w.r.t. the increase of $ \delta>0 $. When the variance is higher, e.g., $ \sigma = 2 $ in Fig. \ref{overestimation}(b) and $ \sigma = 10 $ in Fig. \ref{overestimation}(c), the ER performance of PPM-OP first improves and then deteriorates w.r.t. the increase of $ \delta>0 $, similar to the results in Fig. \ref{ER_inexact_case} when $ p $ is uniformly distributed. An interesting result revealed by Fig. \ref{overestimation} is that PPM-OP can tolerate a higher estimation error of $ \overline{p} $ when the variance of $ p $ is higher. In other words, when the arrival instance is highly uncertain or volatile, it tends to be more beneficial for the provider to overestimate $ \overline{p} $.  This insight shows that when there exists no exact statistical model about future arrivals, the information uncertainty is not always a disadvantage. Instead,  the provider can artificially amplify the estimation of  $ \overline{p} $ so as to benefit from the uncertainty of arrival instances. We argue that this is another advantage of our proposed PPM-OP as the prior theoretic analysis does not provide such a guarantee.

\section{Conclusion}
We studied the online combinatorial auctions for resource allocation with supply costs and capacity limits. In the studied model, the provider charges payment from customers who purchase a bundle of resources  and incurs an increasing supply cost with respect to the total resource allocated. We focused on maximizing the social welfare. Adopting the competitive analysis framework we provided an optimal online mechanism via posted-price.  Our online mechanism is optimal in that no other online algorithms can achieve a better competitive ratio. Our theoretic results improve and generalize the results in prior work. Moreover, we validated our results via empirical studies of online resource allocation in cloud computing, and showed that our pricing mechanism is more competitive than existing benchmarks. We expect that the model and algorithms presented in this paper will find application in different paradigms of networking and computing systems. Meanwhile, leveraging techniques in artificial intelligence and machine learning to extend our model is an interesting future direction, e.g., posted-price via online learning.

\bibliographystyle{plain}
\bibliography{OCA_arXiv}

\begin{thebibliography}{10}

\bibitem{uniqueness_book1993}
R.~P. Agarwal and V.~Lakshmikantham.
\newblock {\em Uniqueness and Nonuniqueness Criteria for Ordinary Differential
  Equations}.
\newblock World Scientific, 1993.

\bibitem{ODE1973}
Vladimir Arnold.
\newblock {\em {Ordinary differential equations}}.
\newblock MIT Press, 1973.

\bibitem{covering_packing}
Y.~Azar, N.~Buchbinder, T.~H. Chan, S.~Chen, I.~R. Cohen, A.~Gupta, Z.~Huang,
  N.~Kang, V.~Nagarajan, J.~Naor, and D.~Panigrahi.
\newblock Online algorithms for covering and packing problems with convex
  objectives.
\newblock In {\em 2016 IEEE 57th Annual Symposium on Foundations of Computer
  Science (FOCS)}, volume~00, pages 148--157, Oct. 2016.

\bibitem{Bartal2003}
Yair Bartal, Rica Gonen, and Noam Nisan.
\newblock Incentive compatible multi unit combinatorial auctions.
\newblock In {\em Proceedings of the 9th Conference on Theoretical Aspects of
  Rationality and Knowledge}, TARK '03, pages 72--87, New York, NY, USA, 2003.
  ACM.

\bibitem{Blum2011}
Avrim Blum, Anupam Gupta, Yishay Mansour, and Ankit Sharma.
\newblock Welfare and profit maximization with production costs.
\newblock In {\em Proceedings of the 2011 IEEE 52nd Annual Symposium on
  Foundations of Computer Science}, pages 77--86, Washington, DC, USA, 2011.

\bibitem{Borodin1998}
Allan Borodin and Ran El-Yaniv.
\newblock {\em Online Computation and Competitive Analysis}.
\newblock Cambridge University Press, 1998.

\bibitem{Buchbinder2015}
Niv Buchbinder and R.~Gonen.
\newblock {Incentive Compatible Mulit-Unit Combinatorial Auctions: A Primal
  Dual Approach}.
\newblock {\em Algorithmica}, 72:167--190, 2015.

\bibitem{Buchbinder2009}
Niv Buchbinder and Joseph~(Seffi) Naor.
\newblock {Online Primal-Dual Algorithms for Covering and Packing}.
\newblock {\em Mathematics of Operations Research}, 34(2):270--286, May 2009.

\bibitem{PPM2010}
Shuchi Chawla, Jason~D. Hartline, David~L. Malec, and Balasubramanian Sivan.
\newblock Multi-parameter mechanism design and sequential posted pricing.
\newblock In {\em Proceedings of the Forty-second ACM Symposium on Theory of
  Computing}, pages 311--320, New York, NY, USA, 2010. ACM.

\bibitem{fog_computing}
M.~{Chiang} and T.~{Zhang}.
\newblock Fog and {IoT}: An overview of research opportunities.
\newblock {\em IEEE Internet of Things Journal}, 3(6):854--864, Dec 2016.

\bibitem{power_models}
C.~Kozyrakis D.~Economou, S.~Rivoire and P.~Ranganathan.
\newblock Full-system power analysis and modeling for server environments.
\newblock In {\em Workshop on Modeling Benchmarking and Simulation (MOBS)},
  2006.

\bibitem{data_center_power}
M.~{Dayarathna}, Y.~{Wen}, and R.~{Fan}.
\newblock Data center energy consumption modeling: A survey.
\newblock {\em IEEE Communications Surveys \& Tutorials}, 18(1):732--794,
  Firstquarter 2016.

\bibitem{CA_survey}
Sven de~Vries and Rakesh Vohra.
\newblock Combinatorial auctions: A survey.
\newblock {\em INFORMS J. Comput.}, 15(3):284--309, 2003.

\bibitem{Adword2009}
Nikhil~R. Devanur and Thomas~P. Hayes.
\newblock The adwords problem: Online keyword matching with budgeted bidders
  under random permutations.
\newblock In {\em Proc. of the 10th ACM Conference on Electronic Commerce},
  pages 71--78, New York, NY, USA, 2009. ACM.

\bibitem{Huang2017}
Nikhil~R. Devanur and Zhiyi Huang.
\newblock Primal dual gives almost optimal energy-efficient online algorithms.
\newblock {\em ACM Trans. Algorithms}, 14(1), December 2017.

\bibitem{OM_concave}
Nikhil~R. Devanur and Kamal Jain.
\newblock Online matching with concave returns.
\newblock In {\em Proceedings of the Forty-fourth Annual ACM Symposium on
  Theory of Computing (STOC)}, New York, NY, USA, 2012.

\bibitem{Huang2018}
Z.~Huang and A.~Kim.
\newblock Welfare maximization with production costs: A primal dual approach.
\newblock {\em Games and Econ. Behav.}, 2018.

\bibitem{b_matching}
Bala Kalyanasundaram and Kirk~R. Pruhs.
\newblock An optimal deterministic algorithm for online b-matching.
\newblock {\em Theor. Comput. Sci.}, 233(1-2):319--325, February 2000.

\bibitem{economics1995}
Andreu Mas-Colell, Michael~D. Whinston, and Jerry~R. Green.
\newblock {\em Microeconomic Theory}.
\newblock Oxford University Press, 1995.

\bibitem{Mehta2013}
Aranyak Mehta.
\newblock Online matching and ad allocation.
\newblock {\em Found. Trends Theor. Comput. Sci.}, 8(4):265--368, October 2013.

\bibitem{AMD2001}
Noam Nisan and Amir Ronen.
\newblock Algorithmic mechanism design.
\newblock {\em Games and Economic Behavior}, 35(1):166 -- 196, 2001.

\bibitem{Perko2001}
Lawrence Perko.
\newblock {\em {Differential Equations and Dynamical Systems}}.
\newblock Springer New York, New York, NY, 2001.

\bibitem{ODE_book}
A.~D. Polyanin and V.~F. Zaitsev.
\newblock {\em {Handbook of exact solutions for ordinary differential
  equations}}.
\newblock Chapman {\&} Hall/CRC, 2003.

\bibitem{CA_PNAS}
David Porter, Stephen Rassenti, Anil Roopnarine, and Vernon Smith.
\newblock Combinatorial auction design.
\newblock {\em Proceedings of the National Academy of Sciences},
  100(19):11153--11157, 2003.

\bibitem{google_trace_analysis}
Charles Reiss, Alexey Tumanov, Gregory~R. Ganger, Randy~H. Katz, and Michael~A.
  Kozuch.
\newblock Heterogeneity and dynamicity of clouds at scale: Google trace
  analysis.
\newblock In {\em Proceedings of the 3rd ACM Symposium on Cloud Computing
  (SOCC)}, New York, NY, USA, 2012. ACM.

\bibitem{network_slicing}
P.~{Rost}, C.~{Mannweiler}, D.~S. {Michalopoulos}, C.~{Sartori},
  V.~{Sciancalepore}, N.~{Sastry}, O.~{Holland}, S.~{Tayade}, B.~{Han},
  D.~{Bega}, D.~{Aziz}, and H.~{Bakker}.
\newblock Network slicing to enable scalability and flexibility in {5G} mobile
  networks.
\newblock {\em IEEE Communications Magazine}, 55(5):72--79, May 2017.

\bibitem{bo2018}
B.~Sun, X.~Tan, and {D.H.K. Tsang}.
\newblock Eliciting multi-dimensional flexibilities from electric vehicles: a
  mechanism design approach.
\newblock {\em IEEE Transactions on Power Systems}, 34(5):4038--4047, 2019.

\bibitem{Adam_speed_scaling}
A.~{Wierman}, L.~L.~H. {Andrew}, and A.~{Tang}.
\newblock Power-aware speed scaling in processor sharing systems.
\newblock In {\em IEEE INFOCOM 2009}, pages 2007--2015, April 2009.

\bibitem{speed_scaling}
F.~{Yao}, A.~{Demers}, and S.~{Shenker}.
\newblock A scheduling model for reduced cpu energy.
\newblock In {\em Proceedings of IEEE 36th Annual Foundations of Computer
  Science}, pages 374--382, Oct 1995.

\bibitem{XZhang_ToN}
X.~{Zhang}, Z.~{Huang}, C.~{Wu}, Z.~{Li}, and F.~C.~M. {Lau}.
\newblock Online auctions in iaas clouds: Welfare and profit maximization with
  server costs.
\newblock {\em IEEE/ACM Transactions on Networking}, 25(2):1034--1047, April
  2017.

\bibitem{knapsack}
Yunhong Zhou, Deeparnab Chakrabarty, and Rajan Lukose.
\newblock Budget constrained bidding in keyword auctions and online knapsack
  problems.
\newblock In {\em Internet and Network Economics}, pages 566--576, Berlin,
  Heidelberg, 2008. Springer Berlin Heidelberg.

\end{thebibliography}

\newpage
\appendix

\section{Proof of Theorem \ref{a_unified_BVP}}\label{proof_a_unified_BVP}
Our proof of Theorem \ref{a_unified_BVP} is based on the online primal-dual analysis and first-order two-point boundary value problems (BVPs). In the following we first give some mathematical preliminaries, and then prove the sufficient and necessary conditions in Theorem \ref{a_unified_BVP} separately. 

\subsection{Mathematical Preliminaries}
In this section we present some mathematical preliminaries to help our proof of Theorem \ref{a_unified_BVP}.

\subsubsection{Online Primal-Dual Analysis} Let us consider the following convex optimization problem:
\begin{subequations}\label{SWM_relaxed}
	\begin{alignat}{3}
	& \underset{\mathbf{x,y}}{\mathrm{maximize}} \qquad\qquad & &   \sum_{n\in\mathcal{N}}\sum_{b\in\mathcal{B}} v_n^b x_n^b - \sum_{k\in\mathcal{K}} \bar{f}\left(y_k\right), \\
	&\mathrm{subject\ to} & & \sum_{n\in\mathcal{N}}\sum_{b\in\mathcal{B}} r_k^b x_n^b \leq y_k, & &  (p_k)\label{total_utilization_relaxed} \\ 
	&    & & \sum_{b\in\mathcal{B}} x_n^b \leq 1,\forall n, & & (\mu_n) \label{bid_selection}\\
	&    & &  x_n^b \geq 0,  \forall n,b; y_k \geq 0,\forall k,
	\end{alignat}
\end{subequations}
where  $ p_k, \mu_n $ denote the corresponding dual variables of each constraint. The above convex program differs from the original social welfare maximization problem \eqref{SWM} in the following aspects. 
\begin{itemize}[leftmargin=*]
	\item First, in the objective function of Problem \eqref{SWM_relaxed}, we modify the cost function $ f $ to $ \bar{f} $ as follows:
	\begin{align}\label{bar_f}
	\bar{f}(y) =
	\begin{cases}
	f(y) &\text{ if } y\in [0,1],\\
	+\infty &\text{ if } y\in (1,+\infty). 
	\end{cases}
	\end{align}
	Therefore, $ \bar{f} $ is an extended version of $ f $ for the whole range of $ [0,+\infty) $. In optimization theory, $ \bar{f} $ is often regarded as a barrier function of $ f $. It is know that performing such a transformation does not change the optimization problem itself. 
	
	\item Second, we relax the binary status variable $ x_n^b $ to be a continuous variable within $ [0,1] $ for all $ n, b $. 
	
	\item Third, the equality  constraint in Eq. \eqref{total_utilization_original} is relaxed to be an inequality one in Eq.  \eqref{total_utilization_relaxed}. Since the cost function $ f(\cdot) $ is increasing, constraint \eqref{total_utilization_relaxed} will always be binding. 
\end{itemize}

Based on the above discussions, the only difference between Problem \eqref{SWM} and Problem \eqref{SWM_relaxed} is the relaxation of $ \{x_n^b\}_{\forall n,b} $. Given the convex program in Problem \eqref{SWM_relaxed}, the dual problem can be expressed  as follows:
\begin{subequations}\label{dual_relaxed}
	\begin{alignat}{2}
	& \underset{\bm{p},\bm{\mu}}{\mathrm{minimize}}      & \qquad & \sum_{n\in\mathcal{N}}\mu_n + \sum_{k\in\mathcal{K}} f_{\texttt{\#}}(p_k)\\
	&\mathrm{subject\ to} & & \mu_n\geq  v_n^b - \sum_{k\in\mathcal{K}} p_k r_k^b, \forall n,b, \label{dual_constraint_gamma}\\
	&   & & \bm{p} \geq \bm{0}, \bm{\mu} \geq \bm{0}, 
	\end{alignat}
\end{subequations}
where $ f_{\texttt{\#}} $ is the convex conjugate of $ \bar{f} $, and is given by
\begin{align}
f_{\texttt{\#}}(p) = \max_{y\geq 0}\ p y   - \bar{f}(y).
\end{align}
Solving the above optimization leads to the  expression of $ f_{\texttt{\#}} $ as follows:
\begin{align}\label{f_pound}
f_{\texttt{\#}}(p) =
\begin{cases}
0 &\text{ if } p\in [0,\underline{c}],\\
p f'^{-1}(p) - f(f'^{-1}(p)) &\text{ if } p\in (\underline{c},\overline{c}),\\
p - f(1)   &\text{ if } p\in [\overline{c},+\infty].
\end{cases} 
\end{align}

If we denote the optimal objective of the relaxed primal problem \eqref{SWM_relaxed} and its dual \eqref{dual_relaxed} by $ W_{\textsf{r-primal}} $ and $ W_{\textsf{r-dual}} $, respectively, then   we have
\begin{align}\label{weak_duality}
W_{\textsf{opt}}\leq W_{\textsf{r-primal}}\leq W_{\textsf{r-dual}},
\end{align}
where $ W_{\textsf{opt}} $ is the optimal objective of the original offline problem \eqref{SWM}. In particular, the first inequality in Eq. \eqref{weak_duality} is due to the relaxation of $ \{x_n^b\}_{\forall n,b} $ and the second inequality comes from  weak duality.  

The key to the design of $ \text{PPM}_{\phi} $ is to link the pricing function $ p_k^{(n)} = \phi(y_k^{(n-1)})$ to the offline shadow price $ p_k $. Specifically, when there is no future information, it is impossible to know the exact value of $ p_k $. Our idea is to design the posted price $ p_k^{(n)} $ as a function of the current total power consumption $ y_k^{(n-1)} $, and using $ p_k^{(n)} $ to approximate the exact shadow price at each round. 

Following this idea, let us denote the primal and dual objective by $ P_n $ and $ D_n $ after processing customer $ n $, respectively. Intuitively, $ P_0 $ and $ D_0 $ denote the initial values (i.e., before processing the first customer), and $ P_N $ and $ D_N $ represent the terminal values (i.e., after processing the last customer of interest). Obviously, $ P_0 = 0 $ and $ D_0 $  is given by
\begin{equation}\label{D_0}
D_0 = \sum_{k\in\mathcal{K}} f_{\texttt{\#}}\big(p_k^{(1)}\big)  =  \sum_{k\in\mathcal{K}} f_{\texttt{\#}}\left(\phi(y_k^{(0)})\right) = \sum_{k\in\mathcal{K}} f_{\texttt{\#}}\left(\phi(0)\right),
\vspace{-0.2cm}
\end{equation}
where $ \phi(0) $ represents the initial price when the resource utilization level is zero. 

(\textbf{Principles of the Online Primal-Dual Approach}) The principle of the online primal-dual approach is that, if the pricing function $ \phi $ is constructed in a certain way so that i) $ D_0  = 0 $ and the solutions found by $ \text{PPM}_{\phi} $  are feasible, and ii) the following \textbf{incremental inequality}  $ P_n - P_{n-1} \geq \frac{1}{\alpha}\left( D_n - D_{n-1}\right) $ holds for each round with a constant $ \alpha $, then $ P_N = \sum_{n=1}^N\left(P_n - P_{n-1}\right)
\geq\ \frac{1}{\alpha}\sum_{n=1}^N\left(D_n - D_{n-1}\right) = \frac{1}{\alpha} D_N $. Note that $ P_N $ denotes the social welfare achieved by $ \text{PPM}_{\phi} $, i.e., $ W_{\textsf{online}} = P_N $. Based on Eq.  \eqref{weak_duality}, we have $$ W_{\textsf{online}} = P_N \geq \frac{1}{\alpha} D_N\geq \frac{1}{\alpha} W_{\textsf{r-dual}} \geq  \frac{1}{\alpha} W_{\textsf{opt}}, $$ 
which thus indicates that $ \text{PPM}_{\phi} $ is $ \alpha $-competitive.

\subsubsection{Convex Conjugates and Properties} In the following we will heavily rely on the properties of convex conjugates and Fenchel duality.  Below we introduce some properties regarding $ f_{\texttt{\#}} $. 

\begin{lemma}[Properties of $ f_{\texttt{\#}}  $]
	$ f_{\texttt{\#}}  $ has the following properties:
	\begin{enumerate}
		\item $ f_{\texttt{\#}}(p) $ is increasing in $ p\in [\underline{c},+\infty] $ and $ f_{\texttt{\#}}(\underline{c}) = 0 $.
		\item $ f_{\texttt{\#}}(p) $ is convex and differentiable in $ p\in [\underline{c},+\infty] $, even if the original cost function $ f(y) $ is non-convex and non-differentiable.
		\item   For any $ y\in [0,1] $, if $ f'(y) = p $, then $f_{\texttt{\#}}'(p) = y  $ and $ f(y) +  f_{\texttt{\#}}(p) = py $.  
		\item  The derivative of $ f_{\texttt{\#}}(p) $ w.r.t. $ p\in [\underline{c},+\infty] $ is given by
		\begin{align}
		f_{\texttt{\#}}'(p) =
		\begin{cases}
		0 &\text{ if } p\in [0,\underline{c}),\\
		f'^{-1}(p) &\text{ if } p\in [\underline{c},\overline{c}),\\
		1   &\text{ if } p\in [\overline{c},+\infty].
		\end{cases} 
		\end{align}
	\end{enumerate}
\label{lemma_Fenchel_duality}
\end{lemma}

We omit the proof of Lemma \ref{lemma_Fenchel_duality} for brevity.  For a detailed discussion of the properties of the conjugate function $ f_{\texttt{\#}} $, please refer to \cite{Huang2017}. 

\subsubsection{First-Order Two-Point BVPs} 
In the field of differential equations, a first-order boundary value problem (BVP) is a first-order ordinary differential equation (ODE) with  a set of additional boundary conditions. When there is only one additional condition other than the ODE, the resulting problem is a first-order initial value problem (IVP), whose standard form is written as follows:
\begin{align}
\begin{cases}
q'(\omega) = Q(\omega,q), \omega \in \Omega,\\
q(\omega_0) = q_0,
\end{cases}
\end{align}
where $ q(\omega_0) = q_0 $ is usually termed as the initial condition. When there is one more condition,  the resulting first-order two-point BVP can be written  in the following standard form
\begin{align}\label{def_BVP}
\begin{cases}
q'(\omega) = Q(\omega,q), \omega \in \Omega,\\
q(\omega_1) = q_1, q(\omega_2) = q_2,
\end{cases}
\end{align} 
where $ (\omega_1,q_1) $ and $ (\omega_2,q_2) $ are two points in the domain of $ Q $. A solution to the first-order two-point BVP in Eq. \eqref{def_BVP} is a function $ q(\omega) $ that satisfies the ODE and also satisfies the two boundary conditions simultaneously. 

Key to the analysis of IVPs and BVPs is the existence and uniqueness of solutions \cite{ODE1973, ODE_book}. For first-order IVPs, the existence and uniqueness theorem is well understood. In particular,  the Picard–Lindel\"{o}f theorem guarantees the unique exsitence of solutions as long as the function $ Q $ satisfies a certain Lipschitz continuity conditions \cite{ODE_book}. Meanwhile, there are numerous iterative  methods off-the-shelf that can solve IVPs numerically \cite{ODE1973}.  However,  for BVPs, there is no general uniqueness and existence theorem. As argued by \cite{ODE_book}, it is even non-trivial to obtain numerical solutions for some BVPs in the most basic two-point case as Eq. \eqref{def_BVP}.

\subsection{The Proof of Theorem \ref{a_unified_BVP}}
We first prove the sufficient conditions in Theorem \ref{a_unified_BVP}. 
Below we give Theorem \ref{sufficiency} which summarizes the sufficient conditions to guarantee a bounded competitive ratio for $ \text{PPM}_{\phi} $.  

\begin{theorem}[\textbf{Sufficiency}]\label{sufficiency}
	Given a setup $ \mathcal{S} $ with $ \overline{p}\in (\underline{c},+\infty) $, $\normalfont \text{PPM}_{\phi} $ is $ \alpha $-competitive if the pricing function $ \phi $ satisfies the following differential equation
	\begin{align}\label{ODE_sufficiency}
	\phi(y) - f'(y) = \frac{1}{\alpha}\cdot \frac{d f_{\normalfont\texttt{\#}}(\phi(y)) }{dy}, y\in [0,1]
	\end{align}
	with the following boundary conditions:
	\begin{align}
	\begin{cases}
	\phi(0) =  \underline{c}, 
	\phi(v)\geq \overline{p}, & \text{ if } \overline{p}\in (\underline{c},\overline{c}], \quad\  (\textsf{LUC})\\
	\phi(0) =  \underline{c},
	\phi(1)\geq \overline{p},  & \text{ if } \overline{p}\in (\overline{c},+\infty), (\textsf{HUC})
	\end{cases}
	\end{align}  
	where $ v \triangleq f'^{-1}(\overline{p}) $. 
\end{theorem}
\begin{proof}
	The proof of this theorem is based on showing that once the pricing function $ \phi $ satisfies the conditions in Theorem \ref{sufficiency}, then the following incremental inequality 
	\begin{align}
	P_n - P_{n-1} \geq \frac{1}{\alpha}\left( D_n - D_{n-1}\right) 
	\end{align}
	holds at each round with $ D_0 = 0 $. 
	To prove the above incremental inequality holds at each round, we only need to focus on the case when customer $ n $ indeed purchases a bundle of resources, say bundle $ b_* $. Otherwise, $ P_n- P_{n-1} = D_n - D_{n-1}  = 0 $ and the incremental inequality holds obviously. 
	
	We first calculate the change of the primal objective after processing customer $ n $. Based on Problem \eqref{SWM_relaxed}, we can calculate the difference between $ P_n $ and $ P_{n-1} $ as follows: 
	\begin{align*}
	P_n - P_{n-1} 
	=\ &  v_n^{b_*} -  \sum_{k\in\mathcal{K}}\left(\bar{f}\left(y_k^{(n)}\right)-\bar{f}\left(y_k^{(n-1)}\right)\right) \nonumber\\
	=\ & \mu_n + \sum_{k\in\mathcal{K}} p_k^{(n)} r_k^{b_*} -  \sum_{k\in\mathcal{K}}\left(\bar{f}\left(y_k^{(n)}\right)-\bar{f}\left(y_k^{(n-1)}\right)\right)\\
	\stackrel{(i)}{=} \ & \mu_n + \sum_{k\in\mathcal{K}} \phi\left(y_k^{(n-1)}\right) r_k^{b_*} -  \sum_{k\in\mathcal{K}}\left(\bar{f}\left(y_k^{(n)}\right)-\bar{f}\left(y_k^{(n-1)}\right)\right)\\
	\stackrel{(ii)}{=}\ & \mu_n + \sum_{k\in\mathcal{K}} \phi\left(y_k^{(n-1)}\right) \left(y_k^{(n)} - y_k^{(n-1)}\right)-  \sum_{k\in\mathcal{K}}\left(\bar{f}\left(y_k^{(n)}\right)-\bar{f}\left(y_k^{(n-1)}\right)\right),
	\end{align*}
	where $ (i) $ comes from constraint \eqref{dual_constraint_gamma} in the dual problem, namely, we set $ \mu_n =  v_n^{b_*} - \sum_{k\in\mathcal{K}} \phi (y_k^{(n-1)} ) r_k^{b_*}$, and $ (ii) $ is because $ r_k^{b_*} = y_k^{(n)} - y_k^{(n-1)} $ based on line \ref{update} in Algorithm \ref{PM}.
	
	Similarly, we calculate the change of the dual objective after processing customer $ n $. Based on Problem \eqref{dual_relaxed}, we have
	\begin{align}
	D_n - D_{n-1}
	=  \mu_n +\sum_{k\in\mathcal{K}} f_{\texttt{\#}}\left(\phi(y_k^{(n)})\right) - \sum_{k\in\mathcal{K}} f_{\texttt{\#}}\left(\phi(y_k^{(n-1)})\right),
	\end{align}
	where $  \phi(y_k^{(n)}) $ denotes the posted price after processing customer $ n $ (i.e., the posted price for customer $ n+1 $). Since $ \mu_n \geq 0 $ holds for all $ n\in\mathcal{N} $, to guarantee the incremental inequality holds at each round, the following inequality must be satisfied:
	\begin{align}\nonumber
	&\sum_{k\in\mathcal{K}} \phi(y_k^{(n-1)}) \left(y_k^{(n)} - y_k^{(n-1)}\right)-  \sum_{k\in\mathcal{K}}\left(\bar{f}(y_k^{(n)})-\bar{f}(y_k^{(n-1)})\right)\\
	\geq\ & \frac{1}{\alpha} \left(\sum_{k\in\mathcal{K}} f_{\texttt{\#}}\left(\phi(y_k^{(n)})\right) - \sum_{k\in\mathcal{K}} f_{\texttt{\#}}\left(\phi(y_k^{(n-1)})\right)\right).
	\end{align}
	Since the posted-price is designed for each type of resource, the above inequality holds if the following inequality holds
	\begin{align}
	\phi(y_k^{(n-1)}) \left(y_k^{(n)} - y_k^{(n-1)}\right)-  \left(\bar{f}(y_k^{(n)})-\bar{f}(y_k^{(n-1)})\right)
	\geq \frac{1}{\alpha} \left( f_{\texttt{\#}}\big(\phi(y_k^{(n)})\big) -  f_{\texttt{\#}}\big(\phi(y_k^{(n-1)})\big)\right),\label{before_diff}
	\end{align}
	which can be equivalently written as follows: 
	\begin{align*}
	& \phi(y_k^{(n-1)}) - \frac{\bar{f}\left(y_k^{(n-1)} + r_k^{b_*}\right)-\bar{f}\left(y_k^{(n-1)}\right)}{y_k^{(n-1)} + r_k^{b_*} - y_k^{(n-1)}}\\
	\geq\ & \frac{1}{\alpha} \cdot \frac{\phi\left(y_k^{(n-1)} + r_k^{b_*}\right)  - \phi\left(y_k^{(n-1)}\right)}{y_k^{(n-1)} + r_k^{b_*} - y_k^{(n-1)}} \cdot  \frac{ f_{\texttt{\#}}\left(\phi\left(y_k^{(n-1)} + r_k^{b_*}\right)\right) -  f_{\texttt{\#}}\left(\phi\left(y_k^{(n-1)}\right)\right)}{\phi\left(y_k^{(n-1)} + r_k^{b_*}\right)  - \phi\left(y_k^{(n-1)}\right)}.
	\end{align*}
	Since  $ r_k^{b_*} $ is very small (Assumption \ref{assumption_small}),  the above equality can be written as follows:
	\begin{align}
	\phi(y_k^{(n-1)}) - \bar{f}'(y_k^{(n-1)}) 
	\geq \frac{1}{\alpha}\cdot   \phi'(y_k^{(n-1)})\cdot f_{\texttt{\#}}'\big(\phi(y_k^{(n-1)})\big).
	\end{align}
	Therefore, if the above inequality holds for any realization of $ y_k^{(n-1)}\in [0,1) $, namely,
	\begin{align}\label{ODI}
	\phi\left(y\right) - \bar{f}'\left(y\right) 
	\geq \frac{1}{\alpha} \cdot  \phi'\left(y\right)\cdot f_{\texttt{\#}}'\left(\phi(y)\right) = \frac{1}{\alpha} \cdot \frac{d f_{\normalfont\texttt{\#}}\big(\phi(y)\big) }{dy}, \forall y\in [0,1],
	\end{align}
	then the incremental inequality $ P_n - P_{n-1} \geq \frac{1}{\alpha}\left( D_n - D_{n-1}\right) $ holds at each round when $ y\in [0,1] $. Recall that when $ y\in [0,1] $, $ \bar{f} = f $, and thus the above inequality in Eq. \eqref{ODI} can be written as 
	\begin{align}\label{ODI_final}
	\phi\left(y\right) - f'\left(y\right) 
	\geq  \frac{1}{\alpha} \cdot \frac{d f_{\normalfont\texttt{\#}}\big(\phi(y)\big) }{dy}, \forall y\in [0,1].
	\end{align}
	Therefore, if Eq. \eqref{ODI_final} holds for all $ y\in [0,1] $, then the incremental inequality holds at each round when $ y\in [0,1] $. However, we emphasize that \textit{this does not mean the incremental inequality holds at each around for all $ y\in [0,+\infty) $}.  
	
	We next show why we need  the two boundary conditions of $ \phi(0)= \underline{c} $ and $ \phi(1) \geq \overline{p} $. First, according to Eq. \eqref{D_0}, when $ \phi(0) = \underline{c} $, we have $ D_0 =  \sum_k f_{\texttt{\#}}\big(\phi(y_k^{(0)})\big) = \sum_k f_{\texttt{\#}}(\underline{c}) = 0 $, where we use the property of  $ f_{\texttt{\#}}(\underline{c}) = 0  $ based on Lemma \ref{lemma_Fenchel_duality}. Therefore, the boundary condition of $ \phi(0) = \underline{c} $ is to guarantee that $ D_0 = 0 $.  Second, taking integration on both sides of Eq. \eqref{ODI_final} leads to
	\begin{align}\label{int_inequality}
	\int_{0}^{y}\left(\phi(\eta)- f'(\eta)\right) d\eta = \int_{0}^{y}\phi(\eta)  d\eta - f(y) \geq \frac{1}{\alpha}\Big(f_{\texttt{\#}}\big(\phi(y)\big) - f_{\texttt{\#}}\big(\phi(0)\big)\Big) = \frac{1}{\alpha} f_{\texttt{\#}}\big(\phi(y)\big). 
	\end{align} 
	As can be seen from Fig. \ref{three_pricing_schemes}, the left-hand-side of Eq. \eqref{int_inequality} is the area of the grey region between $ \phi(y) $ and $ f'(y) $.  Based on Eq. \eqref{f_pound}, the above inequality in Eq. \eqref{int_inequality} can be written as follows:
	\begin{align}\label{int_inequality_two_cases}
	\int_{0}^{y} \phi(\eta)d\eta- f(y) \geq 
	\begin{cases}
	\frac{1}{\alpha}\Big(\phi(y) \cdot f'^{-1}\left(\phi(y)\right) -  f\big(f'^{-1}(\phi(y))\big)\Big) &\text{ if } \phi(y) \in  (\underline{c},\overline{c}],\\
	\frac{1}{\alpha}\big(\phi(y) -  f(1)\big) &\text{ if } \phi(y) \in  (\overline{c},+\infty),
	\end{cases}
	\end{align}
	Let us first focus on  the second case when $ \phi(y) > \overline{c} $, where $ y\in [0,1] $. The above integral inequality must hold for any $ y\in [0,1] $. Therefore, when $ y = 1 $, the second case of the right-hand-side of Eq. \eqref{int_inequality_two_cases} is given by
	\begin{align}\label{int_inequality_y_1}
	\int_{0}^{1} \phi(\eta)d\eta- f(1) \geq \frac{1}{\alpha} \cdot \big(\phi(1) -  f(1)\big).
	\end{align}
	On the other hand, when $ \overline{p}\in (\overline{c},+\infty) $, $ \text{PPM}_\phi $ is $ \alpha$-competitive indicates that the pricing function must satisfy the following inequality
	\begin{align}\label{int_inequality_p_1}
	\int_{0}^{1} \phi(\eta)d\eta- f(1) \geq \frac{1}{\alpha} \cdot \big( \overline{p} -  f(1)\big).
	\end{align}
	Note that the rationality of Eq. \eqref{int_inequality_p_1}	follows the same analogy to our analysis in Section \ref{section_theorem_1} regarding the special arrival instance $ \mathcal{A}_v $ when $ \overline{p}\in (\underline{c},\overline{c}] $. Based on Eq. \eqref{int_inequality_y_1} and Eq. \eqref{int_inequality_p_1}, to guarantee Eq. \eqref{int_inequality_p_1} holds, it suffices to have $ \phi(1) \geq \overline{p} $.  
	
	Therefore, \textit{when Eq. \eqref{ODI_final} holds for all $ y\in [0,1] $ with the boundary conditions of $ \phi(0) = \underline{c} $ and $ \phi(1)\geq \overline{p} $, then the incremental inequality holds at each round for all $ y\in [0,+\infty) $}.  Summarizing the above analysis, when $ \overline{p}\in (\overline{c},+\infty) $ (i.e., \textsf{HUC}), if the differential equation in Eq. \eqref{ODE_sufficiency} holds with the boundary conditions of $ \phi(0) = \underline{c} $ and $ \phi(1)\geq \overline{p} $, then $ \text{PPM}_\phi $ is $ \alpha $-competitive. We skip the proof for the case of \textsf{LUC} as it is similar to that of \textsf{HUC}. Hence, we finish the proof of Theorem \ref{sufficiency}. 
\end{proof}	
		
As we mentioned in Section \ref{section_theorem_1}, the division of the two cases of \textsf{LUC} and \textsf{HUC} is not artificial, it is derived from the online primal-dual analysis of the original social welfare problem in Eq. \eqref{SWM}. Note that substituting  $ f_{\normalfont\texttt{\#}}' $ into the differential equation in Eq. \eqref{ODE_sufficiency} leads to the two BVPs in \textsf{HUC} in Theorem \ref{a_unified_BVP}. \textbf{We thus complete the proof of the sufficient conditions in Theorem \ref{a_unified_BVP}}. 

We next prove the necessity of Theorem \ref{a_unified_BVP} by giving the following Theorem \ref{necessity}.  

\begin{theorem}[\textbf{Necessity}]\label{necessity}
	Given a setup $ \mathcal{S} $ with $ \overline{p}\in (\underline{c},+\infty) $, if there exists an $ \alpha $-competitive online algorithm, then there must exist a strictly increasing function $ \psi(\eta) $ that satisfies:   
		\begin{equation}\label{integral_equality}
		\begin{cases}
		\int_{0}^{p} \eta \psi'(\eta)d\eta - f\big(\psi(p)\big)  =   \frac{f_{\normalfont \texttt{\#}}(p)}{\alpha},\forall p\in [\underline{c},\overline{p}],\\
		\psi\left(\underline{c}\right) = 0, \psi(\overline{p}) \leq  1. 
		\end{cases}
		\end{equation}
\end{theorem}

\begin{proof}
	An online algorithm is $ \alpha $-competitive indicates that the social welfare achieved by this online algorithm is at least $ 1/\alpha $ of the optimal offline social welfare for all possible arrival instances. In the following we first prove that, if there exists an $ \alpha $-competitive online algorithm, then there must exist a monotonically non-decreasing function $ y = \psi(p) $ such that
	\begin{align}\label{necessity_inequality}
	\int_{\underline{c}}^{p} \eta \psi'(\eta)d\eta - f\big(\psi(p)\big)  \geq \frac{1}{\alpha} f_{\texttt{\#}}(p)
	\end{align}
	holds for all $ p\in [\underline{c},\overline{p}] $ with $ \psi(\underline{c}) = 0 $ and $ \psi(\overline{p}) \leq 1 $. After that, we will prove that there exists a strictly-increasing function $ \psi(p) $ that satisfies the inequality in Eq. \eqref{necessity_inequality} with equality.   
	
	Our proof is based on constructing a special arrival instance such that any $ \alpha $-competitive online algorithm must satisfy the inequality in Eq. \eqref{necessity_inequality} in order to achieve at least  $ \frac{1}{\alpha} $ of the offline optimal social welfare in hindsight. Specifically,  for any $ p \in  [\underline{c}, \overline{p}] $, we construct a special arrival instance $ \mathcal{A}_{p} $ as follows.  \textit{Let us assume for each $\eta \in  [\underline{c}, p] $, there is a continuum of groups of customers indexed by $ \eta $, where each group $ \eta $ contains infinitely-many identical customers and has a total demand of $ f_{\texttt{\#}}'(\eta) $ (i.e., each customer's demand is infinitesimally small). The PUV of all the customers in group $ \eta $ is $ \eta $, namely, the total valuation of all the customers in this group is given by $ v_{\eta} = \eta  f_{\texttt{\#}}'(\eta) $.}  Note that $ f_{\texttt{\#}}'(\eta) $ is the maximum units of resource that can be provided when the marginal cost is $ \eta $ per unit. Based on Lemma \ref{lemma_Fenchel_duality}, when $ \eta \in  [\underline{c}, \overline{c}]$, $ f_{\texttt{\#}}'(\eta) =  f'^{-1}(\eta) $; when $ \eta\in  (\overline{c},+\infty)$, $ f_{\texttt{\#}}'(\eta) = 1 $. 
	
	For a given arrival instance $ \mathcal{A}_{p} $ with $ p\in [\underline{c}, \overline{p}] $, the social welfares achieved by the optimal offline algorithm and the $ \alpha $-competitive online algorithm are given as follows:
	\begin{itemize}[leftmargin=*]
		\item \textbf{Offline}: the optimal offline result in hindsight is to allocate $ f_{\texttt{\#}}'(p) $ units of resources to the customers in the last group, i.e., group $ p $, and none to all the previous continuum of customers. The optimal social welfare is thus
		\begin{align}
		W_{\textsf{opt}} = p  f_{\texttt{\#}}'(p) - f\left(f_{\texttt{\#}}'(p)\right) = f_{\texttt{\#}}(p),
		\end{align}
		where we use the third property of Fenchel duality in Lemma \ref{lemma_Fenchel_duality}.
		
		\item \textbf{Online}: for the $ \alpha $-competitive online algorithm, let $ y = \psi(\eta) $ denote the total resource consumption after processing the customers in group  $ \eta \in  [\underline{c}, p] $, and thus $ \psi(\eta) $ represents the resources sold to the continuum of groups of customers in $ [\underline{c},\eta] $. Intuitively, $ \psi(\underline{c}) = 0 $ and  $ \psi(\eta) $ is monotonically non-decreasing in $ \eta\in [\underline{c}, p] $. The social welfare achieved by this online algorithm is thus the total valuation minus the total cost, namely,
		\begin{align}
		W_{\textsf{online}}
		= \int_{\psi(\underline{c})}^{\psi(p)} \eta d(\psi(\eta)) -   
		f\big(\psi(p)\big)  =  \int_{\underline{c}}^{p} \eta \psi'(\eta)d\eta - f\big(\psi(p)\big)
		\end{align}
	\end{itemize}
	
	The online algorithm is $ \alpha $-competitive means that
	\begin{align}\label{integral_appendix}
	\int_{\underline{c}}^{p} \eta \psi'(\eta)d\eta - f\big(\psi(p)\big)  \geq \frac{1}{\alpha} f_{\texttt{\#}}(p)
	\end{align}
	holds for all $ p \in [\underline{c},\overline{p}] $. According to the definition of $ \psi $, we have $ \psi(\underline{c}) = 0 $ and $ \psi(p)\leq 1 $ holds for all $ p\in [\underline{c},\overline{p}] $, and thus $ \psi(\overline{p})\leq 1 $ holds as well. Therefore, if there exists an $ \alpha $-competitive online algorithm, then there must exit a non-decreasing function $ \psi(\eta) $ that satisfies Eq. \eqref{integral_appendix}. 
	
	Next we prove that there exists a strictly-increasing function $ \psi(p) $ that satisfies Eq. \eqref{integral_appendix} with equality. Suppose for a given $ p\in [\underline{c},\overline{p}] $,  $ \psi(\eta) $ is a feasible solution to Eq. \eqref{integral_appendix} and there is another non-decreasing function $ \hat{\psi}(\eta) $ such that $ \psi(\eta) \leq   \hat{\psi}(\eta) $ and  $ \hat{\psi}(p)  = \psi(p) \triangleq \psi_p $, then we have
	 \begin{align}
	 \int_{\underline{c}}^{p} \eta \psi'(\eta)d\eta = p\psi(p) - \int_{\underline{c}}^{p} \psi(\eta)d\eta \geq  p\hat{\psi}(p) - \int_{\underline{c}}^{p} \hat{\psi}(\eta)d\eta = \int_{\underline{c}}^{p} \eta \hat{\psi}'(\eta)d\eta,
	 \end{align}
	 which indicates that we can find another function $ \hat{\psi} $ so that the following inequality holds:
	 \begin{align}\label{two_solution_inequality}
	 \int_{\underline{c}}^{p} \eta \psi'(\eta)d\eta -f(\psi(p)) \geq \int_{\underline{c}}^{p} \eta \hat{\psi}'(\eta)d\eta -f(\hat{\psi}(p))\geq \frac{1}{\alpha} f_{\texttt{\#}}(p).
	 \end{align}
	 This means that for any given solution $ \psi(\eta) $ to Eq. \eqref{integral_appendix}, it is always possible to get a new solution  $ \hat{\psi}(\eta) $ to Eq. \eqref{integral_appendix} by pushing $ \psi(\eta) $ up while keeping the initial and terminal points fixed (i.e., $ \psi(\underline{c}) = \hat{\psi}(\underline{c}) =  0$ and $ \psi(p) = \hat{\psi}(p) = \psi_p$).
	 
	 Recall that for a given $ p\in [\underline{c},\overline{p}] $, when $ \psi(\eta) $ is a feasible solution to Eq. \eqref{integral_appendix}, we have
	 \begin{align}
	 \int_{\underline{c}}^{p} \eta \psi'(\eta)d\eta -f(\psi(p))  = p\psi(p) -f(\psi(p)) - \int_{\underline{c}}^{p} \psi(\eta)d\eta \geq   \frac{1}{\alpha} f_{\texttt{\#}}(p). 
	 \end{align}
	 Based on the above analysis, we can  prove that there always exists a strictly-increasing function $ \psi_*(\eta)\geq \psi(\eta) $ and $ \psi_*$ has the same boundary conditions as $ \psi $ such that
	 \begin{align}\label{psi_star_equality}
	 p\psi(p) -f(\psi(p)) - \int_{\underline{c}}^{p} \psi(\eta)d\eta \geq p\psi_*(p) -f(\psi_*(p)) - \int_{\underline{c}}^{p} \psi_*(\eta)d\eta =   \frac{1}{\alpha} f_{\texttt{\#}}(p). 
	 \end{align}
	 We can prove this as follows. Since $ \psi(p) = \psi_*(p) $ and $ \psi(\eta)\leq \psi_*(\eta) $, the first inequality definitely holds. We just need to prove that  such a strictly-increasing function $ \psi_*(\eta) $ exists so that the second equality in Eq. \eqref{psi_star_equality} holds.
	 
	 \begin{figure}
	 	\centering
	 	\includegraphics[width=4.7cm]{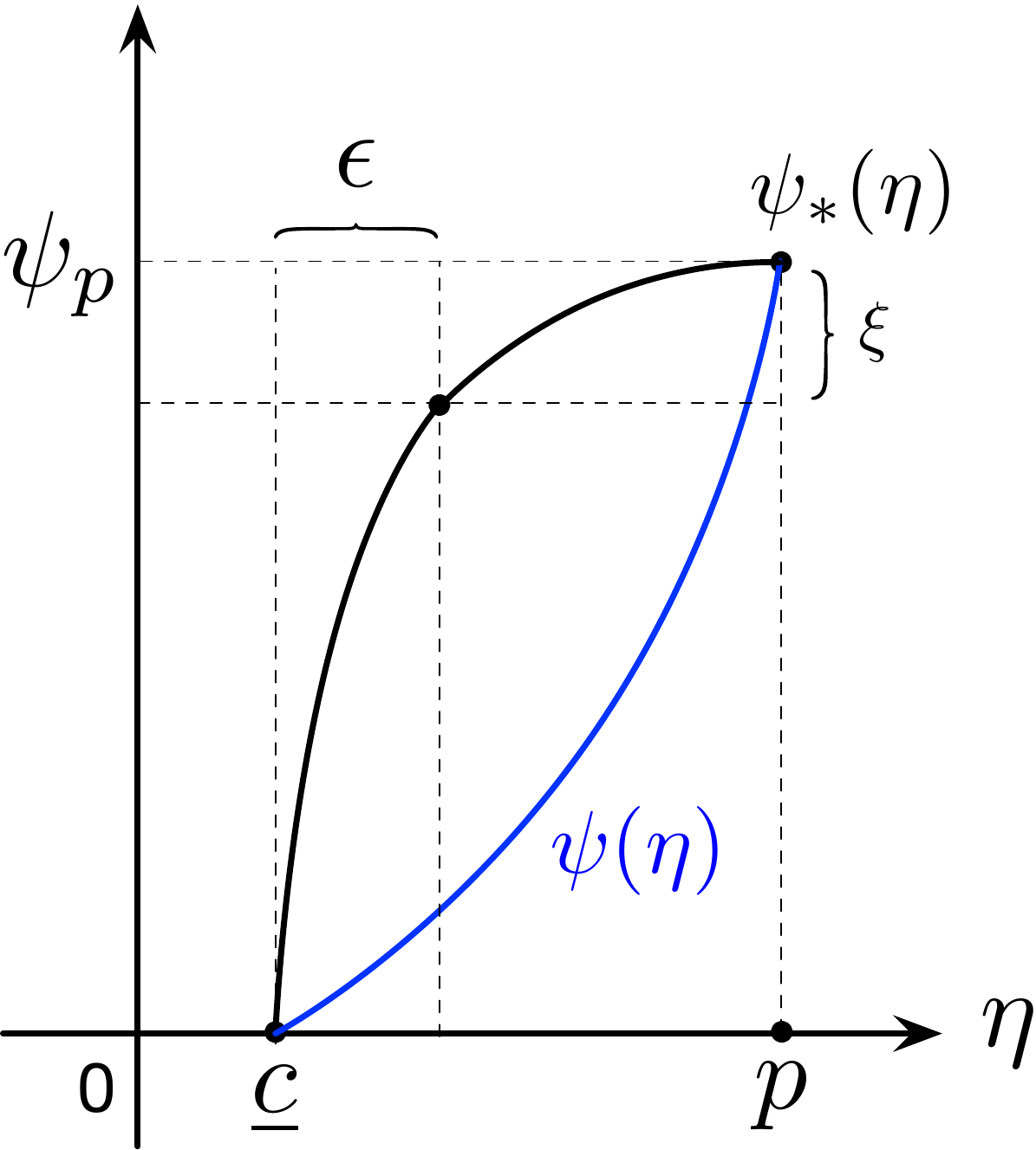}
	 	\caption{Illustration of how to construct a strictly-increasing function $ \psi_*(\eta) $ to satisfy Eq. \eqref{integral_appendix} with equality.}
	 \end{figure}
 
	 Our proof is based on constructing a strictly-increasing function as follows: suppose $ \epsilon \in (0,p-\underline{c}] $ and $ \xi \in (0,\psi_p] $. We assume that $ \psi_*(\eta) $ is strictly-increasing in $ \eta\in [\underline{c}, \underline{c}+\epsilon] $ with $ \psi_*(\underline{c}) = 0 $ and $ \psi_*(\underline{c}+\epsilon) =  \psi_p -  \xi $; $ \psi_*(\eta) $ is strictly-increasing in $ \eta\in [\underline{c}+\epsilon, \overline{p}] $ with $ \psi_*(\underline{c}+\epsilon) = \psi_p -  \xi $ and $ \psi_*(p) = \psi_p $. For such a function $ \psi_*(\eta) $, when $ \xi = 0 $, we have 
	 \begin{align*}
	 & p\psi_*(p) -f(\psi_*(p)) - \int_{\underline{c}}^{p} \psi_*(\eta)d\eta\\
	 =\ &p\psi_*(p) -f(\psi_*(p)) -\psi_*(p)(p-\underline{c} - \epsilon) - \int_{\underline{c}}^{\underline{c}+\epsilon} \psi_*(\eta)d\eta\\
	 =\ & \psi_*(p)(\underline{c}+\epsilon) -f(\psi_*(p)) - \int_{\underline{c}}^{\underline{c}+\epsilon} \psi_*(\eta)d\eta.
	 \end{align*}
	 In particular, when $ \epsilon $ approaches 0, based on the mean value theorem, we have
	 \begin{align*}
	 p\psi_*(p) -f(\psi_*(p)) - \int_{\underline{c}}^{p} \psi_*(\eta)d\eta = p\psi_*(p) -f(\psi_*(p)) < 0 \leq \frac{1}{\alpha} f_{\texttt{\#}}(p).
	 \end{align*}
	
	 On the other hand, $ \psi(\eta) $ is a feasible solution to Eq. \eqref{integral_appendix} means that we have at least $ \int_{\underline{c}}^{p}\psi_*(\eta)d\eta = \int_{\underline{c}}^{p}\psi(\eta)d\eta $ so that $ p\psi_*(p) -f(\psi_*(p)) - \int_{\underline{c}}^{p} \psi_*(\eta)d\eta \geq    \frac{1}{\alpha} f_{\texttt{\#}}(p) $. Therefore, it is always possible to adjust the values of $ \epsilon \in (0,p-\underline{c}] $ and $ \xi \in (0,\psi_p] $ to get a strictly-increasing function $ \psi_*(\eta) $ so that Eq. \eqref{psi_star_equality} holds, namely, Eq. \eqref{integral_appendix} holds with equality. Hence, we complete the proof of Theorem \ref{necessity}.  
\end{proof}

Notice that, based on Theorem \ref{sufficiency} and Theorem \ref{necessity}, if we assume $ p = \phi(y) $ and $ y = \psi(p) $,  then $ \phi $ and $\psi $ are inverse to each other since $ \phi $ and $ \psi $ are both strictly increasing. In particular, the following two equations are basically equivalent to each other:
\begin{align}
\int_{\underline{c}}^{p} \eta \psi'(\eta)d\eta - f\big(\psi(p)\big)  =  \frac{1}{\alpha} f_{\texttt{\#}}(p) \Leftrightarrow \int_{0}^{y}\phi(\eta)d\eta - f(y) = \frac{1}{\alpha}  f_{\normalfont\texttt{\#}}(\phi(y)).
\end{align}
Therefore, if there exits an $ \alpha $-competitive online algorithm, there must exit a strictly-increasing function $ y= \psi(p) $ that satisfies Eq. \eqref{integral_equality} for all $ p\in [\underline{c},\overline{p}] $, and the inverse of $ y = \psi(p) $, denoted by $ p = \psi^{-1}(y) $, is the pricing function that satisfies the conditions in Theorem \ref{sufficiency}.  \textbf{Therefore, we complete the proof of the necessary conditions in Theorem \ref{a_unified_BVP}}.

\subsection{Proof of Theorem \ref{a_unified_BVP_general}}\label{proof_a_unified_BVP_general}
The proof of Theorem \ref{a_unified_BVP_general} is similar to that of Theorem \ref{a_unified_BVP}. In particular, the following two corollaries directly follow Theorem \ref{sufficiency} and Theorem \ref{necessity}, respectively.

\begin{corollary}[\textbf{Sufficiency}]\label{sufficiency_general}
	Given a setup $ \mathcal{S} $ with $ \overline{p}_k\in (\overline{c}_k,+\infty), \forall k\in\mathcal{K}$, $\normalfont \text{PPM}_{\bm{\phi}} $ is $ \max\{\alpha_k\} $-competitive if  $\bm{\phi}=\{\phi_k\}_{\forall k} $  and for all $ k\in\mathcal{K} $, the pricing function $ \phi_k $ satisfies
	\begin{align}\label{ODE_sufficiency_general}
	\begin{cases}
	\phi_k(y) - f_k'(y) = \frac{1}{\alpha_k}\cdot \frac{d f_{\normalfont\texttt{\#}}^k(\phi_k(y)) }{d\phi_k}, y\in (0,1),\\
	 \phi_k(0) =  \underline{c}_k, \phi_k(1)\geq \overline{p}_k.
	\end{cases}
	\end{align} 
\end{corollary}

\begin{corollary}[\textbf{Necessity}]\label{necessity_general}
	Given a setup $ \mathcal{S} $ with $ \overline{p}_k\in (\overline{c}_k,+\infty), \forall k\in\mathcal{K}$, if there exists an $ \alpha $-competitive online algorithm, then for all $ k\in\mathcal{K} $ there must exist a strictly increasing function $ \psi_k(\eta) $ and a constant $ \alpha_k\in [1,\alpha] $ that satisfy:   
	\begin{equation}\label{integral_equality_general}
	\begin{cases}
	\int_{0}^{p} \eta \psi_k'(\eta)d\eta - f_k\big(\psi_k(p)\big)  =   \frac{f_{\normalfont \texttt{\#}}^k(p)}{\alpha_k},\forall p\in (\underline{c}_k,\overline{p}_k),\\
	\psi_k\left(\underline{c}_k\right) = 0, \psi_k(\overline{p}_k) \leq  1. 
	\end{cases}
	\end{equation}
\end{corollary}

We skip the proofs of the above two corollaries since they follow the same principle as our previous proof of Theorem \ref{a_unified_BVP}. Theorem \ref{a_unified_BVP_general} directly follows the above two corollaries. Note that here we only consider the case of \textsf{HUC}. The discussion of \textsf{LUC} is similar and thus is omitted for brevity.

\section{Proof of Theorem \ref{lower_bound_LUC} and Theorem \ref{lower_bound_1}}
\label{proof_lower_bound_1}

\subsection{Preliminaries}
We first give some preliminaries to aid our following proofs of the two lower bounds in Theorem \ref{lower_bound_LUC} and Theorem \ref{lower_bound_1}.

\subsubsection{Characteristic Polynomial}
The first step of our lower bound analysis is to show that the ODE of Problem \eqref{BVP_power_LUC} and Problem \eqref{two_BVP_power_1}, i.e.,  the following ODE
\begin{align}\label{original_ODE}
\phi'(y) =  \alpha\cdot \frac{\phi(y) - f'(y)}{(\phi(y)/\overline{c})^{\frac{1}{s-1}}},
\end{align}
can be expressed in a separable form of differential equations. In particular, when we assume $ \varphi = (\phi/\overline{c})^{\frac{1}{s-1}} $, we have
\begin{align}\label{def_ODE}
\varphi' = \alpha \frac{\varphi^{s-1} - y^{s-1}}{(s-1)\varphi^{s-1}} =  \alpha\cdot
\frac{\left(\varphi/y\right)^{s-1} - 1}{(s-1)\left(\varphi/y\right)^{s-1}}.
\end{align}
Let us assume $ \chi = \varphi/y $, then the ODE in Eq. \eqref{def_ODE} becomes
\begin{align}\label{ODE_chi}
\frac{-\chi^{s-1}}{\chi^s-\frac{\alpha}{s-1}\chi^{s-1}+\frac{\alpha}{s-1}} d\chi = \frac{1}{y}dy.
\end{align}
Taking integration on both sides of Eq. \eqref{ODE_chi} leads to 
\begin{align}\label{integration_equation}
\int_0^{\chi}\frac{-\eta^{s-1}}{\eta^s-\frac{\alpha}{s-1}\eta^{s-1}+\frac{\alpha}{s-1}} d\eta = \ln(y)+C,
\end{align}
where $ C $ is any real constant.  Let us define $ P_s\left(\eta;\alpha\right) $  as
\begin{align}\label{cp}
P_{s}\left(\eta;\alpha\right) \triangleq
\eta^s-\frac{\alpha}{s-1}\eta^{s-1}+\frac{\alpha}{s-1}.
\end{align}
Note that Eq. \eqref{cp}  is the denominator of the left-hand-side of Eq. \eqref{integration_equation}.  This polynomial is referred to as the \textbf{characteristic polynomial} hereinafter, where the notation $ P_{s}\left(\eta;\alpha\right)  $ means that the characteristic polynomial is in degree $ s $ with variable $ \eta $ for a given $ \alpha\geq 1 $. 

The characteristic polynomial plays a critical role in our following lower bound analysis of $ \alpha $. In particular, the existence of positive roots to equation $ P_s\left(\eta;\alpha\right) = 0 $ is summarized in the following Lemma \ref{positive_roots}.

\begin{lemma}\label{positive_roots}
	Given $ \alpha\geq 1 $ and $ s> 1 $,  $P_{s}\left(\eta;\alpha\right) = 0 $ has at most two positive roots in variable $ \eta $. In particular, when $ \alpha < \alpha_s^{\min} $, $P_{s}\left(\eta;\alpha\right) = 0 $ has no positive root; when $ \alpha > \alpha_s^{\min} $, $P_{s}\left(\eta;\alpha\right) = 0 $ has two positive roots; when $ \alpha = \alpha_s^{\min} $, $P_{s}\left(\eta;\alpha\right) = 0 $ has a double positive root. 
\end{lemma}
\begin{proof}
	We can prove that the characteristic polynomial is a unimodal function in $ \eta\in [0,+\infty) $. Taking derivative of $  P_{s}\left(\eta,\alpha \right)  $ w.r.t. $ \eta\in [0,+\infty)$, we have
	\begin{align}
	\frac{d P_{s}\left(\eta;\alpha \right)}{d\eta} = s\eta^{s-1}-\alpha \eta^{s-2} = \eta^{s-2}\left(s\eta-\alpha \right).
	\end{align}
	Therefore, $P_{s}\left(\eta,\alpha\right) $ is decreasing when $\eta\in [0,\alpha/s]  $, and is increasing when $ \eta\in (\alpha/s,+\infty) $. Since $ P_s(0;\alpha) = \frac{\alpha}{s-1}>0 $, to have at least one positive root, we must have
	\begin{align}
	P_{s}\left(\frac{\alpha}{s};\alpha\right) = \left(\frac{\alpha}{s}\right)^s-\frac{\alpha}{s-1}\left(\frac{\alpha}{s}\right)^{s-1}+\frac{\alpha}{s-1} \leq 0,
	\end{align}
	which thus leads to  
	$ \alpha\geq   s^{s/(s-1)} = \alpha_{s}^{\min} $. In particular, when $ \alpha = \alpha_{s}^{\min} $, we have a double positive root, which is $ \eta = \alpha_{s}^{\min}/s = s^{1/(s-1)}$.   
\end{proof}

Based on  Lemma \ref{positive_roots}, for any $ \alpha \geq \alpha_{s}^{\min}$, we denote the two positive roots of $ P_{s}\left(\eta;\alpha\right) = 0 $  by $ R_{s}^{-}(\alpha) $ and $ R_{s}^{+}(\alpha) $, where  $ R_{s}^{-}(\alpha)\leq  R_{s}^{+}(\alpha) $. In particular, when $ \alpha =  \alpha_s^{\min} $, we have $ R_{s}^{-}(\alpha)= R_{s}^{+}(\alpha) $ and thus we have a double positive root. We will see in subsequent sections that the positive roots of the characteristic polynomial play a critical role throughout the proofs of Theorem \ref{lower_bound_LUC} and Theorem \ref{lower_bound_1}.

\subsubsection{Preliminaries of IVP and BVPs}
To prove the existence of monotonically-increasing solutions to Problem \eqref{BVP_power_LUC} and Problem \eqref{two_BVP_power_1},  let us first focus on the following BVP:
\begin{align}\label{BVP_1_varphi}
\text{\textbf{BVP}}(\varphi; \alpha,u) 
\begin{cases}
\varphi'  = \alpha\cdot \frac{\varphi^{s-1} - y^{s-1}}{(s-1)\varphi^{s-1}}, 0<y<u, \\
\varphi(0) = 0, \varphi(u) = 1.
\end{cases}
\end{align} 
For any $ \alpha \geq 1 $ and $ u\in (0,1) $, we denote the solution of $\text{\textbf{BVP}}(\varphi; \alpha,u) $ (if it exists) by $ \varphi_{\textsf{bvp}} \big(y;\alpha,u\big) $.  Note that the only difference between $ \text{\textbf{BVP}}(\varphi; \alpha,u)  $ and Problem \eqref{two_BVP_power_1} is the rescaling of the coordinates. In the following, we may refer to these two BVPs interchangeably. We will call $ \phi(y) $ as the (original) pricing function while $ \varphi(y) $ the \textbf{scaled-pricing function}. Similarly, we will also call $ f'(y)=\overline{c}y^{s-1} $ the (original)  marginal cost and $ f_{\varphi}'(y) = (f'(y)/\overline{c})^{1/(s-1)} = y $ as the \textbf{scaled-marginal cost}. We will see in the subsequent analyses that performing such an equivalent transformation helps us reveal rich structural properties of the ODE in Eq. \eqref{original_ODE}.

Directly working with BVPs is usually very challenging \cite{Perko2001}. Worse yet is that our $ \text{\textbf{BVP}}(\varphi; \alpha,u) $ consists  of a singular boundary condition, namely the right-hand-side of the ODE is undefined when $ \varphi(0) = 0 $. A typical idea is to approach  $ \text{\textbf{BVP}}(\varphi; \alpha,u) $ via its associated IVP, and thus we define $ \text{\textbf{IVP}}(\varphi; \alpha,u)  $ as follows:
\begin{align}\label{IVP_1_varphi}
\text{\textbf{IVP}}(\varphi; \alpha,u) 
\begin{cases}
\varphi'  = \alpha\cdot \frac{\varphi^{s-1} - y^{s-1}}{(s-1)\varphi^{s-1}}, 0<y<u, \\
\varphi(u) = 1,
\end{cases}
\end{align} 
We denote the solution of $\text{\textbf{IVP}}(\varphi; \alpha,u) $ (if it exists) by $ \varphi_{\textsf{ivp}}\big(y;\alpha,u\big) $. Intuitively,  when we have 
$$ \lim_{y\rightarrow 0^+} \varphi_{\textsf{ivp}}(y;\alpha,u) = 0, $$
then $ \varphi_{\textsf{ivp}}\left(y;\alpha,u\right) $ is also a solution to  $ \text{\textbf{BVP}}(\varphi; \alpha,u)  $, i.e., $ \varphi_{\textsf{ivp}}\left(y;\alpha,u\right) = \varphi_{\textsf{bvp}} (y;\alpha,u)  $. Note that we check the limit of $ \varphi_{\textsf{ivp}}\left(y;\alpha,u\right)  $ when $ y $ approaches 0 since   $ \varphi_{\textsf{ivp}}\left(y;\alpha,u\right) $ may be undefined at $ y = 0 $. 

Solving $\text{\textbf{IVP}}(\varphi; \alpha,u) $ is trivial since we only have one initial condition. In particular, substituting the initial condition of $ \varphi(u) = 1 $ into Eq. \eqref{integration_equation} leads to
\begin{align}
\int_0^{1/u}\frac{-\eta^{s-1}}{P_s\left(\eta,\alpha\right)} d\eta = \ln(u)+C,
\end{align}
which thus indicates that $ \varphi_{\textsf{ivp}}\left(y;\alpha,u\right) $ is the root to the following equation in variable $ \varphi $:
\begin{align}\label{equation_of_varphi}
\int_{1/u}^{\varphi/y}\frac{\eta^{s-1}}{P_s\left(\eta,\alpha\right)} d\eta = \ln\left(\frac{u}{y}\right).
\end{align}
Below we give some standard results regarding the existence, uniqueness and monotonicity of $ \varphi_{\textsf{ivp}}\left(y;\alpha,u\right) $.

\subsubsection{Existence, Uniqueness and Monotonicity} Our pricing function design is related to the existence and uniqueness property of solutions to $ \text{\textbf{IVP}}(\varphi; \alpha,u) $  in the following lemma.   

\begin{lemma}\label{existence_uniquess}
	For each $ (\alpha, u)\in [1,+\infty)\times (0,1) $,  $\normalfont\text{\textbf{IVP}}(\varphi; \alpha,u) $ has a unique solution  $ \varphi_{\textsf{ivp}}\left(y;\alpha,u\right)  $ that is defined over $y\in (0,u] $.
\end{lemma}

Lemma \ref{existence_uniquess} follows one of the most important theorems in ODEs, namely the  Picard–Lindel\"{o}f theorem for the existence and uniqueness of solutions to IVPs.  We refer the details to \cite{Perko2001}, \cite{ODE1973}, \cite{uniqueness_book1993}. Basically the Picard–Lindel\"{o}f theorem guarantees that there always exists a unique solution to $ \text{\textbf{IVP}}(\varphi; \alpha,u)  $, defined on a small neighbourhood of the initial point $ \varphi(u) = 1 $, as long as the right-hand-side of the ODE in $ \text{\textbf{IVP}}(\varphi; \alpha,u)  $ is Lipschitz continuous within that neighbourhood. Moreover, this unique solution extends to the whole region of $ y\in (0,u] $. Based on this existence and uniqueness property, we can prove the following monotonicity properties in Lemma \ref{monotonicity_omega} and Lemma \ref{monotonicity_alpha}.

\begin{lemma}\label{monotonicity_omega}
	Given $ \alpha\geq 1 $, $ \varphi_{\textsf{ivp}}(y;\alpha,u)$ is non-decreasing in $ y\in(0,u] $ and lower bounded by $ f_\varphi'(y) $ at each point in $ (0,u] $.
\end{lemma}
\begin{proof}
	Please refer to Appendix \ref{proof_monotonicity_omega}.
\end{proof}

Lemma \ref{monotonicity_omega} guarantees that for any $ \alpha\geq 1 $ and $ u\in (0,1) $, the unique solution to $ \text{\textbf{IVP}}(\varphi; \alpha,u)  $ is a feasible scaled-pricing function (i.e., posted prices are always larger than or equal to the marginal costs, and thus no negative social welfare contribution will be introduced). Below we give Lemma \ref{monotonicity_alpha}, which states that $ \varphi_{\textsf{ivp}}(y;\alpha,u)  $ is also monotonic in $ (\alpha, u)\in [1,+\infty)\times (0,1] $.

\begin{lemma}\label{monotonicity_alpha}
	$ \varphi_{\textsf{ivp}}(y;\alpha,u) $ is  continuous and non-increasing   in $ (\alpha, u)\in [1,+\infty)\times (0,1] $.
\end{lemma}
\begin{proof}
	The proof is given in Appendix \ref{proof_monotonicity_alpha}.
\end{proof}

We also have the following Lemma \ref{unique_BVP}, which shows that if $ \varphi_{\textsf{ivp}}(y;\alpha,u) $ approaches 0 when $ y \rightarrow 0^+ $, then it must be the unique solution to $ \text{\textbf{BVP}}(\varphi; \alpha,u)  $.

\begin{lemma}\label{unique_BVP}
	For any $ u\in (0,1) $ and $ \alpha\geq 1 $, $ \varphi_{\textsf{ivp}}(y;\alpha,u) $ is the unique solution to $\normalfont \text{\textbf{BVP}}(\varphi; \alpha,u) $ if and only if $ \lim\limits_{y\rightarrow 0^+}\varphi_{\textsf{ivp}}(y;\alpha,u) = 0 $. 
\end{lemma}
\begin{proof}
	The necessity is obvious, and the sufficiency  can be proved by contradiction. Since for a given $ (\alpha, u)\in [1,+\infty)\times (0,1) $, there exists a unique solution to $ \text{\textbf{IVP}}(\varphi; \alpha,u) $, and thus if $ \lim\limits_{y\rightarrow 0^+}\varphi_{\textsf{ivp}}(y;\alpha,u) = 0 $ and $ \varphi_{\textsf{ivp}}(y;\alpha,u) $ is not the unique solution for $ \text{\textbf{BVP}}(\varphi; \alpha,u) $, then there must exist another solution for $ \text{\textbf{IVP}}(\varphi; \alpha,u) $, leading to a contradiction with Lemma \ref{existence_uniquess}.
\end{proof}

Note that Lemma \ref{unique_BVP} does not directly states any condition to show the existence of solutions to $ \text{\textbf{BVP}}(\varphi; \alpha,u) $ in terms of $ \alpha $ and $ u $. In fact  it is unclear at the moment whether there exists a feasible design of $ (\alpha, u) $ so that $\lim\limits_{y\rightarrow 0^+} \varphi_{\textsf{ivp}}(y;\alpha,u) = 0 $. We answer this question in the next section.

\subsection{Structural Properties}\label{structural_properties}
Based on the characteristic polynomial, below we give an important structural property of $\text{\textbf{IVP}}(\varphi; \alpha,u) $.
\begin{proposition}\label{convexity_of_varphi}
	For any $ u\in (0,1) $ and $ \alpha > \alpha_s^{\min} $,  $ \varphi_{\textsf{ivp}}\left(y;\alpha,u\right) $ has the following properties:
	\begin{itemize}
		\item If $ \alpha = \alpha_{s}(u) $, $ \varphi_{\textsf{ivp}}(y;\alpha,u)$ is linear in $ y \in(0,u] $, given by
		\begin{align}\label{linear_solution_u}
		\varphi_{\textsf{ivp}}\left(y; \alpha_{s}(u),u\right) = \frac{y}{u}.
		\end{align} 
		
		\item If $ \alpha >  \alpha_{s}(u)  $, $ \varphi_{\textsf{ivp}}(y;\alpha,u)$ is strictly convex in $ y \in(0,u] $.
		
		\item If $ \alpha < \alpha_{s}(u) $, $  \varphi_{\textsf{ivp}}\left(y;\alpha,u\right) $ is strictly concave in $ y \in(0,u] $.
	\end{itemize}
	Recall that $ \alpha_s(u) = \frac{s-1}{u-u^{s}} $, which is defined in Eq. \eqref{def_alpha_gamma}. 
\end{proposition}
\begin{proof}
	Please refer to Appendix \ref{proof_of_convexity} for the detailed proof.	We emphasize that $ \varphi $ being linear does not necessarily mean the original pricing function $ \phi $ is linear since $ \phi(y) = \overline{c} \left(\varphi_{\textsf{ivp}}(y;\alpha,u)\right)^{s-1} $. The same argument also applies to the convexity and concavity of $ \phi $ and $ \varphi $. 
\end{proof}

\begin{corollary}\label{two_linear_solutions}
	For any $ \alpha > \alpha_s^{\min} $, $ \varphi_{\textsf{ivp}}(y;\alpha,u) $ is given by
	\begin{align}\label{linear_solution_alpha}
	\varphi_{\textsf{ivp}}(y;\alpha,u) =
	\begin{cases}
	y R_{s}^{+}(\alpha)
	&\text{if } u = \frac{1}{R_{s}^{+}(\alpha)} \in [0,u_s],\\
	y R_s^{-}(\alpha) &\text{if } u = \frac{1}{R_s^{-}(\alpha)} \in [u_s,1].
	\end{cases} 
	\end{align}
	In particular, when $ \alpha = \alpha_s^{\min} $,  the two linear solutions reduce to one as follows: $ \varphi_{\textsf{ivp}}(y;\alpha,u) = y/u_s $, that is, $ R_s^{-}(\alpha_s^{\min}) = R_s^{+}(\alpha_s^{\min}) = 1/u_s $.
\end{corollary}
\begin{proof}
	Eq. \eqref{linear_solution_u} and \eqref{linear_solution_alpha} are basically equivalent to each other. In fact, if we substitute $ \varphi_{\textsf{ivp}}\left(y; \alpha,u\right) = y/u $ back in the ODE, we have $ P_s(\frac{1}{u};\alpha) = 0 $, and thus the corollary follows. 	
\end{proof}

\begin{figure}
	\centering
	\includegraphics[width = 6.5 cm]{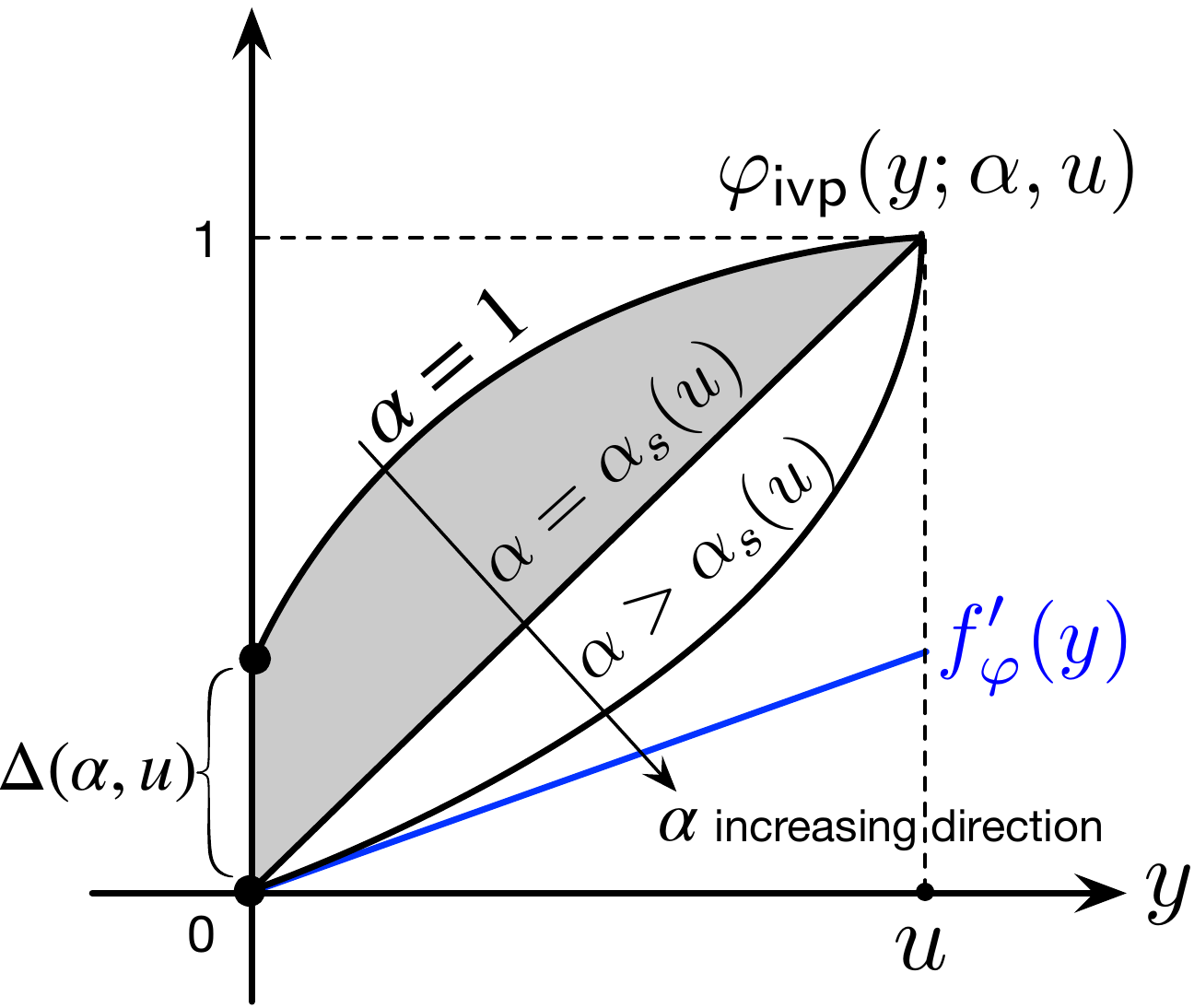}
	\caption{Illustration of $ \varphi_{\textsf{ivp}}(y;\alpha,u)$ in three cases when $ \alpha =1 $,  $\alpha =  \alpha_s(u)  $ and $ \alpha > \alpha_s(u) $. The scaled-marginal cost $ f'_\varphi(y) = y $.  The grey region is for the cases when $ 1\leq \alpha \leq \alpha_s(u)  $. }
	\label{bounded_L}
	\vspace{-0.4cm}
\end{figure}

We illustrate the above properties of $ \varphi_{\textsf{ivp}}\left(y; \alpha,u\right) $ in Fig. \ref{bounded_L}. It is obvious that  $ \varphi_{\textsf{ivp}}\left(y; \alpha_{s}(u),u\right) = y/u $  satisfies both the ODE and the two boundary conditions, namely
$ \varphi_{\textsf{ivp}}\left(0; \alpha_s(u),u\right) = 0 $ and $ \varphi_{\textsf{ivp}}\left(u; \alpha_s(u),u\right) = 1 $. Therefore, $ \varphi_{\textsf{ivp}}\left(y; \alpha_s(u),u\right) $ is the unique solution to $\text{\textbf{BVP}}(\varphi; \alpha_s(u),u) $. Since our target is to get an as small $ \alpha $ as possible, it is interesting to know whether there exists any  $\alpha\in [1,\alpha_s(u))$ such that $ \varphi_{\textsf{ivp}}\left(y; \alpha,u\right) $ is also the unique solution to $\text{\textbf{BVP}}(\varphi; \alpha,u) $ for some $ u\in (0,1) $, namely the grey region in Fig. \ref{bounded_L}.  

For notational convenience, let us define $ \Delta(\alpha,u) $ and $ \Delta'(\alpha,u) $ by:
\begin{align}
\Delta(\alpha,u) \triangleq \lim_{\omega\rightarrow 0^+} \varphi_{\textsf{ivp}}(y;\alpha,u),  \Delta'(\alpha,u) \triangleq \lim_{y\rightarrow 0^+} \varphi'_{\textsf{ivp}}(y;\alpha,u).
\end{align}
Therefore, $ \Delta(\alpha,u) $ captures the distance between $ \varphi_{\textsf{ivp}}(y;\alpha,u) $ and $ f'_\varphi(y) $ when $ y \rightarrow  0^+ $, as shown in Fig. \ref{bounded_L}. Recall that the necessary condition for $ \varphi_{\textsf{ivp}}\left(y; \alpha,u\right) $ being the unique solution to $\text{\textbf{BVP}}(\varphi; \alpha,u) $ is to have $ \Delta(\alpha,u) = 0$.  Below we give a proposition to show the necessary condition for $ \Delta(\alpha,u) = 0$ in terms of $ \alpha $ and $ u $. 

\begin{proposition}\label{limit_finite}
	For any $ u\in (0,1] $ and $ 1\leq \alpha \leq  \alpha_s(u) $, we have
	\begin{align}
	0\leq \Delta(\alpha,u)   \leq 1 - \frac{\alpha}{\alpha_s(u)}.
	\end{align}
	Meanwhile, if $ \Delta(\alpha,u) = 0 $, 
	then  $ P_s(\Delta'(\alpha,u);\alpha) = 0 $.  
\end{proposition}
\begin{proof}
	The proof is given in Appendix \ref{proof_limit_finite}.
\end{proof}

The above proposition shows that for any given $ u\in (0,1) $ and $ 1\leq \alpha \leq  \alpha_s(u) $, $ \Delta(\alpha,u) $ can be any value within $ [0, 1 - \alpha/\alpha_s(u)]$. In particular, when $ \Delta(\alpha,u) = 0$, we have $ P_\gamma(\Delta'(\alpha,u),\alpha)= 0 $, namely $ \Delta'(\alpha,u) $ is one of the positive roots of the characteristic polynomial. Note that as a special case, when $u = \frac{1}{R_s^{\text{+}}(\alpha)}$ or $\frac{1}{R_s^{\text{-}}(\alpha)} $, the proposition clearly holds according to Corollary \ref{two_linear_solutions}. 

The above necessary condition for $ \Delta(\alpha,u) = 0 $ is very useful in our following lower bound analysis. In fact, based on Proposition \ref{limit_finite}, we immediately have the following two corollaries. 

\begin{corollary}\label{no_solution_BVP1}
	For all $ u\in (0,1) $, $\normalfont\text{\textbf{BVP}}(\varphi;\alpha, u)$ has no solution if $ \alpha<\alpha_s^{\min}$. 
\end{corollary}
\begin{corollary}\label{global_lower_bounded}
	For any $ \epsilon> 0 $, there are no $ (\alpha_s^{\min}-\epsilon) $-competitive online algorithms.
\end{corollary}
\begin{proof}
	The above two corollaries are equivalent to each other. Note that if $  \Delta(\alpha,u) = 0 $, then $\Delta'(\alpha,u) $ must be a positive root of the characteristic polynomial. However, according to Lemma \ref{positive_roots}, when $ \alpha  <  \alpha_s^{\min} $ the characteristic polynomial has no positive root. Therefore, for any $ u\in (0,1)$, we have $  \Delta( \alpha,u) \neq  0 $ if $ \alpha<\alpha_s^{\min} $,  which implies that $\text{\textbf{BVP}}(\varphi;\alpha,u)$  has no solution, and thus no online algorithm can achieve a competitive ratio that is smaller than $ \alpha_s^{\min} $ with zero additive loss\footnote{Note that Corollary \ref{global_lower_bounded} is not a new result and was first proved in \cite{Huang2018}. However, here we provide a new proof, which is much simpler and more intuitive.}.	
\end{proof}

\subsection{Lower Bound (Proof of Theorem \ref{lower_bound_1})}
\label{proof_lower_bound_limited}
This section presents the formal proof of Theorem \ref{lower_bound_1}.  We show that for each given $ u\in (0,1) $, the necessary and sufficient  condition for the existence of a unique solution to $ \text{\textbf{BVP}}(\varphi; \alpha,u)  $ is to have $ \alpha  \geq \underline{\alpha}_1(u)$, where $ \underline{\alpha}_1(u) $ is given in Eq. \eqref{def_lower_bound} and is revisited as follow:
\begin{align}
\underline{\alpha}_1(u) = 
\begin{cases}
\alpha_s(u)  & \text{ if } u \in\left(0,u_s\right),\\
\alpha_s^{\min} & \text{ if } u\in\left[u_s,1\right).
\end{cases}
\end{align} 
We illustrate $  \underline{\alpha}_1(u) $ in Fig. \ref{four_cases} to help our following analysis. Below we give the details of the proof.

\begin{figure}
	\centering
	\includegraphics[width = 7.5 cm]{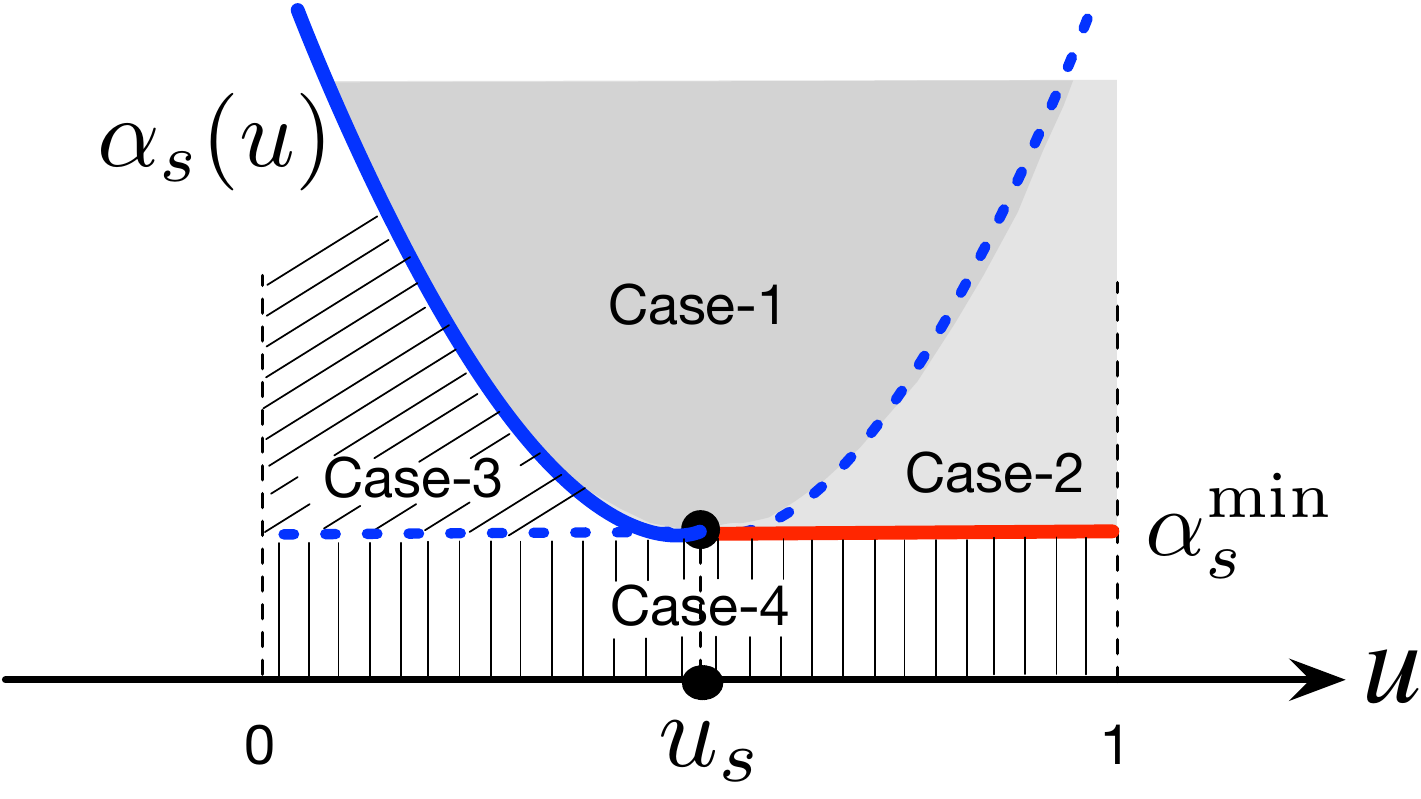}
	\caption{Illustration of the lower bound $ \underline{\alpha}_1(u) $ (the blue real curve and the red real curve) with four cases to prove Theorem \ref{lower_bound_1}. In the figure, the blue dashed curve denotes function $ \alpha_s(u) = \frac{s-1}{u-u^s}$ when $ u\in [u_s,1) $. }
	\label{four_cases}
\end{figure}

We first prove the sufficiency of Theorem \ref{lower_bound_1}, that is, if $ \alpha \geq \underline{\alpha}_1(u) $, then we can find a unique solution for $\text{\textbf{BVP}}(\varphi;\alpha,u)$, namely \textbf{Case-1} and \textbf{Case-2} in Fig. \ref{four_cases}. 
	
	\underline{\textbf{Case-1}: $ \alpha \geq  \alpha_s(u) $ and $ u\in (0,1) $.} Based on Lemma \ref{monotonicity_omega}, Lemma \ref{monotonicity_alpha} and Proposition \ref{convexity_of_varphi}, the unique solution to $\text{\textbf{IVP}}(\varphi;\alpha,u)$ is convex and stays within the following range
	\begin{align}
	y R_s^-(\alpha)\leq \varphi_{\textsf{ivp}}\left(y; \alpha,u\right) \leq \frac{y}{u}, \forall y\in (0,u].
	\end{align}
	The upper limit holds because of the monotonicity w.r.t. $ \alpha $, i.e.,  $ \varphi_{\textsf{ivp}}\left(y; \alpha,u\right)\leq \varphi_{\textsf{ivp}}\left(y; \alpha_s(u),u\right) = y/u $,
	while for the lower bound, we can prove it by contradiction. Suppose $ \varphi_{\textsf{ivp}}\left(y; \alpha,u\right) $ is not lower bounded by $ y R_s^-(\alpha) $, then these two functions must have at least one intersection points (other than the singular point). Let us denote one of these intersection points by $ y = u_0 $. Then it is easy to see that  $\text{\textbf{IVP}}(\varphi;\alpha,u_0)$ has two solutions: one is linear and the other one is convex, which contradicts with the uniqueness property of $\text{\textbf{IVP}}(\varphi;\alpha,u_0)$. Therefore, the lower bound holds. Based on the squeeze theorem, when $ y \rightarrow 0^+$, we have  
	$  \Delta(\alpha,u) = \lim_{y\rightarrow 0^+} \varphi_{\textsf{ivp}}\left(y; \alpha,u\right) = 0 $, 
	which means that $ \varphi_{\textsf{ivp}}\left(y; \alpha,u\right) $ is the unique solution to $\text{\textbf{BVP}}(\varphi;\alpha,u)$. 
	
	\underline{\textbf{Case-2}: $  \alpha_s^{\min}\leq  \alpha \leq   \alpha_s(u) $ and $ u\in [u_s,1) $.}	
	We argue that for any $ y\in (0,u] $, the unique solution $ \varphi_{\textsf{ivp}}\left(y; \alpha,u\right) $ is concave and stays within $ [y/u,y R_s^-(\alpha)] $, namely,
	\begin{align}
	\frac{y}{u}\leq \varphi_{\textsf{ivp}}\left(y; \alpha,u\right) \leq y R_s^-(\alpha), \forall y\in (0,u].
	\end{align}
	As illustrated by the concave curve \texttoptiebar{$ AD $} in Fig. \ref{concave}(a),   the lower bound is represented by  the straight line $ \overline{AD} $ and the upper bound is represented by  $ \overline{AC} $. The lower limit follows the monotonicity of $ \varphi_{\textsf{ivp}}\left(y; \alpha,u\right) $ in $ \alpha $, i.e., 
	$\varphi_{\textsf{ivp}}\left(y; \alpha,u \right)\geq \varphi_{\textsf{ivp}}\left(y; \alpha_s(u),u\right) = y/u $, and we can prove  the upper limit by contradiction in the same way as \textbf{Case-1}. Based on the lower bound and upper bound of $ \varphi_{\textsf{ivp}}\left(y; \alpha,u\right) $, we  have 
	$ \Delta(\alpha,u) = \lim_{y\rightarrow 0^+} \varphi_{\textsf{ivp}}\left(y; \alpha,u\right) = 0 $, 
	and thus $ \varphi_{\textsf{ivp}}\left(y; \alpha,u\right) $ is the unique solution to $\text{\textbf{BVP}}(\varphi;\alpha,u)$.  Moreover, in this case, we have
	\begin{align}
	\Delta'(\alpha,u) = \lim\limits_{y\rightarrow 0^+} \varphi'_{\textsf{ivp}}\left(y; \alpha,u\right) =   R_s^-(\alpha)= \frac{1}{u_{\texttt{\#}}},
	\end{align}
	where $ u_{\texttt{\#}}\in [u_s,u] $ is such that $ \alpha_s(u_{\texttt{\#}}) = \alpha $. This means that in Fig. \ref{concave}(a), $ \overline{AC} $ is a tangent line for \texttoptiebar{$ AD $} at point $ y = 0 $.  In particular, if $ \alpha = \alpha_s^{\min} $, then the upper bound of $ \varphi_{\textsf{ivp}}\left(y; \alpha,u\right) $ is $ y/u_s $, as illustrated by $ \overline{AB} $ in Fig. \ref{concave}(a). When $ \alpha = \alpha_s(u) $, both the upper bound and lower bound become $ y/u $, namely, $ \varphi_{\textsf{ivp}}\left(y; \alpha_s(u),u\right) =y/u $, which follows Proposition \ref{convexity_of_varphi}.
	
	\begin{figure}
		\centering
		\includegraphics[width = 12.5cm]{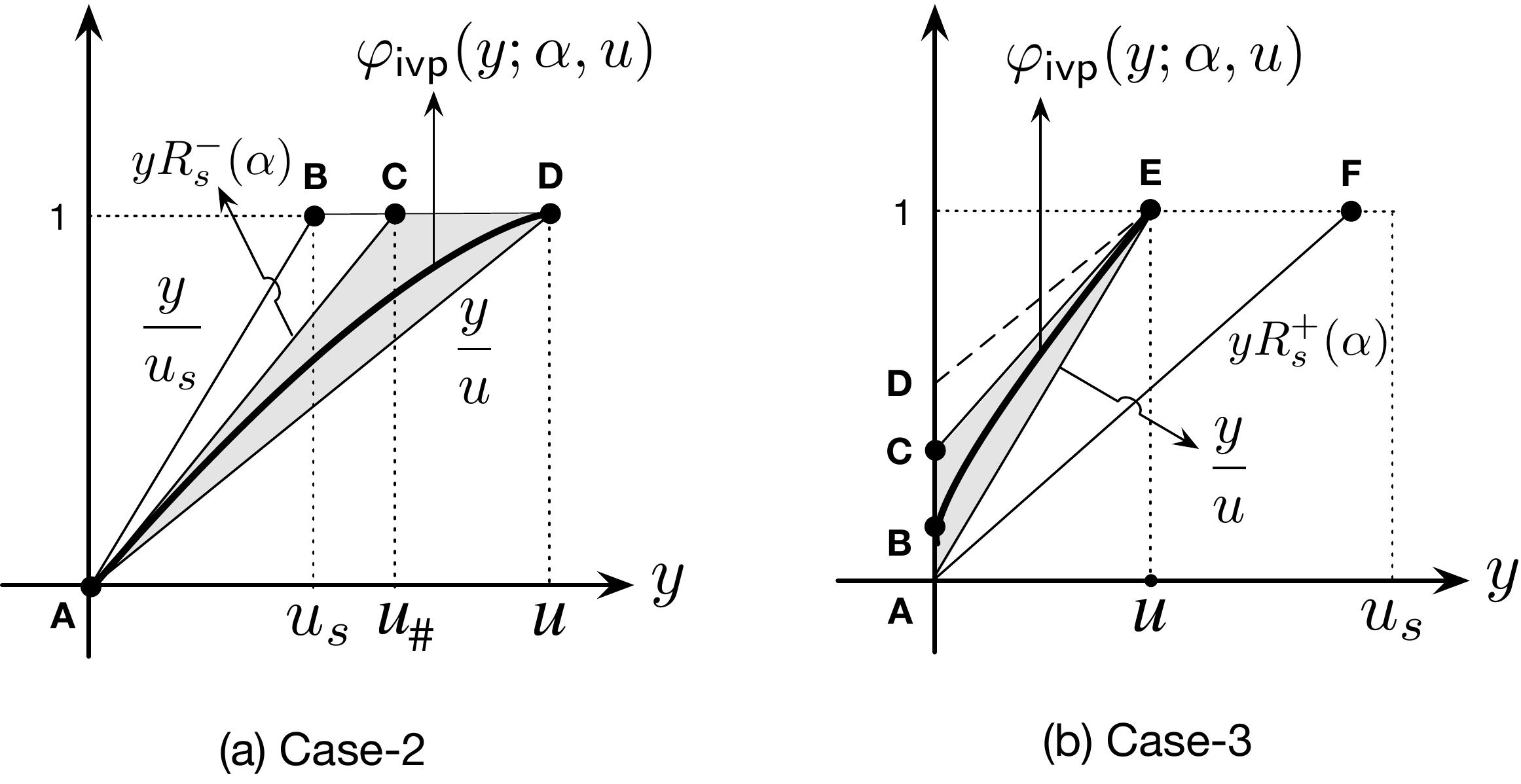}
		\caption{
			The scaled-pricing function in Case-2  and Case-3. In subfigure (a), the slopes of $ \overline{AB} $ and $ \overline{AD} $ are $ \frac{1}{u_s}$ and $ \frac{1}{u} $, respectively, and $ \overline{AC} $  represents function $ y R_s^{-}(\alpha) $ or $ \frac{y}{u_{\texttt{\#}}} $, where $ u_{\texttt{\#}} $ is such that $ \alpha_{s}(u_{\texttt{\#}}) = \alpha $. \texttoptiebar{$ AD $} represents the scaled-pricing function $ \varphi_{\textsf{ivp}}\left(y; \alpha,u\right) $. In subfigure(b), the slope of $ \overline{AE} $ is $ \frac{1}{u} $. The slope of $ \overline{AF} $ is $ R_s^{+}(\alpha) $, and $ \overline{DE} $ is parallel to $ \overline{AF} $.  \texttoptiebar{$ BE $} represents the scaled-pricing function $ \varphi_{\textsf{ivp}}\left(y; \alpha,u\right) $. $ \overline{CE} $ is a tangent line for \texttoptiebar{$ BE $} at $ y = u $.   
		}
		\label{concave}
	\end{figure}
	
	In summary, for any $ u\in (0,1) $, when $ \alpha\geq \underline{\alpha}_1(u) $, there exists a unique solution for $\text{\textbf{BVP}}(\varphi;\alpha,u)$, we thus complete the proof of sufficiency.  Below we prove the necessity of the theorem, namely, if $ \alpha < \underline{\alpha}_1(u) $, then there is no solution for $\text{\textbf{BVP}}(\varphi;\alpha,u)$.
	
	\underline{\textbf{Case-3}: $ \alpha_s^{\min} \leq \alpha <  \alpha_s(u) $ and $ u\in (0, u_s) $}. Based on Proposition \ref{convexity_of_varphi}, we know that $ \varphi_{\textsf{ivp}}\left(y; \alpha,u\right) $ is strictly concave in $ y\in (0,u] $, and thus the monotonicity of $  \varphi'_{\textsf{ivp}}\left(y; \alpha,u\right) $ implies that
	\begin{align}
	\Delta'(\alpha,u) = \lim\limits_{y\rightarrow 0^+}\varphi'_{\textsf{ivp}}\left(y; \alpha,u\right) >\varphi'_{\textsf{ivp}}\left(u; \alpha,u\right) = \alpha\cdot \frac{1 - u^{s-1}}{(s-1)}.
	\end{align}
	Note that when $\alpha <  \alpha_s(u) $, we can show that
	\begin{align}
	P_s\left(\varphi'_{\textsf{ivp}}\left(u; \alpha,u\right),\alpha\right) > 0, 
	\end{align}
	which thus indicates that
	\begin{align}\label{derivative_larger}
	\Delta'(\alpha,u) > \varphi'_{\textsf{ivp}}\left(u; \alpha,u\right)>  R_s^+(\alpha)\geq R_s^-(\alpha).
	\end{align}
	Therefore,  it is impossible to have $\Delta(\alpha,u) =  0 $, since otherwise based on Corollary \ref{limit_finite},  $ \Delta'(\alpha,u) $ must be one of the positive roots of the characteristic polynomial, which contradicts with the inequality in Eq. \eqref{derivative_larger}. 
	Therefore,  $ \Delta(\alpha,u) \neq   0 $, and thus there exists no solution for $\text{\textbf{BVP}}(\varphi;\alpha,u)$.
	
	We illustrate \textbf{Case-3} in Fig. \ref{concave}(b). In this case, the solution to $\text{\textbf{IVP}}(\varphi;\alpha,u)$ always satisfies 
	\begin{align}\label{boundary_case_3}
	\frac{y}{u} \leq \varphi_{\textsf{ivp}}\left(y; \alpha,u\right) \leq \varphi_{\textsf{ivp}}'(u;\alpha,u)\cdot y+ 1 -  \frac{\alpha}{\alpha_s(u)},
	\end{align}
	where the lower limit follows the monotonicity  property w.r.t. $ \alpha $  and the upper limit follows Proposition \ref{limit_finite}. We illustrate the lower and upper limits by $ \overline{AE} $ and $ \overline{CE} $  in Fig. \ref{concave}(b), respectively. Since $ \varphi_{\textsf{ivp}}'(u;\alpha,u) >  R_s^+(\alpha) $, $ \overline{CE} $ must be upper bounded by $ \overline{DE} $ as well, where $ \overline{DE} $ is parallel to $ \overline{AF} $, i.e., the slopes of $ \overline{DE} $ and $ \overline{AF} $ are $ R_s^+(\alpha)  $. In summary, in this case, we have
	\begin{align}
	0<\Delta(\alpha,u) < 1 - \frac{\alpha}{\alpha_s(u)},
	\end{align}
	where the upper bound is equal to the length of $ \overline{AC} $ in Fig. \ref{concave}(b).  
		
	\underline{\textbf{Case-4}: $ \alpha  <  \alpha_s^{\min} $ and $ u\in (0,1) $}. In this case, based on Corollary \ref{no_solution_BVP1} we can directly conclude that $\Delta(\alpha,u)  > 0 $, and thus there exists no solution for $\text{\textbf{BVP}}(\varphi;\alpha,u)$.
	
	Summarizing the analysis of the above four cases,  for any $ u\in (0,1) $, $ \alpha \geq \underline{\alpha}_1(u) $ is the necessary and sufficient condition for the unique existence of solutions to $\text{\textbf{BVP}}(\varphi;\alpha,u)$. \textbf{We thus complete the proof of Theorem \ref{lower_bound_1}}.

Theorem \ref{lower_bound_LUC} is a special case of Theorem \ref{lower_bound_1} and can be proved similarly. Hence, we skip the details.

\section{Proof of Theorem \ref{lower_bound_2}}\label{proof_lower_bound_2}
The proof of this theorem is based on the analytical solution to the ODE in Eq. \eqref{two_BVP_power_2}. Here we briefly sketch the proof. The general solution to Problem \eqref{two_BVP_power_2} is as follows:
	\begin{align}\label{general_solution}
	\phi(y;\alpha,u) = e^{\alpha y}\cdot \left( \int_0^{y}\alpha e^{-\alpha \eta}  f'(\eta) d\eta + D\right) = e^{\alpha y}\cdot \left(D - \frac{\overline{c}}{\alpha^{s-1}}\cdot \Gamma(s,\alpha y)\right),
	\end{align}
	where $ D $ is a constant and $ \Gamma(s,\alpha y) $ is the incomplete Gamma function defined as follows:
	\begin{align}
	\Gamma(s,\alpha y) \triangleq \int_0^{\alpha y} \eta^{s-1}e^{-\eta}d\eta.
	\end{align}
	Substituting the first boundary condition into Eq. \eqref{general_solution} leads to the following solution
	\begin{align}\label{phi_y_alpha_u}
	\phi(y;\alpha,u) =   e^{\alpha y}\cdot \left(\frac{\overline{c}}{e^{\alpha u}} + \frac{\overline{c}}{\alpha^{s-1}}\cdot \Gamma(s,\alpha u) - \frac{\overline{c}}{\alpha^{s-1}}\cdot \Gamma(s,\alpha y)\right). 
	\end{align} 
	
	It is easy to prove that $ \phi(y;\alpha,u) $ is always a strictly-increasing function in $ y\in [u,1] $. Similar to Lemma \ref{monotonicity_alpha}, we can prove that $ \phi(y;\alpha,u) $ is increasing in $ \alpha \in (0,+\infty) $ and decreasing in $ u \in (0,1) $. The second boundary condition of  $ \phi(1;\alpha,u) \geq  \overline{p} $ indicates that
	\begin{align}\label{phi_1_p_bar}
	\phi(1;\alpha,u) =  e^{\alpha}\cdot \left(\frac{\overline{c}}{e^{\alpha u}} + \frac{\overline{c}}{\alpha^{s-1}}\cdot \Gamma(s,\alpha u) - \frac{\overline{c}}{\alpha^{s-1}}\cdot \Gamma(s,\alpha)\right)\geq \overline{p}. 
	\end{align}
	
	Based on the monotonicity of $ \phi(1;\alpha,u) $ in $ y $ and $ \alpha $, for each given $ u\in (0,1) $, solving the inequality in Eq. \eqref{phi_1_p_bar} leads to $ \alpha\geq \underline{\alpha}_2(u) $, where $ \underline{\alpha}_2(u) $ is the unique root that satisfies 
	\begin{align}\label{equation_alpha_2_u}
	\int_{u \underline{\alpha}_2(u)}^{\underline{\alpha}_2(u) } \eta^{s-1}e^{-\eta}d\eta =   \frac{\left(\underline{\alpha}_2(u)\right)^{s-1}}{\exp\left(u\underline{\alpha}_2(u)\right)} - \frac{\overline{p}\left(\underline{\alpha}_2(u)\right)^{s-1}}{\overline{c} \exp\left(\underline{\alpha}_2(u)\right)}.
	\end{align}
	Based on the monotonicity of $ \phi(y;\alpha,u) $ in  $ u\in (0,1) $, $ \underline{\alpha}_2(u) $ is strictly-increasing in $ u\in (0,1) $. 
	\textbf{We thus complete the proof of Theorem \ref{lower_bound_2}}.

\section{Proof of Proposition \ref{calculation_u_star}}\label{proof_calculation_u_star}
The proof of this proposition is trivial by following Theorem \ref{lower_bound_1} and Theorem \ref{lower_bound_2}. In particular, if we substitute $ u = u_s = \left(\frac{1}{s}\right)^{\frac{1}{s-1}} $ and $\alpha =  \underline{\alpha}_2(u_s) = \alpha_s^{\min} = s^{\frac{s}{s-1}} $ into Eq. \eqref{phi_y_alpha_u}, we have
\begin{align}\label{phi_1_alpha_s_min_u_s}
\phi(1;\alpha_s^{\min},u_s) =  e^{\alpha_s^{\min}}\cdot \left(\frac{\overline{c}}{e^{\alpha_s^{\min}u_s}} + \frac{\overline{c}}{(\alpha_s^{\min})^{s-1}}\cdot \Gamma(s,\alpha_s^{\min} u_s) - \frac{\overline{c}}{(\alpha_s^{\min})^{s-1}}\cdot \Gamma(s,\alpha)\right).
\end{align}
Simplifying the right-hand-side of Eq. \eqref{phi_1_alpha_s_min_u_s} leads to the expression of $ C_s $ in Eq. \eqref{C_s}.

Based on the monotonicity of $ \underline{\alpha}_2(u) $ in $ u\in [u_s,1) $ and Eq. \eqref{equation_alpha_2_u}, the two cases of $ \textsf{HUC}_1 $ and $ \textsf{HUC}_2 $ in Proposition \ref{calculation_u_star} directly follow.

\section{Proof of Lemma \ref{solution_to_phi_2} and Lemma \ref{rho_s}}\label{proof_of_solution_to_phi_2}
\textbf{Proof of Lemma \ref{solution_to_phi_2}}. 
Based on the calculation of $ u_{\mathsf{cdt}} $ in Proposition \ref{calculation_u_star}, when $ y = 1 $ and $ u = u_{\mathsf{cdt}} $, we have $ \phi_{\textsf{ivp}}(1;u_{\mathsf{cdt}}) = \overline{p}$. Therefore, $ \phi_{\textsf{ivp}}(y;u_{\mathsf{cdt}}) $ is a feasible solution to Problem \eqref{two_BVP_power_2}. For any $  u\in (0,u_{\mathsf{cdt}}]$, based on the monotonicity of $ \underline{\alpha}_1(u) $ and $ \underline{\alpha}_2(u) $, we have $ \underline{\alpha}_1(u) \geq \underline{\alpha}_2(u) $.  Theorem \ref{lower_bound_2} shows that when $ \alpha = \underline{\alpha}_1(u) \geq \underline{\alpha}_2(u) $, Problem \eqref{two_BVP_power_2} has a unique solution. We can prove that this unique solution must be the same as the solution to Problem \eqref{IVP}, since otherwise there would be two different solutions for the IVP in Eq.  \eqref{IVP}, leading to contradictions.  On the other hand, the monotonicity of $ \phi_{\textsf{ivp}}(y;u) $ in $ u $ implies that 
	\begin{align*}
	\phi_{\textsf{ivp}}(y;u) \geq \phi_{\textsf{ivp}}(y;u_{\mathsf{cdt}}), \forall u\in [u_s,u_{\mathsf{cdt}}],
	\end{align*} 
	which indicates that $ \phi_{\textsf{ivp}}(1;u)\geq  \phi_{\textsf{ivp}}(1;u_{\mathsf{cdt}}) = \overline{p} $. The lemma thus follows.
	
	\textbf{Proof of Lemma \ref{rho_s}}. Based on $ \phi_{\textsf{ivp}}(\rho_s;u_s) = \overline{p} $,  the value of $ \rho_s $ satisfies
	\begin{align}
	\frac{\overline{c}}{e^{\alpha_s^{\min} u_s}} - \frac{\overline{p}}{e^{\alpha_s^{\min} \rho_s}} = \frac{\overline{c}}{(\alpha_s^{\min})^{s-1}}\cdot \int_{\alpha_s^{\min} u_s}^{\alpha_s^{\min} \rho_s} \eta^{s-1}e^{-\eta}d\eta, 
	\end{align}
	which indicates that
	\begin{align}
	\int_{\alpha_s^{\min} u_s}^{\alpha_s^{\min} \rho_s} \eta^{s-1}e^{-\eta}d\eta = \frac{(\alpha_s^{\min})^{s-1}}{e^{\alpha_s^{\min} u_s}} - \frac{\overline{p} (\alpha_s^{\min})^{s-1}}{\overline{c} e^{\alpha_s^{\min} \rho_s}} .
	\end{align}
	Therefore, we have
	\begin{align}
	\int_{\alpha_s^{\min} u_s}^{\alpha_s^{\min} \rho_s} \eta^{s-1}e^{-\eta}d\eta = \frac{s^s}{\exp(\alpha_s^{\min} u_s)} - \frac{\overline{p} s^s}{\overline{c} \exp(\alpha_s^{\min} \rho_s)}.
	\end{align}
	Since $ \alpha_s^{\min} u_s = s $, Eq. \eqref{calculation_rho_s} in Lemma \ref{rho_s} follows.

\section{Proof of Theorem \ref{optimal_pricing_functions_theorem}}\label{proof_of_optimal_pricing_functions_theorem}
The optimal pricing functions in the three cases are obtained by solving the corresponding BVPs. Below we briefly sketch the proof.

\underline{\textsf{LUC}}: Based on Theorem \ref{lower_bound_LUC},  given $ v = f'^{-1}(\overline{p}) $,  and $ \alpha = \alpha_s^{\min} $, Problem \eqref{BVP_power_LUC} has a feasible solution $ \phi(y)  = \overline{c} \left(\varphi_{\textsf{luc}}(y)\right)^{s-1}$, where $ \varphi_{\textsf{luc}}(y) $ is the unique root to the following equation in variable $ \varphi_{\textsf{luc}} $:
\begin{equation}\label{unique_root_luc}
\int_{\frac{1}{v}}^{\frac{\varphi_{\textsf{luc}}}{y}}\frac{\eta^{s-1}}{P_s\left(\eta;\alpha_s^{\min}\right)} d\eta = \ln\left(\frac{v}{y}\right), y\in (0,v].
\end{equation}
In this case, $ \phi(v) = \overline{p} $. Since we just need $ \phi(v)\geq \overline{p} $,  based on Proposition \ref{convexity_of_varphi}, $ \phi_w(y) = sf'(y) $ is also a feasible solution to Problem \eqref{BVP_power_LUC}. When $ \phi_w(y) = \overline{p} $, the resource utilization level $ y $ is $ y =  f'^{-1}(\overline{p}/s) \triangleq w $. Therefore, based on the monotonicity property of $ \phi(y) $, for any $ m\in [w,v] $, we can find an optimal pricing function $\phi_m $ that is given by Eq. \eqref{opt_pricing_case_luc}. 

\underline{$ \textsf{HUC}_1 $}:  Based on Theorem \ref{lower_bound_1}, for any $ u\in [u_s,u_{\textsf{cdt}}] $ and $ \alpha = \underline{\alpha}_1(u) = \alpha_s^{\min} $, Problem \eqref{two_BVP_power_1} has a  unique solution $ \phi(y)  = \overline{c} \left(\varphi_{\textsf{huc}}(y)\right)^{s-1} $, where $ \varphi_{\textsf{huc}}(y) $ is the unique root to the following equation in variable $ \varphi_{\textsf{huc}} $:
\begin{equation}\label{unique_root_u}
\int_{\frac{1}{u}}^{\frac{\varphi_{\textsf{huc}}}{y}}\frac{\eta^{s-1}}{P_s\left(\eta;\alpha_s^{\min}\right)} d\eta = \ln\left(\frac{u}{y}\right), y\in (0,u]. 
\end{equation} 
The optimal pricing function in Eq. \eqref{opt_pricing_case_huc_1} thus follows. In particular, when $ u = u_s $, Proposition \ref{convexity_of_varphi} implies that the analytical solution to Eq. \eqref{unique_root_u} is given by $ \varphi_{\textsf{huc}} = \frac{y}{u_s} $. In this case, the optimal pricing function can be given by Eq. \eqref{phi_u_s}. 

\underline{$ \textsf{HUC}_2 $}:  Based on Proposition \ref{convexity_of_varphi} and Eq. \eqref{phi_2_u_opt}, the unique solution to Problem \eqref{two_BVP_power_2} directly follows when we have $ u = u_{\textsf{cdt}}\in (0,u_s) $ and $ \alpha = \underline{\alpha}_1(u_{\textsf{cdt}}) = \alpha_s(u_{\textsf{cdt}}) =  \frac{s-1}{u_{\textsf{cdt}} - u_{\textsf{cdt}}^s}  $,

\section{Proof of Proposition \ref{monotonicity_omega}} \label{proof_monotonicity_omega}
Below we revisit $ \textbf{IVP}(\varphi; \alpha,u)  $ for a better reference. 
\begin{align}\label{IVP_1_varphi_appendix}
\textbf{IVP}(\varphi; \alpha,u) 
\begin{cases}
\varphi'  = \alpha \frac{\varphi^{s-1} - y^{s-1}}{(s-1)\varphi^{s-1}}, 0<y<u, \\
\varphi(u) = 1,
\end{cases}
\end{align} 
According to Lemma \ref{existence_uniquess}, the above IVP has a unique solution which is denoted by $  \varphi_{\textsf{ivp}}(y;\alpha,u) $. 

We first prove that $  \varphi_{\textsf{ivp}}(y;\alpha,u) \geq  f_\varphi'(y) = y $ holds for all $ y\in (0,u] $. Note that when $ y = u\in (0,1) $, we have $  \varphi_{\textsf{ivp}}(u;\alpha,u) >  f_\varphi'(u) $ and $ \varphi_{\textsf{ivp}}'(u;\alpha,u)>0 $. Therefore, if $  \varphi_{\textsf{ivp}}(y;\alpha,u)\geq f_\varphi'(y) $ does not hold for all $ y\in (0,u] $, then there must exists at least one point within $ (0,u) $, say $ y_0\in (0,u) $, such that $ \varphi_{\textsf{ivp}}(y;\alpha,u) >  f_\varphi'(y) $ for all $ y\in (y_0,u] $, $ \varphi_{\textsf{ivp}}(y_0;\alpha,u) =   f_\varphi'(y_0) $, and $ \varphi_{\textsf{ivp}}(y;\alpha,u) < f_\varphi'(y) $ for all $ y\in (y_0-\epsilon,y_0) $, where $ \epsilon $ is a small positive value. However, when $ \varphi_{\textsf{ivp}}(y;\alpha,u) < f_\varphi'(y) $, $ \varphi_{\textsf{ivp}}'(y;\alpha,u) $ is negative according to the ODE, and thus $ \varphi_{\textsf{ivp}}(y;\alpha,u) $ is decreasing in $ (y_0-\epsilon,y_0) $. This means that $ \varphi_{\textsf{ivp}}(y;\alpha,u) > \varphi_{\textsf{ivp}}(y_0;\alpha,u) = f_\varphi'(y_0)>f_\varphi'(y) $ for all $y\in (y_0-\epsilon,y_0)  $, leading to a contradiction. Therefore, $  \varphi_{\textsf{ivp}}(y;\alpha,u) \geq  f_\varphi'(y) = y $ always holds for all $ y\in (0,u] $, and the monotonicity directly follows.

\section{Proof of Proposition \ref{monotonicity_alpha}}\label{proof_monotonicity_alpha}
The continuity direct follows since $\varphi_{\textsf{ivp}}(y;\alpha,u)$ is well defined for all  $ (\alpha, u)\in [1,+\infty)\times (0,1] $.

We first prove the monotonicity in $ u\in(0,1)$  by contradiction. Suppose we have $ u_1\in (0,1) $ and $ u_2\in (0,1) $, and assume w.l.o.g. that $ u_1>u_2 $, we can prove that  $ \varphi_{\textsf{ivp}}(y;\alpha,u_1) <  \varphi_{\textsf{ivp}}(y;\alpha,u_2) $ holds for all $ y\in (0,u_2) $. The idea is that these two functions cannot have any intersection point, since otherwise the IVP with the same ODE as  $ \textbf{IVP}(\varphi; \alpha,u)  $ but with the initial condition defined at the intersection point will have at least two solutions, namely  $ \varphi_{\textsf{ivp}}(y;\alpha,u_1) $ and $ \varphi_{\textsf{ivp}}(y;\alpha,u_2) $, which is impossible due to the uniqueness property. Note that it is also impossible for $ \varphi_{\textsf{ivp}}(y;\alpha,u_1) > \varphi_{\textsf{ivp}}(y;\alpha,u_2) $ since if this is the case, then $ \varphi_{\textsf{ivp}}(y;\alpha,u_1)  $ is not monotonic in $ y\in (0,u_1) $. Therefore, when $ u_1>u_2 $, we always have $ \varphi_{\textsf{ivp}}(y;\alpha,u_1) < \varphi_{\textsf{ivp}}(y;\alpha,u_2) $. 

We now prove the monotonicity in $ \alpha \in [1,+\infty) $. Suppose we have $ \alpha_1 $ and $ \alpha_2 $, and assume w.l.o.g. that $ \alpha_1>\alpha_2 $. We need to prove that $ \varphi_{\textsf{ivp}}(y;\alpha_1,u) <  \varphi_{\textsf{ivp}}(y;\alpha_2,u) $ for all $ y \in (0,u] $.  Based on the ODE in Eq. \eqref{IVP_1_varphi_appendix}, when $ \alpha _1>\alpha_2 $, the derivative of $ \varphi $ at $ y = u $ satisfies
\begin{align}
\varphi_{\textsf{ivp}}'(u;\alpha_1,u) > \varphi_{\textsf{ivp}}'(u;\alpha_2,u).  
\end{align}
Therefore, there must exist a small interval on the left-hand-side of $ u $, say $ [u-\sigma,u] $, where $ \sigma $ is a small positive value, such that $ \varphi_{\textsf{ivp}}(y,\alpha_1,u)< \varphi_{\textsf{ivp}}(y,\alpha_2,u)$ for all $ y \in  [u-\sigma,u] $. This can be easily proved based on the definition of derivative, which is omitted for brevity.

Now suppose $ \varphi_{\textsf{ivp}}(y;\alpha_1,u) <  \varphi_{\textsf{ivp}}(y;\alpha_2,u) $ does not hold for all $ y\in (0,u] $, then there must exist an intersection point, say $ u_0 $, such that $ \varphi(y;\alpha_1,u)< \varphi(y;\alpha_2,u)$ when $ y\in (u_0,u] $, and $ \varphi(y;\alpha_1,u)\geq  \varphi(y;\alpha_2,u)$ when $ y \in (u_0-\epsilon,u_0] $, where $ \epsilon $ is a very small positive value. Now let us consider two new IVPs with the same initial condition defined at point $ y = u_0 $, and denote the unique solutions to these two new IVPs by $ \varphi_{\mathrm{new}}(y;\alpha_1,u_0) $ and $ \varphi_{\mathrm{new}}(y;\alpha_2,u_0) $, according to the uniqueness property, we must have
\begin{align}
\varphi_{\mathrm{new}}(y;\alpha_1,u_0) = \varphi_{\textsf{ivp}}(y;\alpha_1,u),\forall y\in (0,u_0),\\
\varphi_{\mathrm{new}}(y;\alpha_1,u_0)= \varphi_{\textsf{ivp}}(y;\alpha_1,u), \forall y\in (0,u_0).
\end{align}
Since $ \varphi_{\textsf{ivp}}(y;\alpha_1,u)\geq  \varphi_{\textsf{ivp}}(y;\alpha_2,u)$ when $ y \in (u_0-\epsilon,u_0] $, which means that we cannot find a small interval on the left-hand-side of $ u_0 $, say $ [u_0-\hat{\sigma},u_0] $, such that $ \varphi_{\mathrm{new}}(y;\alpha_1,u_0)< \varphi_{\mathrm{new}}(y;\alpha_2,u_0)$. However, this contradicts with the fact that
\begin{align}
\varphi_{\mathrm{new}}'(u_0;\alpha_1,u_0) > \varphi_{\mathrm{new}}'(u_0;\alpha_2,u_0).  
\end{align}
Therefore, we have $ \varphi_{\textsf{ivp}}(y;\alpha_1,u)<  \varphi_{\textsf{ivp}}(y;\alpha_2,u)$ for all $ y\in (0,u] $. 

\section{Proof of Proposition \ref{convexity_of_varphi}}\label{proof_of_convexity}
Let us revisit the ODE of  $\textbf{IVP}(\varphi; \alpha,u) $ as follows:
\begin{align}
\varphi' =  \alpha\cdot \frac{\varphi^{s-1} - y^{s-1}}{(s-1)\varphi^{s-1}},
\end{align}
which can be written as follows:
\begin{align*}
\varphi^{s-1}-  y^{s-1} =  \frac{s-1}{\alpha} \varphi^{s-1} \varphi',
\end{align*}
Let us take derivative w.r.t. $ y $ in both sides, and after some simple manipulation, we have the following equation:
\begin{align*}
\varphi'' = \frac{-(s-1)\left(\varphi'\right)^2+\alpha\varphi'-\alpha\left(\frac{y}{\varphi}\right)^{s-2}}{\varphi}
\end{align*}

1) If $ \alpha = \alpha_s(u) $, we prove that the following equality
\begin{align}\label{quadratic_varphi}
-(s-1)\left( \varphi' \right)^2+\alpha \varphi' -\alpha\left(\frac{y}{\varphi}\right)^{s-2} = 0
\end{align}
holds for all $ y\in (0,u] $, which means $ \varphi'' = 0 $ and thus leads to the linearity of $ \varphi $.  We prove it by finding such a linear solution. Let us assume $ \varphi = Ay + B $ and  substitute it into Eq. \eqref{quadratic_varphi}, we have
\begin{equation}\label{key}
-(s-1)A^2+A \alpha_s(u)-\alpha_s(u)\cdot\left(\frac{1}{A}\left(1-\frac{B}{\varphi}\right)\right)^{s-2} = 0.
\end{equation}
To make the above equation hold for all  $ y\in (0,u] $, we let $ B = 0 $ and $ A $ be the solution to the following equation
\begin{equation}\label{key}
A^{s} - \frac{\alpha_s(u)}{s-1}A^{s-1}+\frac{\alpha_s(u)}{s-1} = P_s\big(A;\alpha_s(u)\big) = 0.
\end{equation}
Substituting $ \alpha_s(u) = \frac{s-1}{u-u^{s}} $ into the above equation leads to 
\begin{align}
A^{s} - \frac{A^{s-1}}{u-u^s}+\frac{1}{u-u^s} = 0.
\end{align}
Note that $ A = \frac{1}{u} $ is always a solution to the above equation for all $ u\in (0,1) $. 
Therefore, we have
\begin{align}
\varphi_{\textsf{ivp}}(y;\alpha_s(u),u) = \frac{y}{u}. 
\end{align}

2) If $ \alpha > \alpha_s(u) $, we prove that the following inequality 
\begin{align}\label{convex_inequality}
-(s-1)\left( \varphi'\right)^2+\alpha \varphi'-\alpha\left(\frac{y}{\varphi}\right)^{s-2} > 0
\end{align}
holds for all $ y\in (0,u] $, and thus $ \varphi''>0 $, leading to the convexity of $ \varphi $ in Proposition \ref{convexity_of_varphi}. 

In fact, according to the original ODE, we have 
\begin{align}
\varphi' = \frac{\alpha}{s-1}-\frac{\alpha}{s-1}\left(\frac{y}{\varphi}\right)^{s-1}.
\end{align}
Substituting the above equation to the left-hand-side of \eqref{convex_inequality}, we have
\begin{align*}
\ & -(s-1)\left( \varphi'\right)^2+\alpha \varphi'-\alpha\left(\frac{y}{\varphi}\right)^{s-2} \\
=\ & -(s-1)\left(\frac{\alpha}{s-1}-\frac{\alpha}{s-1}\left(\frac{y}{\varphi}\right)^{s-1}\right)^2+  \left(\frac{\alpha}{s-1}-\frac{\alpha}{s-1}\left(\frac{y}{\varphi}\right)^{s-1}\right)-\alpha\left(\frac{y}{\varphi}\right)^{s-2}\\
=\ & \frac{\alpha^2}{s-1}\left(\frac{y}{\varphi}\right)^{s-1} - \frac{\alpha^2}{s-1}\left(\frac{y}{\varphi}\right)^{2s-2} -\alpha\left(\frac{y}{\varphi}\right)^{s-2}\\
= \ & \alpha \left(\frac{y}{\varphi}\right)^{s-2}\left(\frac{\alpha}{s-1}\cdot\frac{y}{\varphi} - \frac{\alpha}{s-1}\left(\frac{y}{\varphi}\right)^{s} -1\right)\\
=\ & \alpha \left(\frac{y}{\varphi}\right)^{s-2}\cdot\frac{-1}{\left(\frac{\varphi}{y}\right)^{s}}\cdot P_s\left(\frac{\varphi}{y};\alpha\right)\\
=\ & -\alpha \left(\frac{y}{\varphi}\right)^{-2}\cdot P_s\left(\frac{\varphi}{y};\alpha\right), 
\end{align*}
where $ P_s\left(\frac{\varphi}{y};\alpha\right) $ is the characteristic polynomial with variable $ \varphi/y $. 

When $ \alpha > \alpha_s(u) $, according to the monotonicity of $ \varphi $ in $ \alpha $, we have $ \varphi < y/u $,
and thus $ \varphi/y < 1/u $. Therefore, when $ u\in (0,u_s) $, we have $ 1/u = R_s^+(\alpha_s(u)) $ and thus
\begin{align}
P_s\left(\frac{\varphi}{y};\alpha\right)  < P_s\left(\frac{\varphi}{y};\alpha_s(u)\right) <  P_s\left(R_s^+(\alpha_s(u));\alpha_s(u)\right) = 0. 
\end{align}
When $ u\in [u_s,1) $, we have
\begin{align}
P_s\left(\frac{\varphi}{y};\alpha\right)  < P_s\left(\frac{\varphi}{y};\alpha_s(u)\right) <  P_s\left(R_s^-(\alpha_s(u));\alpha_s(u)\right) = 0. 
\end{align}
Therefore, we have 
\begin{align}
-\alpha \left(\frac{y}{\varphi}\right)^{-2}\cdot P_s\left(\frac{\varphi}{y};\alpha\right) > 0,
\end{align}
which indicates that $ \varphi'' >  0 $, namely, $ \varphi $ is strictly convex.

3) If $ \alpha > \alpha_s(u) $, we can use the same approach to prove
\begin{align*}
-(s-1)\left( \varphi'\right)^2+\alpha  \varphi'-\alpha \left(\frac{y}{\varphi}\right)^{s-2} < 0,
\end{align*}
which leads to the concavity of $ \varphi $. We thus complete the proof.

\section{Proof of Proposition \ref{limit_finite}}
\label{proof_limit_finite}
The first part of this proposition directly follows Proposition \ref{convexity_of_varphi}. Specifically, based on the concavity of $ \varphi_{\textsf{ivp}}(y;\alpha,u)$, when $ 1\leq \alpha \leq  \alpha_s(u) $, we have 
\begin{align*}
\varphi_{\textsf{ivp}}(y;\alpha,u) \leq\ & \varphi_{\textsf{ivp}}'(u;\alpha,u)(y-u) + \varphi_{\textsf{ivp}}(u;\alpha,u)\\
=\ &  \frac{\alpha(1 - u^{s-1})}{(s-1)}y+ 1 -  \frac{\alpha(1 - u^{s-1})}{(s-1)}u\\
= \ &\frac{\alpha(1 - u^{s-1})}{(s-1)}y+ 1 -  \frac{\alpha}{\alpha_s(u)},
\end{align*}
which means that $ \Delta(\alpha,u) \leq 1 - \alpha/\alpha_s(u) $. In particular, when $ \alpha = \alpha_s(u) $, we have $ \Delta(\alpha,u) = 0 $.

For the second part, note that when $ u = \frac{1}{R_{\gamma}^{\texttt{+}}(\alpha)}, \frac{1}{R_{\gamma}^{\texttt{-}}(\alpha)} $, the above corollary holds naturally based on Corollary \ref{two_linear_solutions}. However, Corollary \ref{limit_finite} extends the limit to general $ u\in (0,1] $ as long as $ \Delta(\alpha,u) = 0 $.  We sketch the proof as follows. If  
$\Delta(\alpha,u) = 0 $,  then  we have 
\begin{align}\label{key}
\lim\limits_{y\rightarrow 0^+} \frac{\varphi_{\textsf{ivp}}\left(y; \alpha,u\right)}{y} 
=  \lim\limits_{y\rightarrow 0^+}\varphi'_{\textsf{ivp}}\left(y; \alpha,u\right) = \Delta'(\alpha,u)
\end{align}
Based on the  ODE of $ \textbf{IVP}(\varphi; \alpha,u)  $ in  Eq. \eqref{def_ODE}, we have
\begin{align}
\lim\limits_{y\rightarrow 0^+}\varphi'_{\textsf{ivp}}\left(y; \alpha,u\right) = \frac{\alpha}{s-1} \left(1- \lim\limits_{y\rightarrow 0^+} \left(\frac{y}{\varphi_{\textsf{ivp}}\left(y; \alpha,u\right)}\right)^{s-1}\right)
\end{align}
Therefore,  according to the power rule of limits, we have
\begin{align}
\Delta'(\alpha,u)
= \frac{\alpha}{s-1} \left(1-  \left(\frac{1}{\Delta'(u,\alpha)}\right)^{s-1}\right)
\end{align}
It is obvious that if $ \Delta'(\alpha,u) $ is not  finite and  positive, the above equation cannot hold. For finite $ \Delta'(\alpha,u) $, after a simple manipulation, the above equation indicates that $ P_s(\Delta'(\alpha,u);\alpha)= 0 $, and thus the proposition follows.

\end{document}